\documentclass[a4paper,onecolumn,10pt,accepted=2026-02-27]{quantumarticle}
\pdfoutput=1
\usepackage[utf8]{inputenc}
\usepackage[english]{babel}
\usepackage[T1]{fontenc}
\usepackage{amsmath}
\usepackage{hyperref}
\usepackage{tikz}
\usepackage{lipsum}
\usepackage{graphicx}
\usepackage{amssymb}
\usepackage{verbatim}
\usepackage{multirow}
\usetikzlibrary{quantikz2}
\usepackage{float}
\usepackage{subcaption}
\usepackage{array}
\usepackage{pgfplots}
\pgfplotsset{compat=1.18}
\usepackage{amsthm}
\usepackage{url}
\usepackage{xcolor}
\hypersetup{colorlinks=false,allbordercolors=white}
\usepackage{algorithmic}
\usepackage{algorithm}
\usepackage{authblk}
\usepackage[numbers,sort&compress]{natbib}

\begin{document}

\title{A Novel Single-Layer Quantum Neural Network for Approximate SRBB-Based Unitary Synthesis}

\author{Giacomo Belli}
\email{giacomo.belli@unipr.it}
\affiliation{Quantum Software Laboratory, Dept. of Engineering and Architecture, University of Parma, 43124 Parma, Italy}
\orcid{0009-0008-3739-9525}
\author{Marco Mordacci}
\email{marco.mordacci1@unipr.it}
\orcid{0009-0001-1955-7197}
\affiliation{Quantum Software Laboratory, Dept. of Engineering and Architecture, University of Parma, 43124 Parma, Italy}
\author{Michele Amoretti}
\email{michele.amoretti@unipr.it}
\homepage{https://www.qslab.unipr.it/}
\affiliation{Quantum Software Laboratory, Dept. of Engineering and Architecture, University of Parma, 43124 Parma, Italy}
\orcid{0000-0002-6046-1904}

\maketitle

\begin{abstract}
  In this work, a novel quantum neural network is introduced as a means to approximate any unitary evolution through the Standard Recursive Block Basis (SRBB) and is subsequently redesigned with the number of CNOTs asymptotically reduced by an exponential contribution. This algebraic approach to the problem of unitary synthesis exploits Lie algebras and their topological features to obtain scalable parameterizations of unitary operators. First, the original SRBB-based scalability scheme, already known in the literature only from a theoretical point of view, is reformulated for efficient algorithm implementation and complexity management. Remarkably, 2-qubit operators emerge as a special case of the original scaling scheme. Furthermore, an algorithm is proposed to reduce the number of CNOT gates in the scalable variational quantum circuit, thus deriving a new implementable scaling scheme that requires only one layer of approximation. The single layer CNOT-reduced quantum neural network is implemented, and its performance is assessed with a variety of different unitary matrices, both sparse and dense, up to 6 qubits via the PennyLane library. The effectiveness of the approximation is measured with different metrics in relation to two optimizers: a gradient-based method and the Nelder-Mead method. The approximate CNOT-reduced SRBB-based synthesis algorithm is also tested on real hardware and compared with other valid approximation and decomposition methods available in the literature.
\end{abstract}

\tableofcontents

\section{Introduction}
Quantum machine learning (QML) aims to leverage quantum computing to improve machine learning results in terms of performance and training efficiency~\cite{schuld2018supervised,cerezo2021variational,dunjko2020non}. Central to QML is the concept of Quantum Neural Network (QNN), a hybrid quantum-classical algorithm that can be framed as a Variational Quantum Algorithm (VQA). It typically consists of three key components: a quantum encoding circuit, a Variational Quantum Circuit (VQC)~\cite{benedetti2019parameterized}, and an optimization process.
The quantum encoding circuit serves to encode classical data into a quantum state using one of several methods, such as basis encoding, amplitude encoding, or angle encoding~\cite{schuld2018supervised}.
The VQC is characterized by a set of parameters $\theta_i$ (rotation angles) that must be trained to minimize a classical loss function.
Lastly, the optimization process consists of using classical gradient descent to update the parameters $\theta_i$.
 
Since quantum information evolves according to unitary operators, it is clear why the parametric representation of dense unitary matrices and the relative decomposition into primitive gate sets of near-term quantum devices have aroused such great interest among mathematicians, physicists, and computer scientists~\cite{shende2004minimal,shende2005synthesis,bergholm2005quantum,mottonen2005decompositions,vatan2004optimal,vartiainen2004efficient,ashhab2022numerical}. Designing quantum circuits for unitary evolutions on a multiqubit system and identifying suitable approximations for a given unitary is known as gate synthesis~\cite{kliuchnikov2015framework,ross2015algebraic,kliuchnikov2013fast,bullock2004asymptotically,jiang2020optimal,kang2023cnot}, for which some techniques, such as Recursive CS Decomposition~\cite{mottonen2004quantum} and Quantum Shannon Decomposition~\cite{krol2022efficient}, are well known. The existence of such a construction is validated by the Solovay-Kitaev algorithm~\cite{dawson2006solovay}, the first example of approximate unitary synthesis using only CNOTs and one-qubit (rotation) gates, which represent a universal model for quantum computation~\cite{barenco1995elementary}. Optimizations of this algorithm followed~\cite{pham2013optimization,zhiyenbayev2018quantum}, and then a whole new class of efficient approximate synthesis algorithms based on Number Theory with $O(\log(1/\epsilon))$ gates opened a new line of research~\cite{kliuchnikov2013asymptotically,Selinger2012EfficientCA,Ross2015OptimalAC,kliuchnikov2015practical,ross2016optimal}.

A few years ago, an approach to gate synthesis based on optimization methods garnered significant interest in approximating a matrix according to a metric criterion and a set of constraints~\cite{madden2022best,madden2022sketching}, also analyzing the resilience to various noise sources~\cite{sharma2020noise,bilek2022recursive}. Along this line of research, and similarly to unitary approximation methods based on Lie geometries~\cite{younis2020qfast,younis2021qfast}, an approach to approximate synthesis based on Lie algebras and classical optimization techniques on unitary parameterizations was recently proposed~\cite{sarkar2023scalable}. The latter introduces a Hermitian unitary basis for the $\mathbb{C}^{2^n\times2^n}$ algebra, called the \emph{Standard Recursive Block Basis} (SRBB), to obtain a scalable parameterized representation for any unitary operator via a recursive procedure (which is advantageous compared to the Pauli string basis). This type of algebraic representation is the core of an approximate gate synthesis algorithm for which a quantum neural network \emph{framework} is applicable and, even better, is able to leverage the topological features of unitary groups in performing the approximation task. Nevertheless, the approximation procedure presented in~\cite{sarkar2023scalable} is not addressed in the traditional hybrid (quantum-classical) manner, since it does not work with the VQC that represents the QNN and uses only classical optimization techniques on the parameters of the SRBB-based unitary decomposition (a hard problem for classical computers). Furthermore, the variational quantum circuit theoretically proposed in~\cite{sarkar2023scalable} requires several approximation layers for dense matrices. Since the depth grows exponentially with the number of qubits $n$, its implementation on both a simulator and real hardware would be quite inefficient from a variational perspective. Based on the algebraic properties of this new basis, the authors of~\cite{sarkar2023scalable} proposed a specific ordering of the algebraic elements in the definition of the approximating operator so as to minimize the number of CNOTs. However, the simplifications that lead to this minimization have not been incorporated into the scalable scheme of the whole quantum circuit, making the VQC effectively unusable from a practical point of view.

This work provides the following contributions. 
\begin{enumerate}

\item A \emph{revised version} of the recursive algorithm that builds the SRBB~\cite{sarkar2024quantum} for any $n$ is provided. Compared to the original one, some indices responsible for the correct ordering of the basis elements are better specified.

\item It is proved that the case $n=2$ is a special instance of the original scalability scheme proposed in the literature~\cite{sarkar2024quantum} due to its peculiar algebraic properties.

\item The application of the synthesis algorithm to the specific cases $n=3$ and $n=4$, which have never been explored in the literature before, is presented following the synopsis of a \emph{new formulation} capable of handling the increasing complexity of the algebraic basis with practical shortcuts.

\item A \emph{new algorithm} capable of reducing the total number of CNOT gates within the original scalable structure of the approximating VQC has been found.

\item The \emph{implementation} of the corresponding CNOT-reduced QNN for SRBB-based gate synthesis, which requires only one layer of approximation, is provided through the PennyLane library. The code is available in the GitHub repository \cite{srbb-syn}. Its performance is assessed for arbitrary unitary operators, both sparse and dense, up to 6 qubits in simulation, achieving novel and improved results over the state of the art with just a single layer of approximation.

\item The SRBB-based CNOT-reduced synthesis algorithm is also tested on \emph{real IBM hardware} to assess the network usability with 2 qubits.
\end{enumerate}

The remainder of the paper is organized as follows. In Section~\ref{sec:srbb}, the Standard Recursive Block Basis is introduced, providing the fundamental algebraic structure, its properties with respect to Lie groups, the recursive algorithm for creating the correct basis given $n$, and the approximate unitary synthesis formula. In particular, section~\ref{sec:recursiveSRBB} rigorously defines the SRBB and outlines in detail the steps of its implementation, which are very useful for handling and manipulating SRBB-type algebras as their elements grow exponentially. In Section~\ref{sec:2qubits}, the 2-qubit variational quantum circuit for approximate unitary synthesis is analyzed, with emphasis on why this particular case falls outside the implementable scheme proposed by the literature. In Section~\ref{sec:scaling}, the novel scaling scheme, improved with simplifications of CNOT gates, is presented and proved, accompanied by all the diagrammatic schemes necessary to understand and implement it. In Section~\ref{sec:impl}, the implementation of the single layer CNOT-reduced QNN is explained in detail, as well as the testing methodology; at the same time, test results are shown. Finally, Section~\ref{sec:conclusions} concludes the paper with a discussion of future work. In Appendix~\ref{sec:cnot_sequences}, it is shown how the transposition matrices, elements of the permutation group useful for transforming subsets of SRBB into $ZYZ$-type operators, can be implemented via CNOT gate sequences. In appendices~\ref{sec:3qubits} and~\ref{sec:4qubits}, respectively, 3- and 4-qubit systems are discussed in depth, underlining how the new formulation (with practical calculation shortcuts) helps in handling the increasing complexity of both the matrix algebra and the QNN. While Appendix~\ref{app:circuits} lists all the predefined unitaries used in the simulations from 2 to 6 qubits, Appendix~\ref{a_exampleMatrices} shows the approximation of some operators, comparing it with the ideal matrix. 

This is the extended version, enriched with proofs and results on real devices, of the conference paper titled ``\emph{A Scalable Quantum Neural Network for
Approximate Unitary Synthesis}'', accepted at the IEEE International Conference on Quantum Computing and Engineering (QCE24)~\cite{10821064}.

\section{Unitary approximation via SRBB: the algebraic structure}\label{sec:srbb}
The \emph{Standard Recursive Block Basis} (SRBB) is the Hermitian unitary basis for the $\mathbb{C}^{2^n\times2^n}$ matrix algebra, first introduced in~\cite{sarkar2023scalable} to obtain a scalable parameterized representation of unitary matrices. Accordingly, any unitary operator can be expressed as a product of exponentials of SRBB elements, thanks to the properties of the connected topological space that characterizes the Lie group $U(2^n)$. This algebraic basis is defined through a recursive method, starting from a very similar Hermitian unitary basis called \emph{Recursive Block Basis} (RBB), which can be thought of as the generalization of the Pauli basis for complex matrices of higher orders\footnote{The SRBB differs from the RBB only in a new definition of the diagonal elements of the basis, as explained in Section~\ref{sec:2qubits}.}. Below, the fundamental steps of the recursive construction are retraced to clarify some details of the mathematical context and to introduce the algebraic objects that will play a key role in the approximate unitary synthesis algorithm and in its implementation with a reduced number of CNOTs.

\subsection*{Remark 1}\label{sec:problem_phase}
The approximation algorithm~\cite{sarkar2024quantum} works only with special unitary operators without losing generality, thanks to the surjective relation that exists between the special unitary group $SU(2^n)$ and the unitary group $U(2^n)$. Indeed, $U(d)\simeq U(1)\times SU(d)$, which means that \emph{locally} (at the level of Lie algebras) any unitary matrix is a phase (or unitary) scaling of a special unitary matrix~\cite{nakahara2018geometry}. This result can also be illustrated by the well-known identity $\det(cA)=c^d\det(A)$ with $c\in\mathbb{C},\,A\in GL(d,\mathbb{C})$. In fact, for any unitary matrix $U=e^{i\alpha}V$, where $\det(V)=1$, it holds that $\det(U)=e^{i\alpha\cdot d}$, leading to the determination of the rescaling phase to the resolution of equations of degree $d$ on the complex plane. Actually, $U(d)$ has a richer \emph{global} topological structure defined by a non-trivial principal bundle $SU(d)$ on the basis $U(1)$~\cite{nakahara2018geometry}. However, for many applications of quantum computing, this geometric structure does not emerge explicitly, since unitary transformations modulo global phase are considered, hence belonging to the projective unitary group $PU(d)\cong SU(d)/\mathbb{Z}_d$.

In any case, for any $U\in U(d)$, it is worth noting that
$\frac{U}{\left[\det(U)\right]^{\frac{1}{d}}}\in SU(d)$. This fact allows us, first of all, to work only with special unitary matrices without losing generality, and, second, to have a precise recipe for correcting the arbitrariness in the choice of the phase that the QNN encounters in its learning process when considering loss functions based on density matrices\footnote{This topic will be revisited again in Section~\ref{sec:impl}, dedicated to implementation.}.

\subsection*{Remark 2}\label{remark2}
It is well-known that the set of all special unitary matrices of order $d$, denoted by $SU(d)$, forms a Lie group of dimension $d^2-1$
, and the corresponding Lie algebra is the real vector space of all anti-Hermitian traceless matrices of order
$d$, which is denoted by $su(d)$. However, in physical contexts, it is customary to choose zero-trace \emph{Hermitian} matrices as generators, adding the imaginary unit to the surjective exponential map that constitutes the parametric representation of $SU(d)$. Therefore, if SRBB (or RBB) denotes a basis of the $su(d)$ Lie algebra with elements $U_j^{(d)}$, the map $\psi$ from $\mathbb{R}^{d^2-1}$ into an open subset of $SU(d)$ containing the identity element is an analytic diffeomorphism that generates the entire special Lie group (Corollary~2.9,~\cite{kirillov2008introduction}); in symbols,
$\psi:(\theta_1,\dots,\theta_{d^2-1})\in\mathbb{R}^{d^2-1}\longrightarrow\prod_{j=1}^{d^2-1}\exp\{i\,\theta_jU_j^{(d)}\}\in SU(d)$.

\subsection{Recursive construction of the matrix algebras}\label{sec:recursiveSRBB}
To implement the recursive algorithm that builds the SRBB with elements $U_j^{(d)}$,  Corollary 2.3 of~\cite{sarkar2024quantum} is followed. It is crucial to construct the elements of the algebraic basis correctly and in the \emph{correct order}\footnote{In this regard, the index $k$ of methods B and C is better defined to fix the ambiguity of the original formulation.}, following the six different building methods that depend on the position $j$ of the matrix $B_j^{(d)}$ in the basis\footnote{The elements of the basis are indicated with $B_j^{(d)}$ until the redefinition of the diagonal elements, which characterizes the transition from RBB to SRBB.}.

Starting from the Pauli basis, which describes the matrix algebra of order $d=2^n$ with $n=1$ (quantum register with a single qubit), the matrix algebras of order $d\geqslant3$ can be derived recursively through the following methods:
\begin{description}
    \item[A)]for $j\in\{1,2,...,(d-1)^2-1\}$,
    $$
    B_{j,A}^{(d)}=
    \begin{pmatrix}
    B_j^{(d-1)}&0\\
    0&(-1)^{d-1}
    \end{pmatrix};
    $$
    \item[B)]for $j=(d-1)^2+k\!\!\mod(d-1)$, with $k\in\{d-1,1,2,...,d-2\}$,
    $$
    B_{j,B}^{(d)}=P_{(k,d-1)}\begin{pmatrix}
        D&0\\
        0&\sigma_1
    \end{pmatrix}P_{(k,d-1)}
    $$
    where $P_{(k,d-1)}$ is a 2-cycle of order $d$ and $D=\mbox{diag}\{(-1)^{l-1}:\,1\leqslant l\leqslant d-2\}$. The atypical ordering (only partially increasing) of the index $k$ is fundamental for the correct ordering of the elements in B and C methods. The latter, in fact, depends on the permutation properties of the matrices $P_{(k,d-1)}$, and the first element must be the one for which $P_{k,d-1}=\mathbb{I}_d$, i.e. $k=d-1$;
    \item[C)]for $j=(d-1)^2+(d-1)+k\!\!\mod(d-1)$, with $k\in\{d-1,1,2,...,d-2\}$,
    $$
    B_{j,C}^{(d)}=P_{(k,d-1)}\begin{pmatrix}
        D&0\\
        0&\sigma_2
    \end{pmatrix}P_{(k,d-1)};
    $$
    \item[D1)]for $j=d^2-1$ and $d$ is odd,
    $$
    B_{j,D}^{(d)}=\begin{pmatrix}
        \mathbb{I}_{\lfloor\frac{d}{2}\rfloor+1}&0\\
        0&-\mathbb{I}_{\lfloor\frac{d}{2}\rfloor}
    \end{pmatrix};
    $$
    \item[D2]for $j=d^2-1$ and $d$ is even,
    $$
    B_{j,D}^{(d)}=\begin{pmatrix}
        \Sigma&0\\
        0&\sigma_3
    \end{pmatrix}\quad\mbox{with}\quad\Sigma=\begin{pmatrix}
        \mathbb{I}_{\lfloor\frac{d}{2}\rfloor-1}&0\\
        0&-\mathbb{I}_{\lfloor\frac{d}{2}\rfloor-1}
    \end{pmatrix};
    $$
    \item[E)]for $j=d^2$, $B_{j,E}^{(d)}=\mathbb{I}_d$.
\end{description}
In Section~\ref{sec:2qubits}, an example of this construction is shown for $d=3$ and $d=4$. Furthermore, our Python implementation\footnote{A complete Python library for algebra construction and for SRBB-based gate synthesis, will be made available shortly.} is enriched with useful functions to check the mathematical properties that ensure a correct recursive construction:
\begin{enumerate}
    \item[a)]$d^2$ is the basis cardinality, where $d=2^n$ is the matrix order and $n$ is the number of qubits;
    \item[b)]for $1\leqslant j\leqslant d^2-1$, $\operatorname{tr}[B_j^{(d)}]=\left\{
        \begin{array}{l}
        1\mbox{ if $d$ odd}\\
        0\mbox{ if $d$ even}
        \end{array}\right.$;
    \item[c)]$\left[B_j^{(d)}\right]^2=\mathbb{I}_d$;
    \item[d)]$\{B_j^{(d)}:\,1\leqslant j\leqslant d^2-1\}$ forms a basis for $su(d)$ when $d$ is even;
    \item[e)]$\{B_j^{(d)}:\,j=m^2-1,\,2\leqslant m\leqslant d\}\cup\{B_{d^2}^{(d)}\}$ is the set of diagonal basis elements;
    \item[f)]$B_{d^2}^{(d)}=\mathbb{I}_d$.
\end{enumerate}
The SRBB elements are the building blocks of an approximate unitary synthesis algorithm in which their ordering, arising from their algebraic properties, and their grouping play a fundamental role. In~\cite{sarkar2024quantum}, an ordering capable of reducing the number of CNOTs is theoretically proposed; nevertheless, the approximate synthesis process occurs in a purely classical manner on the parameterized SRBB-based decomposition of the unitary operator, utilizing optimization methods like Nelder-Mead and passing the trained parameters to the circuit only at the end. In this way, multiple layers for dense matrices are required, and no implementable scalability scheme incorporating CNOT simplifications in quantum circuit design has been provided. Conversely, in this work, the algorithm is framed within a quantum context through the scalable design of the corresponding VQC, such that the optimization process is performed by the PennyLane optimizer. Furthermore, a new scalability scheme that incorporates the simplifications of the CNOT gates is identified, leading to new gate-count formulas.

Given $U\in SU(2^n)$, the unitary synthesis algorithm can approximate $U$ according to~\cite{sarkar2024quantum}:
\begin{equation}\label{eqn:U_approx}
    U_{approx}\equiv\prod_{l=1}^L\,Z(\Theta_Z^l)\Psi(\Theta_\Psi^l)\Phi(\Theta_\Phi^l)
\end{equation}
where $l$ is the layer index\footnote{The layer index $l$ represents the number of times formula~(\ref{eqn:U_approx}) is repeated in the construction of the VQC; our implementation requires only one layer ($l=1$).}. The three main factors are:
\begin{equation}\label{eqn:def_Z}
    Z(\Theta_Z^l)=\prod_{j=2}^{2^n}e^{i\left(\theta_{j^2-1}^lU_{j^2-1}^{(2^n)}\right)}
\end{equation}
\begin{equation}\label{eqn:def_Psi}
    \Psi(\Theta_\Psi^l)=\left[\prod_{j=1}^{2^{n-1}}e^{i\left(\theta_{(2j-1)^2}^lU_{(2j-1)^2}^{(2^n)}\right)}e^{i\left(\theta_{4j^2-2j}^lU_{4j^2-2j}^{(2^n)}\right)}\right]\cdot\prod_{x=1}^{2^{n-1}-1}\left(\prod T_x^e\right)M_x^e(\Theta_\Psi^l)\left(\prod T_x^e\right)
\end{equation}
\begin{equation}\label{eqn:def_Phi}
    \Phi(\Theta_\Phi^l)=\prod_{x=1}^{2^{n-1}-1}\left(\prod T_x^o\right)M_x^o(\Theta_\Phi^l)\left(\prod T_x^o\right)
\end{equation}
in which
\begin{equation}\label{eqn:def_prodT}
\prod T_x^{e/o}=\prod_{(\alpha,\beta)\in T_x^{e/o}}P_{(\alpha,\beta)}
\end{equation}
where $T_x^{e/o}$ are the sets of $2^{n-2}$ disjoint transpositions obtained from the sets of permutations
\begin{equation}\label{eqn:permutation_set}
P_{2^n}^{e/o}=\{P_{(\alpha,\beta)}\in P_{2^n}|\;1\leqslant\alpha<\beta\leqslant2^n,\;\alpha\mbox{ even},\;\beta\mbox{ even/odd}\}
\end{equation}
and
\begin{equation}\label{eqn:def_M}
    M_x^{e/o}(\Theta_{\Psi/\Phi})=\prod T_x^{e/o}\left[\prod_{(\alpha,\beta)\in T_x^{e/o}}e^{i\theta_1U_1}e^{i\theta_2U_2}e^{i\theta_3U_3}e^{i\theta_4U_4}\right]\prod T_x^{e/o}   
\end{equation}
where
$(1,2,3,4)^{o/e}=[h_\beta(\alpha-1),f_\beta(\alpha-1),h_{\beta\pm1}(\alpha),f_{\beta\pm1}(\alpha)]$,
with numerical functions defined by
\begin{equation}
\begin{cases}
f_p(q)&=(p-1)^2+(p-1)+[q\!\!\mod(p-1)]\\
h_p(q)&=(p-1)^2+[q\!\!\mod(p-1)]
\end{cases}
\end{equation}
\newtheorem{definition}{Definition}
\begin{definition}\label{def_zyz}
    A standard Euler-angle decomposition for $SU(2)$ is the $ZYZ$-decomposition, according to which any $U\in SU(2)$ admits the factorization $U=R_z(\alpha)R_y(\beta)R_z(\gamma)$ for suitable $\alpha,\beta,\gamma\in\mathbb{R}$.
\end{definition}
The Pauli matrices $\sigma_z$ and $\sigma_y$ (which, together with $\sigma_x$, generate the Lie algebra $\mathfrak{su}(2)$) suffice to span all subgroups needed for an Euler decomposition. Therefore, the exponential map $(\alpha,\beta,\gamma)\rightarrow R_z(\alpha)R_y(\beta)R_z(\gamma)$, in analogy to what is introduced in Remark~\ref{remark2}, covers $SU(2)$ and establishes the surjectivity of the parametrization. For completeness, the explicit product reproduces the general $SU(2)$ matrix $\begin{pmatrix}
    a&b\\
    -b^*&a^*
\end{pmatrix}$ with $|a|^2+|b|^2=1$. The following properties are proved~\cite{sarkar2024quantum}:
\begin{enumerate}
    \item[i)]$Z(\Theta_Z)$ is the product of exponentials of all diagonal SRBB elements except the last one (which is always the identity matrix);
    \item[ii)]in the first sub-factor of $\Psi(\Theta_\Psi)$ enclosed by square brackets, only non-diagonal SRBB elements belonging to $SU(2^n)$ are included. They admit by construction a $ZYZ$-decomposition;
    \item[iii)]in the second sub-factor of $\Psi(\Theta_\Psi)$, the factors called $M_x^e\in SU(2^n)$ also have a $ZYZ$-decomposition;
    \item[iv)]inside $\Phi(\Theta_\Phi)$, the factors called $M_x^o\in SU(2^n)$ are block-diagonal matrices with $2\times2$ unitary blocks, whose implementation derives only in part from a $ZYZ$-decomposition;
    \item[v)]factors called $\prod T_x^{e/o}$ can be implemented through CNOT sequences, which are permutation matrices.
\end{enumerate}
There is only one exception to this pattern, in particular to property iv (see Proposition~\ref{prop:n2}): for 2-qubit systems, $M_1^o$ admits a complete $ZYZ$-decomposition as well, being a $SU(2)$-block diagonal unitary matrix.

In Section~\ref{sec:2qubits}, a detailed description of these three main factors will be provided for $n=2$, from their mathematical definition to their circuit implementation. With reference to Equation~(\ref{eqn:U_approx}) with $l=1$, the quantum circuit that approximates a general $n$-qubit special unitary operator is shown in Figure~\ref{fig:qcircuit_approx}, in which the three main factors are illustrated only macroscopically in the correct logical order.
\begin{figure}[htbp]
    \centering
    \resizebox{0.6\textwidth}{!}{
        \begin{quantikz}
        \lstick{0}&&\gate[5]{\Phi(\Theta_\Phi)}&&\gate[5]{\Psi(\Theta_\Psi)}&&\gate[5]{Z(\Theta_Z)}&&\\
        \lstick{1}&&&&&&&&\\
        \lstick{2}&&&&&&&&\\
        \lstick{\vdots}&&&&&&&&\\
        \lstick{n}&&&&&&&&
        \end{quantikz}}    
    \caption{The single-layer VQC to approximate $SU(2^n)$ operators.}
    \label{fig:qcircuit_approx}
\end{figure}

\section{Quantum circuit to approximate 2-qubit systems}\label{sec:2qubits}
According to the recursive construction of the SRBB, summarized in Section~\ref{sec:recursiveSRBB}, the design of the variational quantum circuit capable of approximating any 2-qubit operator requires the matrix algebra of order $d=4$. Below, 
starting from the Pauli basis $\mathcal{B}^{(2)}$,
\begin{equation}\label{eqn:Pauli_basis}
    \sigma_1=
    \begin{pmatrix}
        0&1\\
        1&0
    \end{pmatrix}\quad
    \sigma_2=
    \begin{pmatrix}
        0&-i\\
        i&0
    \end{pmatrix}\quad
    \sigma_3=
    \begin{pmatrix}
        1&0\\
        0&-1
    \end{pmatrix}\quad
    \sigma_4=
    \begin{pmatrix}
        1&0\\
        0&1
    \end{pmatrix}
\end{equation}
the RBB $\mathcal{B}^{(4)}$ with cardinality $C=16$ is briefly reconstructed and reported in Figure~\ref{fig:rbb_d4}.
\begin{figure}[htbp]
\centering
\resizebox{1\textwidth}{!}{
\begin{minipage}{1.5\textwidth}
\begin{equation*}
\begin{aligned}
    B^{(4)}_{1,A}&=
    \begin{pmatrix}
        0&1&0&0\\
        1&0&0&0\\
        0&0&1&0\\
        0&0&0&-1
    \end{pmatrix}&
    B^{(4)}_{2,A}&=
    \begin{pmatrix}
        0&-i&0&0\\
        i&0&0&0\\
        0&0&1&0\\
        0&0&0&-1
    \end{pmatrix}&
    B^{(4)}_{3,A}&=
    \begin{pmatrix}
        1&0&0&0\\
        0&-1&0&0\\
        0&0&1&0\\
        0&0&0&-1
    \end{pmatrix}&
    B^{(4)}_{4,A}&=
    \begin{pmatrix}
        1&0&0&0\\
        0&0&1&0\\
        0&1&0&0\\
        0&0&0&-1
    \end{pmatrix}\\
    B^{(4)}_{5,A}&=
    \begin{pmatrix}
        0&0&1&0\\
        0&1&0&0\\
        1&0&0&0\\
        0&0&0&-1
    \end{pmatrix}&
    B^{(4)}_{6,A}&=
    \begin{pmatrix}
        1&0&0&0\\
        0&0&-i&0\\
        0&i&0&0\\
        0&0&0&-1
    \end{pmatrix}&
    B^{(4)}_{7,A}&=
    \begin{pmatrix}
        0&0&-i&0\\
        0&1&0&0\\
        i&0&0&0\\
        0&0&0&-1
    \end{pmatrix}&
    B^{(4)}_{8,A}&=
    \begin{pmatrix}
        1&0&0&0\\
        0&1&0&0\\
        0&0&-1&0\\
        0&0&0&-1
    \end{pmatrix}\\
    B^{(4)}_{9,B}&=
    \begin{pmatrix}
        1&0&0&0\\
        0&-1&0&0\\
        0&0&0&1\\
        0&0&1&0
    \end{pmatrix}&
    B^{(4)}_{10,B}&=
    \begin{pmatrix}
        0&0&0&1\\
        0&-1&0&0\\
        0&0&1&0\\
        1&0&0&0
    \end{pmatrix}&
    B^{(4)}_{11,B}&=
    \begin{pmatrix}
        1&0&0&0\\
        0&0&0&1\\
        0&0&-1&0\\
        0&1&0&0
    \end{pmatrix}&
    B^{(4)}_{12,C}&=
    \begin{pmatrix}
        1&0&0&0\\
        0&-1&0&0\\
        0&0&0&-i\\
        0&0&i&0
    \end{pmatrix}\\
    B^{(4)}_{13,C}&=
    \begin{pmatrix}
        0&0&0&-i\\
        0&-1&0&0\\
        0&0&1&0\\
        i&0&0&0
    \end{pmatrix}&
    B^{(4)}_{14,C}&=
    \begin{pmatrix}
        1&0&0&0\\
        0&0&0&-i\\
        0&0&-1&0\\
        0&i&0&0
    \end{pmatrix}&
    B^{(4)}_{15,D}&=
    \begin{pmatrix}
        1&0&0&0\\
        0&-1&0&0\\
        0&0&1&0\\
        0&0&0&-1
    \end{pmatrix}&
    B^{(4)}_{16,E}&=
    \begin{pmatrix}
        1&0&0&0\\
        0&1&0&0\\
        0&0&1&0\\
        0&0&0&1
    \end{pmatrix}
\end{aligned}
\end{equation*}
\end{minipage}}
\caption{Complete set of elements $B_j^{(4)}$ for the Recursive Block Basis $\mathcal{B}^{(4)}$.}
\label{fig:rbb_d4}
\end{figure}
As pointed out in Section~\ref{sec:recursiveSRBB}, the set $\{B^{(4)}_j:1\leqslant j\leqslant15\}$ is a basis for the $su(4)$ matrix algebra whose elements satisfy the properties of hermiticity, unitarity, and zero trace:
\begin{equation}
\mbox{for }1\leqslant j\leqslant15,\quad[B^{(4)}_j]^{\dagger}=B^{(4)}_j,\quad[B^{(4)}_j]^2=[B^{(4)}_j][B^{(4)}_j]^{\dagger}=\mathbb{I}_4,\quad\operatorname{tr}[B^{(4)}_j]=0
\end{equation}
The next step concerns the transition to the SRBB, and it is crucial to identify the diagonal basis elements and their position: $j\in\mathcal{J}=\{m^2-1,\;2\leqslant m\leqslant4\}\cup\{d^2\}=\{3,8,15,16\}$. The purpose is to replace them with Hermitian unitary diagonal matrices generated starting from Pauli strings, specifically Pauli strings with only $\mathbb{I}_2$ and $\sigma_3$~\cite{sarkar2024quantum}. In Table~\ref{tab:SRBB_diagonals_n2}, the replacement criterion is illustrated.
\begin{table}[htbp]
\centering
\resizebox{0.8\textwidth}{!}{
\begin{tabular}{||c|c|c|c|c|c||}
\hline
Decimal&Binary&String&Matrix&Element&Replaced\\
\hline\hline
0&00&$\mathbb{I}_2\otimes\mathbb{I}_2$&$\mathbb{I}_4$&16&no\\
1&01&$\mathbb{I}_2\otimes\sigma_3$&$\mbox{diag}(1,-1,1,-1)$&3&no\\
2&10&$\sigma_3\otimes\mathbb{I}_2$&$\mbox{diag}(1,1,-1,-1)$&8&no\\
3&11&$\sigma_3\otimes\sigma_3$&$\mbox{diag}(1,-1,-1,1)$&15&yes\\
\hline
\end{tabular}}
\caption{Replacement criterion for the new sequence of diagonal basis elements.}
\label{tab:SRBB_diagonals_n2}
\end{table}
Each diagonal element is associated with the binary representation of the numbers from 0 to $d-1$, according to the following recipe: every 0-valued bit of the binary string is replaced with $\mathbb{I}_2$, every 1-valued bit is instead replaced with $\sigma_3$, and the corresponding diagonal element is obtained through the tensor product of the two bits. Then, the old diagonal elements are compared with the new ones and replaced only if found different. The complete SRBB of order 4, identified as the set $\mathcal{U}^{(4)}=\{U^{(4)}_j:1\leqslant j\leqslant16\}$, is defined by
\begin{equation}
U^{(4)}_j=\left\{
\begin{array}{l}
\mbox{new diagonal elements if }j\in\mathcal{J}=\{3,8,15,16\}\\
B^{(4)}_j\mbox{ otherwise}
\end{array}\right.
\end{equation}
the elements of which are reported for completeness in Figure~\ref{fig:srbb_n2}.
\begin{figure}[htbp]
\centering
\resizebox{1\textwidth}{!}{
\begin{minipage}{1.5\textwidth}
\begin{equation*}
\begin{aligned}
    U^{(4)}_1&=
    \begin{pmatrix}
        0&1&0&0\\
        1&0&0&0\\
        0&0&1&0\\
        0&0&0&-1
    \end{pmatrix}&
    U^{(4)}_2&=
    \begin{pmatrix}
        0&-i&0&0\\
        i&0&0&0\\
        0&0&1&0\\
        0&0&0&-1
    \end{pmatrix}&
    U^{(4)}_3&=
    \begin{pmatrix}
        1&0&0&0\\
        0&-1&0&0\\
        0&0&1&0\\
        0&0&0&-1
    \end{pmatrix}&
    U^{(4)}_4&=
    \begin{pmatrix}
        1&0&0&0\\
        0&0&1&0\\
        0&1&0&0\\
        0&0&0&-1
    \end{pmatrix}\\
    U^{(4)}_5&=
    \begin{pmatrix}
        0&0&1&0\\
        0&1&0&0\\
        1&0&0&0\\
        0&0&0&-1
    \end{pmatrix}&
    U^{(4)}_6&=
    \begin{pmatrix}
        1&0&0&0\\
        0&0&-i&0\\
        0&i&0&0\\
        0&0&0&-1
    \end{pmatrix}&
    U^{(4)}_7&=
    \begin{pmatrix}
        0&0&-i&0\\
        0&1&0&0\\
        i&0&0&0\\
        0&0&0&-1
    \end{pmatrix}&
    U^{(4)}_8&=
    \begin{pmatrix}
        1&0&0&0\\
        0&1&0&0\\
        0&0&-1&0\\
        0&0&0&-1
    \end{pmatrix}\\
    U^{(4)}_9&=
    \begin{pmatrix}
        1&0&0&0\\
        0&-1&0&0\\
        0&0&0&1\\
        0&0&1&0
    \end{pmatrix}&
    U^{(4)}_{10}&=
    \begin{pmatrix}
        0&0&0&1\\
        0&-1&0&0\\
        0&0&1&0\\
        1&0&0&0
    \end{pmatrix}&
    U^{(4)}_{11}&=
    \begin{pmatrix}
        1&0&0&0\\
        0&0&0&1\\
        0&0&-1&0\\
        0&1&0&0
    \end{pmatrix}&
    U^{(4)}_{12}&=
    \begin{pmatrix}
        1&0&0&0\\
        0&-1&0&0\\
        0&0&0&-i\\
        0&0&i&0
    \end{pmatrix}\\
    U^{(4)}_{13}&=
    \begin{pmatrix}
        0&0&0&-i\\
        0&-1&0&0\\
        0&0&1&0\\
        i&0&0&0
    \end{pmatrix}&
    U^{(4)}_{14}&=
    \begin{pmatrix}
        1&0&0&0\\
        0&0&0&-i\\
        0&0&-1&0\\
        0&i&0&0
    \end{pmatrix}&
    U^{(4)}_{15}&=
    \begin{pmatrix}
        1&0&0&0\\
        0&-1&0&0\\
        0&0&-1&0\\
        0&0&0&1
    \end{pmatrix}&
    U^{(4)}_{16}&=
    \begin{pmatrix}
        1&0&0&0\\
        0&1&0&0\\
        0&0&1&0\\
        0&0&0&1
    \end{pmatrix}
\end{aligned}
\end{equation*}
\end{minipage}}
\caption{Complete set of elements $U_j^{(4)}$ for the Standard Recursive Block Basis $\mathcal{U}^{(4)}$.}
\label{fig:srbb_n2}
\end{figure}

\subsection{From the permutation group to ZYZ-decompositions}
Before proceeding further, some algebraic properties need to be analyzed in order to identify $n=2$ as a special case of the general scalability scheme proposed by the literature \cite{sarkar2024quantum}. Furthermore, the same analysis reveals some practical shortcuts for grouping the SRBB matrix elements.

Considering the permutation group $P_4$, the properties of the algebraic basis under permutation are encoded respectively in its two subsets\footnote{In general, it is possible to partition further, identifying the $2^{n-1}-1$ sets $T_x^{e/o}$ of $2^{n-2}$ disjoint transpositions. For $n=2$, there are no further partitions as there is only one transposition.} Definition~(\ref{eqn:def_prodT}) clarifies the general case. of even/odd transpositions (or 2-cycles), defined by $P_4^{e/o}=\{P_{(\alpha,\beta)}\in P_4|\;\alpha\mbox{ even},\beta\mbox{ even/odd}\}$, where $P_4^e=\{P_{(2,4)}\}=T_1^e=\prod T_1^e$ and $P_4^o=\{P_{(2,3)}\}=T_1^o=\prod T_1^o$. These permutation properties support Proposition~\ref{prop:n2}, making the 2-qubit system an exception to the implementable pattern. To understand how these permutations are reflected in the quantum circuit, refer to Appendix~\ref{sec:cnot_sequences}.
\newtheorem{prop}{Proposition}
\begin{prop}\label{prop:n2}
    For $n=2$, $M_x^o\in SU(2^n)$ are block-diagonal matrices with $SU(2)$ blocks and admit a $ZYZ$-type decomposition.
\end{prop}
\begin{proof}\label{proof_prop1}
Considering a general 2-cycle $P_{(\alpha,\beta)}$ in the range $1\leqslant\alpha<\beta\leqslant 2^n$, the permutation properties of SRBB elements can be split into four cases based on the parity\footnote{The function $h_\beta(\alpha)$ follows the same analysis and the corresponding permutation properties are reported for completeness in Section~\ref{sec:odd_n2}.} of indices $\alpha$ and $\beta$:
\begin{itemize}
    \item[i.] $P_{(\alpha+1,\beta)}e^{i\theta_{f_\beta(\alpha)}U_{f_\beta(\alpha)}}P_{(\alpha+1,\beta)}$ with $\alpha$ odd and $\beta$ even $\Longrightarrow P_{(2,4)}\,,\,f_4(1)=13$
    \item[ii.] $P_{(\alpha,\beta+1)}e^{i\theta_{f_\beta(\alpha)}U_{f_\beta(\alpha)}}P_{(\alpha,\beta+1)}$ with $\alpha$ even and $\beta$ odd $\Longrightarrow P_{(2,4)}\,,\,f_3(2)=6$
    \item[iii.] $P_{(\alpha+1,\beta)}e^{i\theta_{f_\beta(\alpha)}U_{f_\beta(\alpha)}}P_{(\alpha+1,\beta)}$ with $\alpha$ odd and $\beta$ odd $\Longrightarrow P_{(2,3)}\,,\,f_3(1)=7$
    \item[iv.] $P_{(\alpha,\beta-1)}e^{i\theta_{f_\beta(\alpha)}U_{f_\beta(\alpha)}}P_{(\alpha,\beta-1)}$ with $\alpha$ even and $\beta$ even $\Longrightarrow P_{(2,3)}\,,\,f_4(2)=14$
\end{itemize}
The four cases above correspond to the SRBB elements that build the $M_x^{e/o}$ factors of equation~(\ref{eqn:def_M}). Following the properties i-iv listed above, it is straightforward to show the membership of these matrices to $M_2ZYZ$ after the appropriate permutation:
{\footnotesize
\begin{equation*}
\begin{split}
P_{(2,4)}\exp\{i\theta_{13}U_{13}\}P_{(2,4)}&=
P_{(2,4)}
\begin{pmatrix}
    \cos\theta_{13}&0&0&\sin\theta_{13}\\
    0&e^{-i\theta_{13}}&0&0\\
    0&0&e^{i\theta_{13}}&0\\
    -\sin\theta_{13}&0&0&\cos\theta_{13}
\end{pmatrix}P_{(2,4)}=\\
&=
\begin{pmatrix}
    \cos\theta_{13}&\sin\theta_{13}&0&0\\
    -\sin\theta_{13}&\cos\theta_{13}&0&0\\
    0&0&e^{i\theta_{13}}&0\\
    0&0&0&e^{-i\theta_{13}}
\end{pmatrix}
\end{split}
\end{equation*}
\begin{equation*}
P_{(2,4)}\exp\{i\theta_6U_{6}\}P_{(2,4)}=
P_{(2,4)}
\begin{pmatrix}
    e^{i\theta_6}&0&0&0\\
    0&\cos\theta_6&\sin\theta_6&0\\
    0&-\sin\theta_6&\cos\theta_6&0\\
    0&0&0&e^{-i\theta_6}
\end{pmatrix}P_{(2,4)}=\\
\begin{pmatrix}
    e^{i\theta_6}&0&0&0\\
    0&e^{-i\theta_6}&0&0\\
    0&0&\cos\theta_6&-\sin\theta_6\\
    0&0&\sin\theta_6&\cos\theta_6
\end{pmatrix}
\end{equation*}
\begin{equation*}
P_{(2,3)}\exp\{i\theta_7U_{7}\}P_{(2,3)}=
P_{(2,3)}
\begin{pmatrix}
    \cos\theta_7&0&\sin\theta_7&0\\
    0&e^{i\theta_7}&0&0\\
    -\sin\theta_7&0&\cos\theta_7&0\\
    0&0&0&e^{-i\theta_7}
\end{pmatrix}P_{(2,3)}=\\
\begin{pmatrix}
\cos\theta_{7}&\sin\theta_{7}&0&0\\
    -\sin\theta_{7}&\cos\theta_{7}&0&0\\
    0&0&e^{i\theta_{7}}&0\\
    0&0&0&e^{-i\theta_{7}}
\end{pmatrix}
\end{equation*}
\begin{equation*}
\begin{split}
P_{(2,3)}\exp\{i\theta_{14}U_{14}\}P_{(2,3)}&=
P_{(2,3)}
\begin{pmatrix}
    e^{i\theta_{14}}&0&0&0\\
    0&\cos\theta_{14}&0&\sin\theta_{14}\\
    0&0&e^{-i\theta_{14}}&0\\
    0&-\sin\theta_{14}&0&\cos\theta_{14}
\end{pmatrix}P_{(2,3)}=\\
&=
\begin{pmatrix}
    e^{i\theta_{14}}&0&0&0\\
    0&e^{-i\theta_{14}}&0&0\\
    0&0&\cos\theta_{14}&\sin\theta_{14}\\
    0&0&-\sin\theta_{14}&\cos\theta_{14}
\end{pmatrix}
\end{split}
\end{equation*}}
where $M_2ZYZ$ indicates the 2-qubit quantum circuit for a $ZYZ$-decomposition~\cite{sarkar2024quantum}. The analysis is completely identical for the function $h_\beta(\alpha)$ and does not add any relevant contributions to the proof.
\end{proof}

With reference to equations~(\ref{eqn:U_approx}) with $l=1$ and Proposition~\ref{prop:n2}, Table~\ref{tab:shortcut_n2} outlines the shortcut to find out which elements come into play in each main factor and the corresponding decomposition property.
\begin{table}[htbp]
\centering
\resizebox{1\textwidth}{!}{
\begin{tabular}{||c|c|c|c|c||}
\hline
Factor&Index&Range&Elements&Property\\
\hline\hline
Z&$j^2-1$&$2\leqslant j\leqslant4$&3, 8, 15&diagonal\\
$\Psi,1^\circ$&$(2j-1)^2(4j^2-2j)$&$1\leqslant j\leqslant2$&1, 2, 9, 12&$ZYZ$-decomp.\\
$\Psi,2^\circ$&$h_\beta(\alpha-1),f_\beta(\alpha-1),h_{\beta-1}(\alpha),f_{\beta-1}(\alpha)$&$\alpha=2,\beta=4$&10, 13, 4, 6&$ZYZ$-decomp. after $P_{(2,4)}$\\
$\Phi$&$h_\beta(\alpha-1),f_\beta(\alpha-1),h_{\beta+1}(\alpha),f_{\beta+1}(\alpha)$&$\alpha=2,\beta=3$&5, 7, 11, 14&$ZYZ$-decomp. after $P_{(2,3)}$\\
\hline
\end{tabular}}
\caption{Shortcut for grouping SRBB elements into the three main factors.}
\label{tab:shortcut_n2}
\end{table}

\subsection{Diagonal contributions}\label{sec:diagonals_n2}
In this subsection, the $Z$-factor of equation~(\ref{eqn:U_approx}) responsible for diagonal contributions in the case $n=2$ will be analyzed. For one single layer, it becomes\footnote{The superscript $2^n$ indicating the matrix order is omitted for simplicity and can be recovered from the context.}: $Z(\Theta_Z)=\prod_{j=2}^{4}\exp\{i\,\theta_{j^2-1}U_{j^2-1}\}$. The elements of the algebraic basis $\mathcal{U}^{(4)}$ that come into play belong to the set $\mathcal{J}_Z=\{m^2-1,\;2\leqslant m\leqslant4\}=\{3,8,15\}=\mathcal{J}-\{16\}$, which are exactly all the diagonal elements of the basis except the last one. Accordingly, referring to Figure~\ref{fig:srbb_n2}, the $Z$-factor can be written in the following matrix form:
\begin{equation}
Z(\Theta_Z)=\mbox{diag}(e^{i(\theta_3+\theta_8+\theta_{15})},e^{i(-\theta_3+\theta_8-\theta_{15})},e^{i(\theta_3-\theta_8-\theta_{15})},e^{i(-\theta_3-\theta_8+\theta_{15})})
\end{equation}
In order to draw the corresponding circuit, Table~\ref{tab:SRBB_diagonals_n2} is useful for associating each diagonal element with the appropriate gate sequence. Of that table, only the rows corresponding to the elements involved in the $Z$-factor are relevant; therefore, all the rows apart from the first are excluded. Then, for each row, the elements of the Pauli string are reported on a line, respecting the order in which they appear in the string, from left to right. The lines are divided into columns, each representing a qubit of the quantum system under consideration, labeled by an integer from left to right that takes into account the numbering of the quantum register $\{0,1,2,\dots\}$. It is important to have the elements corresponding to the binary units on the right end, aligned with the last qubit. This step is fundamental for identifying the simplification scheme at each order $n$, as explained in Section~\ref{sec:impl}. Finally, parameters $m$ and $m'$ have to be defined: while $m$ is the largest integer after which all remaining elements of the string are $\mathbb{I}_2$ (i.e., it is the integer that corresponds to the rightmost $\sigma_3$), $m'$ indicates all $\sigma_3$ elements preceding $m$; Figure~\ref{fig:zeta_diagram_n2} shows this diagrammatic step.
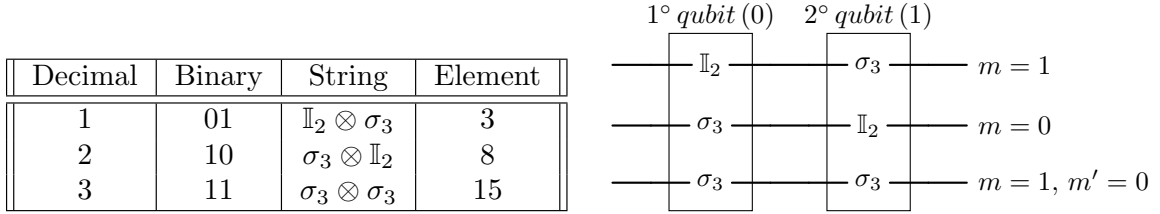
\begin{figure}[htbp]
\centering
\subfloat{
    \centering
    \resizebox{0.45\textwidth}{!}{
    \begin{tabular}{||c|c|c|c||}
    \hline
    Decimal&Binary&String&Element\\
    \hline\hline
    1&01&$\mathbb{I}_2\otimes\sigma_3$&3\\
    2&10&$\sigma_3\otimes\mathbb{I}_2$&8\\
    3&11&$\sigma_3\otimes\sigma_3$&15\\
    \hline
    \end{tabular}}}
\hfill
\subfloat{
    \centering
    \tikzset{phase label/.append style={label position=above}}
    \resizebox{0.45\textwidth}{!}{
    \begin{quantikz}
    &&\push{\;\mathbb{I}_2\;}\gategroup[3,steps=1,style={thin}]{$1^{\circ}\,qubit\,(0)$}&&&\push{\;\sigma_3\;}\gategroup[3,steps=1,style={thin}]{$2^{\circ}\,qubit\,(1)$}&&\rstick{$m=1$}\\
    &&\push{\;\sigma_3\;}&&&\push{\;\mathbb{I}_2\;}&&\rstick{$m=0$}\\
    &&\push{\;\sigma_3\;}&&&\push{\;\sigma_3\;}&&\rstick{$m=1,\,m'=0$}
    \end{quantikz}}}
\caption{Diagram to find the position of CNOTs and rotation gates.}
\label{fig:zeta_diagram_n2}
\end{figure}
Finally, the last step consists of associating each line of the previous diagram (and therefore each exponential factor considered) with one or more merged gates, following this scheme: each $m$ corresponds to the qubit on which to perform a $z$-rotation, and each $m'$ corresponds to a pair of CNOTs, before and after the rotation gate, with the target in the $m$-th qubit and the control in the $m'$-th qubit; Figure~\ref{fig:zeta_components_n2} shows the quantum circuit corresponding to each diagonal element.
\begin{figure}[htbp]
    \begin{subfigure}[b]{0.3\textwidth}
        \centering
        \resizebox{0.8\textwidth}{!}{
        \begin{quantikz}
        \lstick{0}&&&&\\
        \lstick{1}&&\gate{R_z(\theta_3)}&&
        \end{quantikz}}
        \caption{$\exp\{i\,\theta_3U_3\}$}
    \end{subfigure}
    \hfill
    \begin{subfigure}[b]{0.3\textwidth}
        \centering
        \resizebox{0.8\textwidth}{!}{
        \begin{quantikz}
        \lstick{0}&&\gate{R_z(\theta_8)}&&\\
        \lstick{1}&&&&   
        \end{quantikz}}
        \caption{$\exp\{i\,\theta_8U_8\}$}
    \end{subfigure}
    \hfill
        \begin{subfigure}[b]{0.3\textwidth}
        \centering
        \resizebox{0.9\textwidth}{!}{
        \begin{quantikz}
        \lstick{0}&\ctrl{1}&&\ctrl{1}&\\
        \lstick{1}&\targ{}&\gate{R_z(\theta_{15})}&\targ{}&
        \end{quantikz}}    
        \caption{$\exp\{i\,\theta_{15}U_{15}\}$}
    \end{subfigure}
\caption{Circuit representation of the diagonal basis $\mathcal{J}_Z$ for $n=2$.}
\label{fig:zeta_components_n2}
\end{figure}
Being diagonal matrices, the ordering of these little blocks is irrelevant, and what is shown in Figure~\ref{fig:zeta_circuit_n2} is the ordering that minimizes the number of CNOTs if inserted in the whole circuit\footnote{As mentioned previously, $n=2$ case represents an exception to the scalability scheme from an implementation perspective; therefore, the motivation behind this particular sorting lies solely in minimizing the number of CNOTs when combined with the remaining circuit, outside of any recursive scheme.}.
\begin{figure}[htbp]
    \centering
    \resizebox{0.6\textwidth}{!}{
        \begin{quantikz}
        \lstick{0}&&\ctrl{1}&&&&\ctrl{1}&&\gate{R_z(\theta_{8})}&&\\
        \lstick{1}&&\targ{}&&\gate{R_z(\theta_{15})}&&\targ{}&&\gate{R_z(\theta_{3})}&&
        \end{quantikz}}    
    \caption{The CNOT-optimized $Z$-factor circuit for $n=2$.}
    \label{fig:zeta_circuit_n2}
\end{figure}
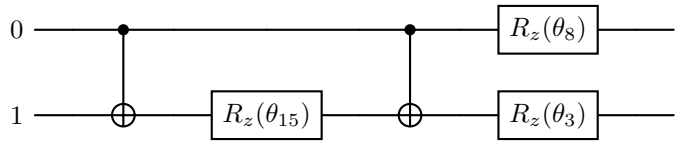

\subsection{Even contributions}
In this subsection, the $\Psi$-factor responsible for even contributions is derived from equation~(\ref{eqn:U_approx})  for one single layer: 
\begin{equation}
    \Psi(\Theta_\Psi)=\left[\prod_{j=1}^{2}e^{i\left(\theta_{(2j-1)^2}U_{(2j-1)^2}\right)}e^{i\left(\theta_{4j^2-2j}U_{4j^2-2j}\right)}\right]\cdot\left(\prod T_1^e\right)M_1^e\left(\prod T_1^e\right)
\end{equation}
It comes from exponentials of SRBB elements (first sub-factor enclosed by square brackets, hereafter called $A_\Psi$) and from permutations of exponentials of SRBB elements (second sub-factor, hereafter called $B_\Psi$). In the above equation, $\prod T_1^e=P_{(2,4)}$ and
\begin{equation}
M_1^e=P_{(2,4)}\left[e^{i\left(\theta_{h_4(1)}U_{h_4(1)}\right)}e^{i\left(\theta_{f_4(1)}U_{f_4(1)}\right)}e^{i\left(\theta_{h_3(2)}U_{h_3(2)}\right)}e^{i\left(\theta_{f_3(2)}U_{f_3(2)}\right)}\right]P_{(2,4)}    
\end{equation}
The elements of the algebraic basis $\mathcal{U}^{(4)}$ that come into play belong to the following subsets, labeled respectively to the sub-factors:
\begin{itemize}
\item[a)] $\mathcal{A}_{\Psi}=\{[(2j-1)^2,(4j^2-2j)],\;1\leqslant j\leqslant2)\}=\{(1,2),(9,12)\}$
\item[b)] $\mathcal{B}_{\Psi}=\{h_4(1),f_4(1),h_3(2),f_3(2)\}=\{10,13,4,6\}$
\end{itemize}
Accordingly, referring to Figure~\ref{fig:srbb_n2}, the first sub-factor $A_\Psi$ is a $2\times2$ block-diagonal matrix where each block belongs to $SU(2)$, as can be quickly verified:
\begin{equation}
    A_\Psi=e^{i\,\theta_1U_1}e^{i\,\theta_2U_2}e^{i\,\theta_9U_9}e^{i\,\theta_{12}U_{12}}=
    \begin{pmatrix}
        A_{\Psi}^1&0\\
        0&A_{\Psi}^2
    \end{pmatrix}
\end{equation}
{\small
\begin{equation*}
\begin{aligned}
    A_{\Psi}^1&=
    \begin{pmatrix}
        e^{i(\theta_9+\theta_{12})}(\cos\theta_1\cos\theta_2-i\sin\theta_1\sin\theta_2)&e^{-i(\theta_9+\theta_{12})}(\cos\theta_1\sin\theta_2+i\sin\theta_1\cos\theta_2)\\
        e^{i(\theta_9+\theta_{12})}(i\sin\theta_1\cos\theta_2-\cos\theta_1\sin\theta_2)&e^{-i(\theta_9+\theta_{12})}(i\sin\theta_1\sin\theta_2+\cos\theta_1\cos\theta_2)
    \end{pmatrix}\\
    A_{\Psi}^2&=
    \begin{pmatrix}
        e^{i(\theta_1+\theta_2)}(\cos\theta_9\cos\theta_{12}-i\sin\theta_9\sin\theta_{12})&e^{i(\theta_1+\theta_2)}(\cos\theta_9\sin\theta_{12}+i\sin\theta_9\cos\theta_{12})\\
        e^{-i(\theta_1+\theta_2)}(i\sin\theta_9\cos\theta_{12}-\cos\theta_9\sin\theta_{12})&e^{-i(\theta_1+\theta_2)}(i\sin\theta_9\sin\theta_{12}+\cos\theta_9\cos\theta_{12})
    \end{pmatrix}
\end{aligned}
\end{equation*}}
Therefore, the first sub-factor $A_\Psi$ is a unitary matrix belonging to $M_2ZYZ$, a fundamental building block of the final quantum circuit that can be implemented with $R_z$ and CNOT gates only~\cite{sarkar2023scalable,krol2022efficient}. Inside the second sub-factor $B_\Psi$, $M_1^e$ is again a $2\times2$ block-diagonal matrix with $SU(2)$ blocks:
\begin{equation}
    M_1^{e}=P_{(2,4)}\left[e^{i\,\theta_{10}U_{10}}e^{i\,\theta_{13}U_{13}}e^{i\,\theta_4U_4}e^{i\,\theta_6U_6}\right]P_{(2,4)}=
    \begin{pmatrix}
        B_{\Psi}^1&0\\
        0&B_{\Psi}^2
    \end{pmatrix}
\end{equation}
{\small
\begin{equation*}
\begin{aligned}
    B_{\Psi}^1&=
    \begin{pmatrix}
        e^{i(\theta_4+\theta_6)}(\cos\theta_{10}\cos\theta_{13}-i\sin\theta_{10}\sin\theta_{13})&e^{-i(\theta_4+\theta_6)}(\cos\theta_{10}\sin\theta_{13}+i\sin\theta_{10}\cos\theta_{13})\\
        e^{i(\theta_4+\theta_6)}(i\sin\theta_{10}\cos\theta_{13}-\cos\theta_{10}\sin\theta_{13})&e^{-i(\theta_4+\theta_6)}(i\sin\theta_{10}\sin\theta_{13}+\cos\theta_{10}\cos\theta_{13})
    \end{pmatrix}\\
    B_{\Psi}^2&=
    \begin{pmatrix}
        e^{i(\theta_{10}+\theta_{13})}(i\sin\theta_4\sin\theta_6+\cos\theta_4\cos\theta_6)&e^{i(\theta_{10}+\theta_{13})}(i\sin\theta_4\cos\theta_6-\cos\theta_4\sin\theta_6)\\
        e^{-i(\theta_{10}+\theta_{13})}(\cos\theta_4\sin\theta_6+i\sin\theta_4\cos\theta_6)&e^{-i(\theta_{10}+\theta_{13})}(\cos\theta_4\cos\theta_6-i\sin\theta_4\sin\theta_6)
    \end{pmatrix}
\end{aligned}
\end{equation*}}
Consequentially, the second sub-factor $B_\Psi$ is also a unitary matrix belonging to $M_2ZYZ$. The VQC for $M_2ZYZ$-type matrices is derived from a standard decomposition technique~\cite{sarkar2024quantum} and is depicted in Figure~\ref{fig:zyz_circuit_n2}; in general, for arbitrary $M_nZYZ$-type matrices, the decomposition can be easily implemented via cyclic Gray Code~\cite{mottonen2004quantum}.
\begin{figure}[htbp]
    \centering
    \resizebox{0.9\textwidth}{!}{
        \begin{quantikz}
        \lstick{0}&&\ctrl{1}&&&\ctrl{1}&&&\ctrl{1}&&\ctrl{1}&\\
        \lstick{1}&\gate{R_z(\theta_1^*)}&\targ{}&\gate{R_z(\theta_2^*)}&
        \gate{R_y(\theta_3^*)}&\targ{}&\gate{R_y(\theta_4^*)}&\gate{R_z(\theta_5^*)}&\targ{}&\gate{R_z(\theta_6^*)}&\targ{}&
        \end{quantikz}}  
    \caption{The fully decomposed $M_2ZYZ$ quantum circuit.}
    \label{fig:zyz_circuit_n2}
\end{figure}
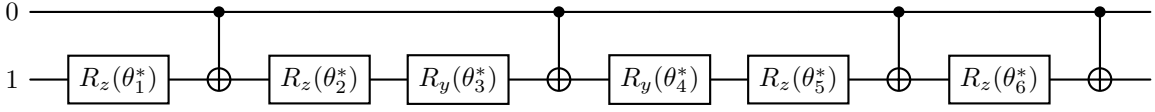
Concerning the $\Psi$-factor, only the sub-factor $\prod T_1^e$ remains to be analyzed; since the latter is represented by a permutation matrix, there are no sorting or decomposition complications from an implementation viewpoint. Using the recipe described in Appendix~\ref{sec:cnot_sequences} and summarized by Figure~\ref{fig:prodT1e_diagram_n2}, with the precaution of labeling the qubits from 0 onwards\footnote{The binary representation of $x$ is now obtained according to the formula $\sum_{i=1}^{n-1}\,2^{n-i-1}\,x_{i-1}$.} it is straightforward to deduce the corresponding quantum circuit. A rigorous proof of this recipe can be found in~\cite[Theorem 5.1]{sarkar2024quantum}.
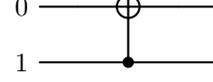
\begin{figure}[htbp]
\centering
\subfloat{
    \centering
    \begin{tabular}{||c|c|c|c||}
    \hline
    Decimal&Bits&Binary&Control-Target\\
    \hline\hline
    1&$x_0$&1&(1,0)\\
    \hline
    \end{tabular}}
\qquad
\qquad
\subfloat{
    \centering
    \begin{quantikz}
    \lstick{0}&&\targ{}&&\\
    \lstick{1}&&\ctrl{-1}&&
    \end{quantikz}}
\caption{Diagram to find the control-target pair for $\prod T_1^e$.}
\label{fig:prodT1e_diagram_n2}
\end{figure}

\subsection{Odd contributions}\label{sec:odd_n2}
This subsection is dedicated to the $\Phi$-factor of equation~(\ref{eqn:U_approx}), responsible for odd contributions. It is made up only of permutations of exponentials of SRBB elements (a fact that would make it equal to $B_\Psi$ if it were not for the permutation properties), and for one single layer becomes:
\begin{equation}
    \Phi(\Theta_\Phi)=\left(\prod T_1^o\right)M_1^o\left(\prod T_1^o\right)
\end{equation}
where $\prod T_1^o=P_{(2,3)}$ and
\begin{equation}
M_1^o=P_{(2,3)}\left[e^{i\left(\theta_{h_3(1)}U_{h_3(1)}\right)}e^{i\left(\theta_{f_3(1)}U_{f_3(1)}\right)}e^{i\left(\theta_{h_4(2)}U_{h_4(2)}\right)}e^{i\left(\theta_{f_4(2)}U_{f_4(2)}\right)}\right]P_{(2,3)}
\end{equation}
The elements of the algebraic basis $\mathcal{U}^{(4)}$ that come into play belong to the following set: $\mathcal{B}_{\Phi}=\{h_3(1),f_3(1),h_4(2),f_4(2)\}=\{5,7,11,14\}$.
Accordingly, referring to Figure~\ref{fig:srbb_n2}, $M_1^o$ is a $M_2ZYZ$-type matrix, as in the even case\footnote{Due to permutation properties proved in Appendix~\ref{proof_prop1}, also under odd-index permutations, the subset of elements building the $\Phi$-factor produces a $M_2ZYZ$-type matrix. This property holds only for 2-qubit systems.}
\begin{equation}
    M_1^o=P_{(2,3)}\left[e^{i\,\theta_5U_5}e^{i\,\theta_7U_7}e^{i\,\theta_{11}U_{11}}e^{i\,\theta_{14}U_{14}}\right]P_{(2,3)}=
    \begin{pmatrix}
        B_{\Phi}^1&0\\
        0&B_{\Phi}^2
    \end{pmatrix}
\end{equation}
{\small
\begin{equation*}
\begin{aligned}
    B_{\Phi}^1&=
    \begin{pmatrix}
        e^{i(\theta_{11}+\theta_{14})}(\cos\theta_5\cos\theta_7-i\sin\theta_5\sin\theta_7)&e^{-i(\theta_{11}+\theta_{14})}(\cos\theta_5\sin\theta_7+i\sin\theta_5\cos\theta_7)\\
        e^{i(\theta_{11}+\theta_{14})}(i\sin\theta_5\cos\theta_7-\cos\theta_5\sin\theta_7)&e^{-i(\theta_{11}+\theta_{14})}(i\sin\theta_5\sin\theta_7+\cos\theta_5\cos\theta_7)
    \end{pmatrix}\\
    B_{\Phi}^2&=
    \begin{pmatrix}
        e^{i(\theta_5+\theta_7)}(\cos\theta_{11}\cos\theta_{14}-i\sin\theta_{11}\sin\theta_{14})&e^{i(\theta_5+\theta_7)}(\cos\theta_{11}\sin\theta_{14}+i\sin\theta_{11}\cos\theta_{14})\\
        e^{-i(\theta_5+\theta_7)}(i\sin\theta_{11}\cos\theta_{14}-\cos\theta_{11}\sin\theta_{14})&e^{-i(\theta_5+\theta_7)}(i\sin\theta_{11}\sin\theta_{14}+\cos\theta_{11}\cos\theta_{14})
    \end{pmatrix}
\end{aligned}
\end{equation*}}
Therefore, the $\Phi$-factor produces another unitary matrix belonging to $M_2ZYZ$ whose circuit is essentially identical to its even counterpart (see Figure~\ref{fig:zyz_circuit_n2}). It is clear that for $n=2$, the odd contributions of the $\Phi$-factor identify a \emph{special case} of the scaling scheme presented in~\cite{sarkar2023scalable}, providing a notable simplification of the circuit that only counts $M_2ZYZ$-type matrices, diagonal elements, and CNOT-permutations. From $n=3$ onward, the $\Phi$-factor requires the use of a more general decomposition, being a block-diagonal unitary matrix with blocks that are only unitary, which still exploits $ZYZ$-type matrices but adds important multi-controlled $R_z$ gates both before and after~\cite{sarkar2023scalable}.

As for even contributions, the circuit representation of the sub-factor $\prod T_1^o$ can be derived from the recipe described in Appendix~\ref{sec:cnot_sequences}, with the precaution of labeling the qubits from 0 onward (see Figure~\ref{fig:prodT1o_diagram_n2}).
\begin{figure}[htbp]
\centering
\subfloat{
    \centering
    \resizebox{0.6\textwidth}{!}{
    \begin{tabular}{||c|c|c|c|c||}
    \hline
    Decimal&Bits&Binary&Control-Target&$k$-index\\
    \hline\hline
    1&$x_0$&1&(1,0)&0\\
    \hline
    \end{tabular}}}
\hfill
\subfloat{
    \centering
    \resizebox{0.3\textwidth}{!}{
    \begin{quantikz}
    \lstick{0}&&\ctrl{1}&\targ{}&\ctrl{1}&&\\
    \lstick{1}&&\targ{}&\ctrl{-1}&\targ{}&&
    \end{quantikz}}}
\caption{Diagram to find the control-target pair and the $k$-index for $\prod T_1^o$.}
\label{fig:prodT1o_diagram_n2}
\end{figure}
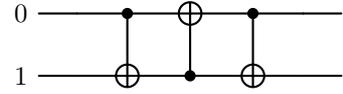

\subsection{2-qubit circuit to approximate SU(4)}
Composing the sub-circuits of the previous subsections according to equation~(\ref{eqn:U_approx}), the overall quantum circuit to approximate every $SU(4)$ operator within a QNN framework can be assembled. As shown in Section~\ref{sec:impl}, the designed QNN can be trained with only a single layer, regardless of the sparsity of the operator. This result is very interesting in comparison to the numerical simulations carried out in a complete classical context using optimization algorithms like Nelder-Mead~\cite{sarkar2024quantum}. In the following, the index layer $l$ is set to 1, and the approximating operator can be rewritten as:
\begin{equation}
\begin{split}
    &\mathcal{U}_{approx}(\Theta)=Z(\Theta_{Z}^1)\,\Psi(\Theta_\Psi^1)\,\Phi(\Theta_\Phi^1)=\\
    &=e^{i\theta_3U_3}e^{i\theta_8U_8}e^{i\theta_{15}U_{15}}e^{i\theta_1U_1}e^{i\theta_2U_2}e^{i\theta_9U_9}e^{i\theta_{12}U_{12}}\cdot\\
    &\cdot\left(\prod T_1^e\right)\left(\prod T_1^e\right)e^{i\theta_{10}U_{10}}e^{i\theta_{13}U_{13}}e^{i\theta_4U_4}e^{i\theta_6U_6}\left(\prod T_1^e\right)\left(\prod T_1^e\right)\cdot\\
    &\cdot\left(\prod T_1^o\right)\left(\prod T_1^o\right)e^{i\theta_5U_5}e^{i\theta_7U_7}e^{i\theta_{11}U_{11}}e^{i\theta_{14}U_{14}}\left(\prod T_1^o\right)\left(\prod T_1^o\right)
\end{split}
\end{equation}
Finally, in Figure~\ref{fig:n2circuit}, the complete VQC is illustrated from left to right, line by line, taking care to label its sub-components; simplified CNOTs are marked in red. Due to its algebraic properties, it is not possible to include this case in the scalable scheme of simplifications that will be the main theme of Section~\ref{sec:scaling}.
\begin{figure}[htbp]
\centering
\resizebox{1\textwidth}{!}{
\begin{subfigure}[b]{1.2\textwidth}
    \centering
    \begin{quantikz}
        \lstick{$1^\circ$q}&\ctrl{1}\gategroup[2,steps=3,style={dashed,rounded corners},label style={label position=below,anchor=north,yshift=-0.3cm}]{$\prod T_1^o$}&\targ{}&\ctrl{1}&\gategroup[2,steps=10,style={dashed,rounded corners},label style={label position=below,anchor=north,yshift=-0.3cm}]{$M_1^o$}&\ctrl{1}&&&\ctrl{1}&&&\ctrl{1}&&\ctrl[style={red}]{1}&\ctrl[style={red}]{1}\gategroup[2,steps=3,style={dashed,rounded corners},label style={label position=below,anchor=north,yshift=-0.3cm}]{$\prod T_1^o$}&\targ{}&\ctrl{1}&\rstick{...}\\
        \lstick{$2^\circ$q}&\targ{}&\ctrl{-1}&\targ{}&\gate{R_z}&\targ{}&\gate{R_z}&\gate{R_y}&\targ{}&\gate{R_y}&\gate{R_z}&\targ{}&\gate{R_z}&\targ[style={red}]{}&\targ[style={red}]{}&\ctrl{-1}&\targ{}&\rstick{...}
        \end{quantikz}
\end{subfigure}}
\resizebox{1\textwidth}{!}{
\begin{subfigure}[b]{1.2\textwidth}
    \centering
    \begin{quantikz}
        \lstick{...}&\targ{}\gategroup[2,steps=1,style={dashed,rounded corners},label style={label position=below,anchor=north,yshift=-0.3cm}]{$\prod T_1^e$}&\gategroup[2,steps=10,style={dashed,rounded corners},label style={label position=below,anchor=north,yshift=-0.3cm}]{$M_1^e$}&\ctrl{1}&&&\ctrl{1}&&&\ctrl{1}&&\ctrl{1}&\targ{}\gategroup[2,steps=1,style={dashed,rounded corners},label style={label position=below,anchor=north,yshift=-0.3cm}]{$\prod T_1^e$}&\rstick{...}\\
        \lstick{...}&\ctrl{-1}&\gate{R_z}&\targ{}&\gate{R_z}&\gate{R_y}&\targ{}&\gate{R_y}&\gate{R_z}&\targ{}&\gate{R_z}&\targ{}&\ctrl{-1}&\rstick{...}
    \end{quantikz}     
\end{subfigure}}
\resizebox{1\textwidth}{!}{
\begin{subfigure}[b]{1.2\textwidth}
    \centering
    \begin{quantikz}
        \lstick{...}&\gategroup[2,steps=10,style={dashed,rounded corners},label style={label position=below,anchor=north,yshift=-0.3cm}]{$A_\Psi$}&\ctrl{1}&&&\ctrl{1}&&&\ctrl{1}&&\ctrl[style={red}]{1}&\rstick{...}\\
        \lstick{...}&\gate{R_z}&\targ{}&\gate{R_z}&\gate{R_y}&\targ{}&\gate{R_y}&\gate{R_z}&\targ{}&\gate{R_z}&\targ[style={red}]{}&\rstick{...}
    \end{quantikz}     
\end{subfigure}}
\resizebox{1\textwidth}{!}{
\begin{subfigure}[b]{1.3\textwidth}
\centering
    \begin{quantikz}
        \lstick{...}&\ctrl[style={red}]{1}\gategroup[2,steps=4,style={dashed,rounded corners},label style={label position=below,anchor=north,yshift=-0.3cm}]{$Z$-factor}&&\ctrl{1}&\gate{R_z}&\\
        \lstick{...}&\targ[style={red}]{}&\gate{R_z}&\targ{}&\gate{R_z}&
    \end{quantikz}     
\end{subfigure}}
\caption{The VQC to approximate $SU(4)$ operators; red gates are simplified gates.}
\label{fig:n2circuit}
\end{figure}

\section{A scalable algorithm for optimizing CNOT-gates in the SRBB-based synthesis framework}\label{sec:scaling}
In the previous Section, a detailed construction of the VQC to approximate unitary evolutions with the SRBB decomposition is provided for $n=2$. Besides, the scalable design of the quantum circuits is derived step by step from the algebraic properties of SRBB elements and their peculiar grouping, following a new and simpler implementation-oriented formulation.

In this Section, a new scalable algorithm that takes into account the exponential growth of CNOT gate simplifications is proposed and analyzed in depth. As anticipated, $n=2$ represents a very particular case of the scalable scheme due to the peculiar matrix properties of its factor $M_1^o$. In order to clearly illustrate the new simplification scheme and its features, various examples will be proposed after a general discussion which aims to present the implementation details; by doing so, the specificity of the $n=2$ case will also be highlighted. 

Considering the original structure of the overall approximating quantum circuit for a generic $n$, represented by Equation~(\ref{eqn:U_approx}), only some circuit blocks allow simplifications of CNOT gates, as stated by the following proposition.
\begin{prop}
For $n\geqslant3$, the SRBB-type decomposition produces simplifications of CNOT gates inside the $Z(\Theta_Z)$-factor, representing the diagonal contributions of the $SU(2^n)$ algebraic basis, and between $\prod T_x^{e/o}$ factors, representing respectively permutation matrices with even/odd indices.
\end{prop}
\begin{proof}
    With reference to Equation~(\ref{eqn:U_approx}) and the scalable features of each quantum sub-circuit (see, for instance, the circuit representations for $n=2,3,4$ in Section \ref{sec:2qubits} and Appendices \ref{sec:3qubits} and \ref{sec:4qubits}), the factors called $M_x^{e/o}$ do not allow CNOT-simplifications either within them or with adjacent factors from $n=3$ onward. Firstly, from their complete decomposition, it is immediate to see that a rotation is always interposed between a pair of CNOTs; secondly, while $M_x^e$ factors start with a rotation and end with a CNOT$_{(0,n-1)}$, $M_x^o$ factors start and end with a couple of rotations, preventing simplifications with adjacent CNOTs of the previous and subsequent factors\footnote{The particular case $n=2$ and its simplifications will be discussed in section~\ref{sec:anomaly_n2}.}. For this reason, it can be said that $M_x^{e/o}$ factors are encapsulated inside a pair of $\prod T_x^{e/o}$ factors for each value of $x$, which, from now on, will be denoted as the \emph{edges} of the $\Psi(\Theta_\Psi)/\Phi(\Theta_\Phi)$-factors.
\end{proof}

\subsection{Simplifying CNOTs between diagonal contributions}
For an $n$-qubit register, the SRBB elements that come into play in the $Z(\Theta_Z)$-factor design~(\ref{eqn:def_Z}) belong to the set $\zeta_n=\{j^2-1,\,2\leqslant j\leqslant2^n\}$. The diagonal nature of these algebraic elements allows for the construction of the corresponding circuit according to different orderings, but the most interesting among them is the one with the largest number of simplifications between CNOT gates. Therefore, it is absolutely important to understand how these simplifications occur, and the $n=3$ case is considered as a reference example to clarify that\footnote{Being the first case for which scalability applies, it is the simplest example but also the best starting point for building any other.}. Starting from the set $\zeta_3=\{3,8,15,24,35,48,63\}$, Table~\ref{tab:mm'_n3} is constructed to order its diagonal elements and derive the pairs of parameters $(m,m')$, which are in turn useful for designing the quantum circuit for the $Z$-factor.
\begin{table}[htbp]
\centering
\resizebox{0.5\textwidth}{!}{
\begin{tabular}{|c||c|c|c|c|c|c|}
\hline
Position&$1^{\circ}\,$q&$2^{\circ}\,$q&$3^{\circ}\,$q&Element&$m$&$m'$\\
\hline
\hline
1&0&0&1&$U_3$&3&\\
2&0&1&0&$U_8$&2&\\
3&0&1&1&$U_{15}$&3&2\\
4&1&0&0&$U_{24}$&1&\\
5&1&0&1&$U_{35}$&3&1\\
6&1&1&0&$U_{48}$&2&1\\
7&1&1&1&$U_{63}$&3&1,2\\
\hline
\end{tabular}}
\caption{Pairs of parameters $(m,m')$ to build the $Z(\Theta_Z)$-factor circuit for $n=3$.}
\label{tab:mm'_n3}
\end{table}
The rows of this table are the binary representation of numbers from 1 to $2^n-1=7$, while some columns are associated with the qubits of the quantum register, labeled from 1 to $n$, from left to right. Then, for each row, the parameter  $m$ is associated with the rightmost 1 bit of the binary string, and the parameter $m'$ is associated with any other bits equal to 1. The last two columns of Table~\ref{tab:mm'_n3} report the values of the pairs $(m,m')$. Lastly, each row corresponds to a diagonal element of the circuit according to the following recipe~\cite{sarkar2024quantum}: the parameter $m$ indicates the qubit on which to perform a $z$-rotation, while the parameter $m'$ indicates the control qubit of a CNOT-pair with target on the qubit indicated by $m$; if $m'\neq0$, the same CNOT must be placed before and after the $R_z$-gate. Figure~\ref{fig_exsimpli_n3} illustrates an example of CNOT-simplification (red gates) for the $Z$-factor in the $n=3$ case, resulting from the concatenation of elements $U_{15}$ and $U_{63}$. It is clear that a CNOT-simplification occurs due to a \emph{juxtaposition} on the same control-target pair between different diagonal contributions if no evolution of information (interference) comes from the gates in between. From a diagrammatic point of view, this fact in Table~\ref{tab:mm'_n3} is reflected in the condition of having the two binary strings representing the elements taken into consideration with at least two pairs of bits with a value of 1 in the same position, and one of these must correspond to the last qubit.
\begin{figure}[!h]
\centering
\resizebox{0.5\textwidth}{!}{
\begin{quantikz}
    \lstick{$1^{\circ}$q}&\gategroup[3,steps=3,style={dashed,rounded corners},label style={label position=below,anchor=north,yshift=-0.3cm}]{$U_{15}$}&&&\ctrl{2}\gategroup[3,steps=5,style={dashed,rounded corners},label style={label position=below,anchor=north,yshift=-0.3cm}]{$U_{63}$}&&&&\ctrl{2}&\\
    \lstick{$2^{\circ}$q}&\ctrl{1}&&\ctrl[style={red}]{1}&&\ctrl[style={red}]{1}&&\ctrl{1}&&\\
    \lstick{$3^{\circ}$q}&\targ{}&\gate{R_z}&\targ[style={red}]{}&\targ{}&\targ[style={red}]{}&\gate{R_z}&\targ{}&\targ{}&
\end{quantikz}}
\caption{Example of a CNOT-simplification within the $Z(\Theta_Z)$-factor for the sub-sequence $U_{15}-U_{63}$.}
\label{fig_exsimpli_n3}
\end{figure}
\begin{prop}
    A CNOT-simplification within the $Z(\Theta_Z)$-factor occurs due to a juxtaposition of the same parametric pair $(m,m')$ between different binary strings representing the diagonal contribution of the SRBB.
\end{prop}
\begin{proof}
    The specific case for $n=3$, described above as an example to clarify the conditions of simplification, has a completely general value since it is based on:
    \begin{itemize}
    \item[i)]the relationship between the pair of parameters $(m,m')$ and the circuit representation of diagonal SRBB contributions (first defined in~\cite{sarkar2023scalable});
    \item[ii)]the properties of the binary representation valid for any $n$.
    \end{itemize}
\end{proof}
To proceed further and deduce a fundamental feature of the scaling pattern for this diagonal factor, an important observation is mandatory: thanks to the properties of the binary representation, the rows of Table~\ref{tab:mm'_n3} that end with 0 correspond to all and only the rows of the binary table for the $n=2$ case (with the obvious precaution of eliminating the rightmost qubit and renaming the diagonal elements). This implies a recursive scheme within the scalability scheme; fixed $n$ by the choice of the quantum register, only the rows of the binary table ending with 1 will provide a new contribution to the circuit, while the remaining ones exactly represent the circuit for the $Z$-factor of the $n-1$ case. This is true except for $n=2$, which is the starting point of the recursion and justifies the following proposition.
\begin{prop}
 For $n\geqslant3$, the optimal number of CNOTs in the $Z$-factor design is reached by implementing only the rows that end with a 1-valued bit of the binary table representing the SRBB diagonal contributions, in addition to the circuit of the previous $n-1$ case, if a scalable method is known to simplify pairs of CNOTs.
\end{prop}
\begin{proof}
   This proof relies entirely on the properties of binary representation, keeping in mind that the last diagonal element (the one corresponding to the binary string with only zeros) is not taken into account. In the transition to the binary table of the $n-1$ case, it is necessary to eliminate the column of the last qubit and the rows ending with 0. Thus, if you have a \emph{scalable} method to find the maximum number of simplifications between the parametric pairs $(m,m')$, exploiting the properties of the binary representation, it is possible to apply it recursively only to the rows ending in 1 of the case $n$ in question and add the $Z$-factor circuit for the $n-1$ case, which has already been simplified in the previous iteration. Since the minimum number of qubits of interest is 2, because for $n=1$ the SRBB would be replaced by the well known Pauli Basis, 2-qubit systems are excluded by this proposition as they are included in the first iteration (every recursive method requires a starting point).
\end{proof}
In the following, the scalable and easily implementable algorithm that directly provides the CNOT-simplified circuit for rows ending with 1 is described. This pattern has been investigated with different methods, some of which are based on combinatorics\footnote{A solution that is not entirely optimal can be found by placing the elements with at least one similar pair $(m,m')$ next to the element with the greatest (or greater) number of pairs $(m,m')$, once on the right and once on the left, until the items run out.}. One method is based on Gray Code: it is possible to prove that the latter encodes the optimal solution.

\subsubsection{Gray-Code-based solution}\label{sec:gray_code}
As anticipated, it turned out that the cyclic Gray Code matrices encode the optimal solution in terms of simplifications of CNOTs. Before providing the proof, the fundamental steps of the algorithm are described, taking the $n=3$ case as a reference to facilitate its generalization and, therefore, its implementation.

The starting point is to recursively construct the (cyclic) Gray Code matrices up to the case $n$ in question; their columns are labeled according to the qubits of the quantum register from left to right. Then, for the recursive scheme already explained, only the rows ending with 1 are considered, one for each diagonal element associated with a new contribution (a subset of $\zeta_3$, highlighted in green in Table~\ref{tab:graycode_n3}).
The latter are important for the following two reasons.
\begin{itemize}
    \item[i.]Reading the table from top to bottom, it is possible to recognize the optimal ordering of the involved elements that ensures the largest number of CNOT-simplifications, recovering the correspondence with algebraic elements described in Table~\ref{tab:mm'_n3}. In particular, as illustrated in Table~\ref{tab:graycode_n3}, the first line ending with 1 will always be the simplest contribution due to the Gray Code rule and thus the last little block of the sequence (see, for instance, Figure~\ref{fig:graycode_n3}). The second line ending with 1, instead, will always be the first little block of the sequence, with all the others following in order. In other words, the first element of the optimal sequence will always be the second line ending with 1, followed by all the others in the order of the \emph{cyclic} Gray Code. The reason why it is best to keep the simplest algebraic element at the end will be clear once the complete structure is obtained.
    \item[ii.]Reading the table from top to bottom, two lines at a time, it is possible to identify a rule that directly designs the CNOT-simplified circuit. In correspondence with the only bit that changes from one row to another, associate a CNOT with the control in that qubit and the target in the $n$-th qubit, followed by a $z$-rotation in the $n$-th qubit. Keeping in mind that simplified gates are marked in red in Figure~\ref{fig:graycode_n3}, the rule can be checked for the sub-sequence 15-63. Thus, every row of Table~\ref{tab:graycode_n3} corresponds to a piece of the circuit, but what matters are the jumps between two green lines. In the first jump, the only bit that changes is the second, so a CNOT$_{(2,3)}$ and a rotation in the third qubit. In the second jump, the only bit that changes is the first, so a CNOT$_{(1,3)}$ and a rotation in the third qubit, bypassing the simplified pair of CNOTs, and so on.
\end{itemize}
\begin{table}[htbp]
\centering
\resizebox{0.5\textwidth}{!}{
\begin{tabular}{|c||c|c||c|c|c|c|}
\hline
$n=1$&\multicolumn{2}{c||}{$n=2$}&\multicolumn{3}{c|}{$n=3$}&\multirow{2}*{Element}\\
\cline{1-6}
$1^{\circ}\,$q&$1^{\circ}\,$q&$2^{\circ}\,$q&$1^{\circ}\,$q&$2^{\circ}\,$q&$3^{\circ}\,$q&\\
\hline
\hline
&&&0&0&0&-\\
&&&\textcolor{green}{0}&\textcolor{green}{0}&\textcolor{green}{1}&$U_{3}$\\
&0&0&\textcolor{green}{0}&\textcolor{green}{1}&\textcolor{green}{1}&$U_{15}$\\
0&0&1&0&1&0&-\\
\hline
1&1&1&1&1&0&-\\
&1&0&\textcolor{green}{1}&\textcolor{green}{1}&\textcolor{green}{1}&$U_{63}$\\
&&&\textcolor{green}{1}&\textcolor{green}{0}&\textcolor{green}{1}&$U_{35}$\\
&&&1&0&0&-\\
\hline
\end{tabular}}
\caption{Cyclic Gray Code matrices.}
\label{tab:graycode_n3}
\end{table}
\begin{figure}[htbp]
\centering
\resizebox{0.7\textwidth}{!}{
\begin{quantikz}
    \lstick{$1^{\circ}$q}&\gategroup[3,steps=3,style={dashed,rounded corners},label style={label position=below,anchor=north,yshift=-0.3cm}]{$U_{15}$}&&&\ctrl{2}\gategroup[3,steps=5,style={dashed,rounded corners},label style={label position=below,anchor=north,yshift=-0.3cm}]{$U_{63}$}&&&&\ctrl[style={red}]{2}&\ctrl[style={red}]{2}\gategroup[3,steps=3,style={dashed,rounded corners},label style={label position=below,anchor=north,yshift=-0.3cm}]{$U_{35}$}&&\ctrl{2}&\gategroup[3,steps=1,style={dashed,rounded corners},label style={label position=below,anchor=north,yshift=-0.3cm}]{$U_{3}$}&\\
    \lstick{$2^{\circ}$q}&\ctrl{1}&&\ctrl[style={red}]{1}&&\ctrl[style={red}]{1}&&\ctrl{1}&&&&&&\\
    \lstick{$3^{\circ}$q}&\targ{}&\gate{R_z}&\targ[style={red}]{}&\targ{}&\targ[style={red}]{}&\gate{R_z}&\targ{}&\targ[style={red}]{}&\targ[style={red}]{}&\gate{R_z}&\targ{}&\gate{R_z}&
\end{quantikz}}
\caption{Simplifications (red gates) coming from new diagonal contributions only.}
\label{fig:graycode_n3}
\end{figure}
Finally, Figure~\ref{fig:Zsimplified_n3} shows the $Z$-factor circuit for the $n=3$ case resulting from all the simplifications just introduced. The circuit for the diagonal contributions is composed by 2 blocks; the first block comes from the rows of the binary table ending with 1 (new contributions of the case $n$ in question) and relies on a Gray Code pattern (scalable part), which encode the optimal ordering in terms of CNOT-count; the second block comes from the rows of the binary table ending with 0 and it is exactly the $Z$-factor circuit for the previous $n-1$ case (recursive part). This scalable and partly recursive simplification scheme guarantees the minimum number of CNOT gates at any $n$ inside the $Z$-factor of the SRBB-based approximate unitary synthesis algorithm.
\begin{figure}[htbp]
\centering
\resizebox{0.75\textwidth}{!}{
\begin{quantikz}
    \lstick{$1^{\circ}$q}&\gategroup[3,steps=8,style={dashed,rounded corners},label style={label position=below,anchor=north,yshift=-0.3cm}]{Gray-Code-scalable $Z$-factor circuit}&&\ctrl{2}&&&&\ctrl{2}&&\ctrl{1}\gategroup[2,steps=4,style={dashed,rounded corners},label style={label position=above,anchor=north,yshift=+0.3cm}]{$n=2$ $Z$-factor circuit}&&\ctrl{1}&\gate{R_z}&\\
    \lstick{$2^{\circ}$q}&\ctrl{1}&&&&\ctrl{1}&&&&\targ{}&\gate{R_z}&\targ{}&\gate{R_z}&\\
    \lstick{$3^{\circ}$q}&\targ{}&\gate{R_z}&\targ{}&\gate{R_z}&\targ{}&\gate{R_z}&\targ{}&\gate{R_z}&&&&&
\end{quantikz}}
\caption{ CNOT-optimized $Z(\Theta_Z)$-factor circuit for $n=3$.}
\label{fig:Zsimplified_n3}
\end{figure}
\begin{prop}\label{prop_Z_scalable}
    For $n\geqslant3$, the minimum number of CNOT gates within the $Z(\Theta_Z)$-factor of the SRBB-type decomposition is achieved by the combination of a scalable Gray-Code ordering of new diagonal contributions and a recursive algorithm that adds the contributions of the $n-1$ case.
\end{prop}
\begin{proof}
This proof is built on logic by contradiction; suppose that there is a better ordering of the diagonal elements than the Gray Code along the linear sequence, i.e., an ordering that implies a larger number of pairs of simplified CNOTs. It is important to note that, depending on the chosen ordering, a higher number of simplified CNOTs can come from:
\begin{itemize}
\item[i.]a greater number of edges (between different diagonal elements) that present simplifications;
\item[ii.]a greater number of simplified pairs (therefore pairs of parameters $(m,m')$ juxtaposed) in the edges already considered.
\end{itemize}
The first step is to demonstrate that Gray Code sorting already uses the largest number of edges that admit simplifications between CNOT-pairs. The number of non-recursive diagonal elements, corresponding to new contributions or binary strings ending in 1, is exactly $2^{n-1}$; but among the latter, the first element (the one with the simplest binary string) cannot generate simplifications of CNOTs\footnote{Simplifications are prohibited for this element because it is the only one with no value for the parameter $m'$, see for instance Tables~\ref{tab:graycode_n3} and \ref{tab:graycode_n4}.} and must be subtracted, remaining with $2^{n-1}-1$ elements. This is the reason why this element is always placed at the end of the sequence of non-recursive diagonal elements and, as can be seen from Figures~\ref{fig:Zsimplified_n3} and \ref{fig:Zsimplified_n4}, allows the circuit to be compacted. Therefore, the maximum number of edges that could present simplifications is $2^{n-1}-2$, i.e., the number of elements suitable for simplifications minus one.

As can be deduced from Section~\ref{sec:gray_code}, in Gray-Code-based sorting the CNOT-simplifications correspond to pairs of 1's in the same qubit (see for example the pairs of 1-valued bits for the subsequence $U_{15}-U_{63}$ in Table~\ref{tab:graycode_n3}). Considering only rows ending in 1 except for the first one, cyclic order on the Gray Code matrix will encounter pairs of rows that always share at least one 1-valued qubit in addition to the last one, and thus also a simplification. For this reason, the Gray-Code-based method uses by definition the maximum number of edges with simplifications, for which the only way to have a larger number of simplifications is to consider case ii. Therefore, the discussion shifts to counting the total number of control-target pairs and how many times each parameter $m'$ produces a pair of CNOTs to be simplified.

Since $m'$ varies from 1 to $n-1$ and the number of 1s in each column of the binary table is equal to $2^{n-1}$ (see for example Tables~\ref{tab:mm'_n3} and \ref{tab:mm'_n4}), the number of parameters $m'$ useful for non-recursive pattern in each column is equal to $2^{n-2}$ and the total number of pairs $(m,m')$ is worth $(n-1)\cdot2^{n-2}$. Due to the linear sequence, each element can be used twice for the same index $m'$, i.e., one pair with a previous element and one with a following element, except for the elements at the ends of the sequence: this implies that each parameter $m'$ can be used at most $2^{n-2}-1$ times. With reference to Tables~\ref{tab:mm'_n3} and \ref{tab:mm'_n4}, Tables~\ref{tab:count_m'_n3} and \ref{tab:count_m'_n4}, respectively, show these counts in cases $n=3,4$ to clarify.
\begin{table}[htbp]
    \centering
    \begin{subtable}[t]{0.45\textwidth}
    \centering
        \begin{tabular}{|c||c|c|c|}
        \hline
        $m'\,(n=3)$&\mbox{available}&\mbox{used}&\mbox{max}\\
        \hline
        \hline
        1&2&1&\mbox{yes}\\
        2&2&1&\mbox{yes}\\
        \hline
        \end{tabular}
    \caption{Case $n=3$.}
    \label{tab:count_m'_n3}
    \end{subtable}
    \quad
    \begin{subtable}[t]{0.45\textwidth}
    \centering
        \begin{tabular}{|c||c|c|c|}
        \hline
        $m'\,(n=4)$&\mbox{available}&\mbox{used}&\mbox{max}\\
        \hline
        \hline
        1&4&3&\mbox{yes}\\
        2&4&3&\mbox{yes}\\
        3&4&2&\mbox{no}\\
        \hline
        \end{tabular}
    \caption{Case $n=4$.}
    \label{tab:count_m'_n4}
    \end{subtable}
\caption{Counts of parameters $m'$ used to form pairs of CNOT-simplifications.}
\label{tab:count_m'}
\end{table}
From the examples in Tables~\ref{tab:count_m'_n3} and \ref{tab:count_m'_n4}, it is clear that the Gray Code method exploits each index to the maximum only for some values of $n$ and with simple counts it is possible to highlight which indices could offer further simplifications\footnote{Actually, it is simple to prove that only for $n=3$ the method exploits each index to the maximum and from $n=4$ onwards the indices not fully used increase as $n$ increases.}. At this point, having a greater number of simplified pairs means having further control-target couplings involving the parameters $m'$ which are not used to the maximum, but this is impossible as all the elements suitable for these couplings have already been used a maximum number of times. To prove this latter sentence, it is useful to count the number of elements involved in the simplification pairs for each parameter $m'$ and also the number of elements shared by simplifications with different $m'$ indexes. In this way it can be shown how the saturation of the simplification pairs for the first parameters $m'$ (the ones used to the maximum at every $n$) reduces the number of elements available for couplings with other subsequent values of $m'$.

Considering the binary tables used for sorting, each parameter $m'$ takes into account an equal number of non-recursive diagonal elements, namely $2^{n-2}$. Since some binary strings have a multi-valued parameter $m'$, the corresponding elements will show different possible pairs of simplification. In this way, in counting the elements involved in the simplification pairs, some elements will be shared by different pairs of parameters $(m,m')$. The method illustrated below counts precisely these shared elements and shows that, for each index $m'$, the number of shared elements is exactly what allows using all the diagonal elements (and having simplifications between each of them). In the binary tables, like Table~\ref{tab:mm'_n3} and \ref{tab:mm'_n4}, $n-1$ columns represent the possible values of $m'$. To count the elements shared by two different simplification pairs the following recipe can be considered:
\begin{itemize}
\item[a)]the value $m'=1$ (the first column on the left) can be ignored as it is the first index to saturate;
\item[b)]for $2\leqslant m'\leqslant n-1$, a shared element is present if in correspondence with a 1 in column $m'$, there are other 1s associated with it in the columns to its left.
\end{itemize}
Then, having ordered the columns of the binary representation from 1 to $n$ and from left to right, the first column will have $2^0=1$ group of $2^{n-1}$ 1-valued bits, the second column will have $2^1=2$ groups of $2^{n-2}$ 1-valued bits each, the third column will have $2^2=4$ groups of $2^{n-3}$ 1-valued bits each, and so on. It is then possible to enumerate the binary strings that have a 1-valued bit in the i-th column and other 1-valued bits in the columns to its left, according to the following formula:
\begin{equation}
\begin{split}
    \sum_{i=1}^{n-2}\frac{2^{n-1}}{2\cdot2^i}(2^i-1)&=2^{n-2}(n-2)-2^{n-2}\left(\sum_{i=1}^{n-2}\frac{1}{2^i}\right)=\\
    &=2^{n-2}\left[(n-2)-\sum_{i=0}^{n-2}\frac{1}{2^i}+1\right]=2^{n-2}(n-3)+1
\end{split}
\end{equation}
The saturation of $(m,m')$ pairs available for subsequent simplifications is shown by the difference between the sum of all the elements involved by each index $m'$ and the number of elements shared by the pairs with different $m'$ along the \emph{linear} sequence:
\begin{equation}
    2^{n-2}(n-1)-2^{n-2}(n-3)-1=2^{n-1}-1
\end{equation}
that is exactly the number of non-recursive diagonal elements considered in the design of the scalable part of the $Z$-factor. In conclusion, a greater number of CNOT-simplifications is not possible since they should be sought in the $(m,m')$ pairs with non-saturated $m'$ values, but there are no elements available after the saturation of the first indexes $m'$. Indeed, it would imply other elements of the sequence in addition to the totality already considered (maximum number of edges with simplifications reached) or other couplings in addition to the possible ones already considered (once the first indices are saturated, there are no longer couplings available for not fully-used $m'$ indexes with the elements left in play). While the first scenario contradicts the initial assumption regarding the number of diagonal elements, the second one contradicts the linear sequence assumption.
\end{proof}

The scalable algorithm for the $Z$-factor circuit stated in Proposition~\ref{prop_Z_scalable}, optimized in terms of CNOT gates through a Gray-Code-ordering, is illustrated via pseudocode in Algorithm~\ref{alg1}.

\begin{algorithm}
\caption{\textbf{[$Z$-FACTOR]} Scalable CNOT-optimized $Z(\Theta_Z)$-factor circuit}
\label{alg1}
\begin{algorithmic}[1]
\STATE\textbf{Provided:} $CNOT(c,t)$ where $c$ stands for control and $t$ for target; $R_z(\theta_a)$; an algorithm that generates the binary matrices where the 1s indicate the only bit that changes line after line in the Gray Code recursion, called \emph{CNOT-remaining-Gray} (see Table~\ref{tab:graycode_n3} and Figure~\ref{fig:graycode_n3}); the usual notation where the first qubit is labeled with 0.
\STATE\textbf{Input:} number of qubits $n>2$; set of parameters $\theta_a\in\Theta_Z$; output of \emph{CNOT-remaining-Gray}.
\STATE\textbf{Output:} circuit for the $Z(\Theta_Z)$-factor, optimized in terms of CNOT gates.
\FORALL{$k$ such that $0\leq k<n-1$}
\STATE Consider the output of \emph{CNOT-remaining-Gray} in the reversed order, i.e., the binary matrix $n-2-k$
\FOR{row $i$ in matrix $n-2-k$}
\FOR{column $j$ in row $i$}
\IF{element $(i,j)$ equals 1}
\STATE $CNOT(j,n-1-k)$
\STATE $R_z(\theta_{a_i})$ on qubit $n-1-k$
\ELSE
\STATE Pass
\ENDIF
\ENDFOR
\ENDFOR
\ENDFOR
\STATE $R_z(\theta_{a_{2^n-1}})$ on qubit 0
\end{algorithmic}
\end{algorithm}

\subsubsection{Solution for n=4}
As for 3-qubit systems, the starting point is to identify the SRBB diagonal elements which participate in the construction of the $Z$-factor. The latter belong to the set $$
\zeta_4=\{3,8,15,24,35,48,63,80,99,120,143,168,195,224,255\}
$$
and Table~\ref{tab:mm'_n4} sorts them according to the binary representation of their position in the set.
\begin{table}[htbp]
\centering
\resizebox{0.5\textwidth}{!}{
\begin{tabular}{|c||c|c|c|c|c|c|c|}
\hline
Position&$1^{\circ}\,$q&$2^{\circ}\,$q&$3^{\circ}\,$q&$4^{\circ}\,$q&Element&$m$&$m'$\\
\hline
\hline
1&0&0&0&1&$U_3$&4&\\
2&0&0&1&0&$U_8$&3&\\
3&0&0&1&1&$U_{15}$&4&3\\
4&0&1&0&0&$U_{24}$&2&\\
5&0&1&0&1&$U_{35}$&4&2\\
6&0&1&1&0&$U_{48}$&3&2\\
7&0&1&1&1&$U_{63}$&4&2,3\\
8&1&0&0&0&$U_{80}$&1&\\
9&1&0&0&1&$U_{99}$&4&1\\
10&1&0&1&0&$U_{120}$&3&1\\
11&1&0&1&1&$U_{143}$&4&1,3\\
12&1&1&0&0&$U_{168}$&2&1\\
13&1&1&0&1&$U_{195}$&4&1,2\\
14&1&1&1&0&$U_{224}$&3&1,2\\
15&1&1&1&1&$U_{255}$&4&1,2,3\\
\hline
\end{tabular}}
\caption{Pairs of parameters $(m,m')$ to build the $Z(\Theta_Z)$-factor circuit for $n=4$.}
\label{tab:mm'_n4}
\end{table}
Then, the cyclic Gray Code matrix for $n=4$ is constructed, and its rows ending with 1 are highlighted in green, as illustrated in Table~\ref{tab:graycode_n4}.
\begin{table}[htbp]
\centering
\resizebox{0.6\textwidth}{!}{
\begin{tabular}{|c||c|c||c|c|c||c|c|c|c|c|}
\hline
$n=1$&\multicolumn{2}{c||}{$n=2$}&\multicolumn{3}{c||}{$n=3$}&\multicolumn{4}{c|}{$n=4$}&\multirow{2}*{Element}\\
\cline{1-10}
$1^{\circ}\,$q&$1^{\circ}\,$q&$2^{\circ}\,$q&$1^{\circ}\,$q&$2^{\circ}\,$q&$3^{\circ}\,$q&$1^{\circ}\,$q&$2^{\circ}\,$q&$3^{\circ}\,$q&$4^{\circ}\,$q&\\
\hline
\hline
&&&&&&0&0&0&0&-\\
&&&&&&\textcolor{green}{0}&\textcolor{green}{0}&\textcolor{green}{0}&\textcolor{green}{1}&$U_{3}$\\
&&&&&&\textcolor{green}{0}&\textcolor{green}{0}&\textcolor{green}{1}&\textcolor{green}{1}&$U_{15}$\\
&&&&&&0&0&1&0&-\\
&&&0&0&0&0&1&1&0&-\\
&&&0&0&1&\textcolor{green}{0}&\textcolor{green}{1}&\textcolor{green}{1}&\textcolor{green}{1}&$U_{63}$\\
&0&0&0&1&1&\textcolor{green}{0}&\textcolor{green}{1}&\textcolor{green}{0}&\textcolor{green}{1}&$U_{35}$\\
0&0&1&0&1&0&0&1&0&0&-\\
\hline
1&1&1&1&1&0&1&1&0&0&-\\
&1&0&1&1&1&\textcolor{green}{1}&\textcolor{green}{1}&\textcolor{green}{0}&\textcolor{green}{1}&$U_{195}$\\
&&&1&0&1&\textcolor{green}{1}&\textcolor{green}{1}&\textcolor{green}{1}&\textcolor{green}{1}&$U_{255}$\\
&&&1&0&0&1&1&1&0&-\\
&&&&&&1&0&1&0&-\\
&&&&&&\textcolor{green}{1}&\textcolor{green}{0}&\textcolor{green}{1}&\textcolor{green}{1}&$U_{143}$\\
&&&&&&\textcolor{green}{1}&\textcolor{green}{0}&\textcolor{green}{0}&\textcolor{green}{1}&$U_{99}$\\
&&&&&&1&0&0&0&-\\
\hline
\end{tabular}}
\caption{Cyclic Gray Code matrices.}
\label{tab:graycode_n4}
\end{table}
Finally, reading the table from top to bottom, two lines at a time, it is possible to directly draw the circuit simplified in terms of CNOTs. In Figure~\ref{fig:graycode_n4}, the optimal sequence of diagonal elements is illustrated (simplified gates are marked in red for clarity) and the resulting circuit for the $Z$-factor of the $n=4$ case is shown in Figure~\ref{fig:Zsimplified_n4}.
\begin{figure}[htbp]
\centering
\resizebox{0.9\textwidth}{!}{
\begin{subfigure}[b]{1.3\textwidth}
    \centering
    \begin{quantikz}
        \lstick{$1^{\circ}$q}&\gategroup[4,steps=3,style={dashed,rounded corners},label style={label position=below,anchor=north,yshift=-0.3cm}]{$U_{15}$}&&&\gategroup[4,steps=5,style={dashed,rounded corners},label style={label position=below,anchor=north,yshift=-0.3cm}]{$U_{63}$}&&&&&\gategroup[4,steps=3,style={dashed,rounded corners},label style={label position=below,anchor=north,yshift=-0.3cm}]{$U_{35}$}&&&\ctrl{3}\gategroup[4,steps=5,style={dashed,rounded corners},label style={label position=below,anchor=north,yshift=-0.3cm}]{$U_{195}$}&&&&\ctrl[style={red}]{3}&\rstick{...}\\
        \lstick{$2^{\circ}$q}&&&&\ctrl{2}&&&&\ctrl[style={red}]{2}&\ctrl[style={red}]{2}&&\ctrl[style={red}]{2}&&\ctrl[style={red}]{2}&&\ctrl[style={red}]{2}&&\rstick{...}\\
        \lstick{$3^{\circ}$q}&\ctrl{1}&&\ctrl[style={red}]{1}&&\ctrl[style={red}]{1}&&\ctrl{1}&&&&&&&&&&\rstick{...}\\
        \lstick{$4^{\circ}$q}&\targ{}&\gate{R_z}&\targ[style={red}]{}&\targ{}&\targ[style={red}]{}&\gate{R_z}&\targ{}&\targ[style={red}]{}&\targ[style={red}]{}&\gate{R_z}&\targ[style={red}]{}&\targ{}&\targ[style={red}]{}&\gate{R_z}&\targ[style={red}]{}&\targ[style={red}]{}&\rstick{...}
    \end{quantikz}
\end{subfigure}}
\resizebox{0.9\textwidth}{!}{
\begin{subfigure}[b]{1.3\textwidth}
    \centering
    \begin{quantikz}
        \lstick{...}&\ctrl[style={red}]{3}\gategroup[4,steps=7,style={dashed,rounded corners},label style={label position=below,anchor=north,yshift=-0.3cm}]{$U_{255}$}&&&&&&\ctrl[style={red}]{3}&\ctrl[style={red}]{3}\gategroup[4,steps=5,style={dashed,rounded corners},label style={label position=below,anchor=north,yshift=-0.3cm}]{$U_{143}$}&&&&\ctrl[style={red}]{3}&\ctrl[style={red}]{3}\gategroup[4,steps=3,style={dashed,rounded corners},label style={label position=below,anchor=north,yshift=-0.3cm}]{$U_{99}$}&&\ctrl{3}&\gategroup[4,steps=1,style={dashed,rounded corners},label style={label position=below,anchor=north,yshift=-0.3cm}]{$U_{3}$}&\\
        \lstick{...}&&\ctrl[style={red}]{2}&&&&\ctrl{2}&&&&&&&&&&&\\
        \lstick{...}&&&\ctrl{1}&&\ctrl[style={red}]{1}&&&&\ctrl[style={red}]{1}&&\ctrl{1}&&&&&&\\
        \lstick{...}&\targ[style={red}]{}&\targ[style={red}]{}&\targ{}&\gate{R_z}&\targ[style={red}]{}&\targ{}&\targ[style={red}]{}&\targ[style={red}]{}&\targ[style={red}]{}&\gate{R_z}&\targ{}&\targ[style={red}]{}&\targ[style={red}]{}&\gate{R_z}&\targ{}&\gate{R_z}&
    \end{quantikz}
\end{subfigure}}
\caption{Simplifications (red gates) coming from new diagonal contributions only.}
\label{fig:graycode_n4}
\end{figure}

\begin{figure}[htbp]
\centering
\resizebox{1\textwidth}{!}{
\begin{subfigure}[b]{1.5\textwidth}
    \centering
    \begin{quantikz}
        \lstick{$1^{\circ}$q}&\gategroup[4,steps=16,style={dashed,rounded corners},label style={label position=above,anchor=north,yshift=+0.3cm}]{Gray-Code-scalable $Z$-factor circuit}&&&&&&\ctrl{3}&&&&&&&&\ctrl{3}&&\rstick{...}\\
        \lstick{$2^{\circ}$q}&&&\ctrl{2}&&&&&&&&\ctrl{2}&&&&&&\rstick{...}\\
        \lstick{$3^{\circ}$q}&\ctrl{1}&&&&\ctrl{1}&&&&\ctrl{1}&&&&\ctrl{1}&&&&\rstick{...}\\
        \lstick{$4^{\circ}$q}&\targ{}&\gate{R_z}&\targ{}&\gate{R_z}&\targ{}&\gate{R_z}&\targ{}&\gate{R_z}&\targ{}&\gate{R_z}&\targ{}&\gate{R_z}&\targ{}&\gate{R_z}&\targ{}&\gate{R_z}&\rstick{...}
    \end{quantikz}
\end{subfigure}}
\resizebox{1\textwidth}{!}{
\begin{subfigure}[b]{1.5\textwidth}
    \centering
    \begin{quantikz}
        \lstick{...}&\gategroup[3,steps=11,style={dashed,rounded corners},label style={label position=above,anchor=north,yshift=+0.3cm}]{$n=3$ $Z$-factor circuit}&&\ctrl{2}&&&&\ctrl{2}&\ctrl{1}&&\ctrl{1}&\gate{R_z}&\\
        \lstick{...}&\ctrl{1}&&&&\ctrl{1}&&&\targ{}&\gate{R_z}&\targ{}&\gate{R_z}&\\
        \lstick{...}&\targ{}&\gate{R_z}&\targ{}&\gate{R_z}&\targ{}&\gate{R_z}&\targ{}&\gate{R_z}&&&&\\
        \lstick{...}&&&&&&&&&&&&
    \end{quantikz}
\end{subfigure}}
\caption{The CNOT-optimezed $Z(\Theta_Z)$-factor circuit for $n=4$.}
\label{fig:Zsimplified_n4}
\end{figure}

\subsection{Simplifying CNOTs between even/odd contributions}\label{sec:simplifications_evenodd}
For an $n$-qubit register, there are exactly $2^{n-1}-1$ factors of type $\prod T_x^e$ that come into play in the construction of the $\Psi(\Theta_\Psi)$-factor~(\ref{eqn:def_Psi}) and the same number for $\prod T_x^o$ that come into play in the construction of the $\Phi(\Theta_\Phi)$-factor~(\ref{eqn:def_Phi}). These factors are denoted as the edges of the $\Psi/\Phi$-factor, respectively, because they flank, both from the right and from the left, each sub-factor $M_x^{e/o}$ that appears inside it. Precisely between these edges, simplifications of CNOT gates will be found. As mentioned in Section~\ref{sec:srbb}, they are constructed as the product of the disjoint transpositions of the set of permutations $P_{2^n}^{even}$ and $P_{2^n}^{odd}$. For instance, the $n=3$ case under consideration has to take into account the sets $P_8^{even}=\{P_{(2,4)},P_{(2,6)},P_{(2,8)},P_{(4,6)},P_{(4,8)},P_{(6,8)}\}$ and $P_8^{odd}=\{P_{(2,3)},P_{(2,5)},P_{(2,7)},P_{(4,5)},P_{(4,7)},P_{(6,7)}\}$. These sets are divided into the $2^{n-1}-1=3$ subsets of disjoint transpositions, each of which contains $2^{n-2}=2$ elements:
\begin{equation}
\begin{aligned}
    \prod T_1^e&=P_{(2,4)}P_{(6,8)}&\prod T_1^o&=P_{(2,3)}P_{(6,7)}\\
    \prod T_2^e&=P_{(2,6)}P_{(4,8)}&\prod T_2^o&=P_{(2,5)}P_{(4,7)}\\
    \prod T_3^e&=P_{(2,8)}P_{(4,6)}&\prod T_3^o&=P_{(2,7)}P_{(4,5)}
\end{aligned}
\end{equation}
Similarly to the case of the Z-factor previously described, the simplifications occur due to the juxtaposition of CNOTs with the same control-target pair. Therefore, keeping in mind the general structure~(\ref{eqn:U_approx}) for one single layer ($l=1$), sequences of this type will appear within the $\Psi$-factor:
\begin{equation}
\begin{split}
&\mathcal{U}_{approx}=Z(\Theta_Z^1)[A_\Psi]\cdot\\
&\cdot\left[\left(\prod T_1^e\right)M_1^e\left(\prod T_1^e\right)\left(\prod T_2^e\right)M_2^e\left(\prod T_2^e\right)\left(\prod T_3^e\right)M_3^e\left(\prod T_3^e\right)\right]\Phi(\Theta_\Phi^1)
\end{split}
\end{equation}
Simplifications can only occur between 2 edges of different macro-factors, as happens for $\prod T_2^e$ and $\prod T_3^e$ in the $n=3$ case. However, in the perspective of defining a general (and scalable) rule to find between which edges it is possible to simplify, a final binary scheme is mandatory. The latter consists of the binary table obtained by arranging the binary representation of numbers from 1 to $2^{n-1}-1=3$ in rows and considering $n-1$ qubits as columns. Then, reading the table from bottom to top (so as to follow the logical order in which the factors appear in the circuit), and comparing two consecutive lines at a time, a simplification occurs if the pair of binary strings has a 1 on the same bit. Consequently, the remaining CNOT gates are obtained by considering as targets the qubits corresponding to the bits which in the two strings takes on different values, while the control is always in the $n$-th qubit. Figure~\ref{fig:prodTe_n3} illustrates this last step, highlighting in red the simplifications and in green the remaining CNOTs.
\begin{figure}[!h]
\centering
\begin{subfigure}[b]{0.4\textwidth}
    \centering
    \begin{tabular}{|c||c|c|c|}
    \hline
    $x$&$1^{\circ}\,$q&$2^{\circ}\,$q&sub-factor\\
    \hline
    \hline
    1&0&1&$\prod T_1^e$\\
    2&\textcolor{red}{1}&0&$\prod T_2^e$\\
    3&\textcolor{red}{1}&\textcolor{green}{1}&$\prod T_3^e$\\
    \hline
    \end{tabular}
\end{subfigure}
\begin{subfigure}[b]{0.4\textwidth}
    \centering
    \begin{quantikz}
    \lstick{$1^{\circ}$q}&\targ[style={red}]{}&\slice[style={black}]{$\prod T_3^e\prod T_2^e$}&\targ[style={red}]{}&\\
    \lstick{$2^{\circ}$q}&&\targ[style={green}]{}&&\\
    \lstick{$3^{\circ}$q}&\ctrl[style={red}]{-2}&\ctrl[style={green}]{-1}&\ctrl[style={red}]{-2}&
    \end{quantikz}
\end{subfigure}
\caption{Simplification scheme for $\prod T_x^e$ factors in the $n=3$ case; only the pair of edges $\prod T_3^e\prod T_2^e$ shows simplifications.}
\label{fig:prodTe_n3}
\end{figure}
\begin{prop}\label{prop_prodT_scalable}
    A CNOT-simplification between $\prod T_x^{e/o}$ factors occurs if, in the correct logical order, these factors show pairs of bits with value 1 at the same position in the binary table that sorts them.
\end{prop}
\begin{proof}
    The specific case for $n=3$, described above as an example to clarify the conditions of simplification, has a completely general value since it is based on:
    \begin{itemize}
    \item[i)]the relationship between binary representation and the circuit implementation of $\prod T_x^{e/o}$ contributions (defined in Appendix~\ref{sec:cnot_sequences});
    \item[ii)]the properties of the binary representation valid for any $n$.
    \end{itemize}
\end{proof}

For the $\prod T_x^o$ factors of the odd case, the considerations are completely similar and the same tabular scheme is used with one only difference: once a pair of binary strings appropriate for simplification has been found, the leftmost pair of bits equal to 1 determines the value of the parameter $k$, which only for the $\prod T_x^o$ factors implies two CNOTs where the control-target pair is reversed. Figure~\ref{fig:prodTo_n3} illustrates this last step, highlighting in red the simplifications and in green the remaining CNOTs. Besides, it is possible to simplify the search for edges suitable for simplifications by restricting it to the even case, which turns out to be faster due to the absence of the parameter $k$, as stated by the following proposition.
\begin{prop}
   Once the pairs of edges with CNOT simplifications of the even case have been found, in the odd case the same pairs will result.
\end{prop}
\begin{proof}
    Both even and odd cases rely on the same binary table, in which the $\prod T_x^{e/o}$ factors are only ordered. The presence of the parameter $k$ does not modify this ordering, but only adds a pair of CNOT gates with control-target qubits reversed, before and after the structure of the even case, for each value of $x$. Thus, in addition to the simplified pairs of the even case, there will always be an additional simplified pair of CNOTs, the ones placed side by side in the simplifying border.
\end{proof}
\begin{figure}[!h]
\centering
\begin{subfigure}[b]{0.4\textwidth}
    \centering
    \begin{tabular}{|c||c|c|c|c|}
    \hline
    $x$&$1^{\circ}\,$q&$2^{\circ}\,$q&sub-factor&$k$\\
    \hline
    \hline
    1&0&1&$\prod T_1^e$&2\\
    2&\textcolor{red}{1}&0&$\prod T_2^e$&1\\
    3&\textcolor{red}{1}&\textcolor{green}{1}&$\prod T_3^e$&1\\
    \hline
    \end{tabular}
\end{subfigure}
\resizebox{0.55\textwidth}{!}{
\begin{subfigure}[b]{0.6\textwidth}
    \centering
    \begin{quantikz}
    \lstick{$1^{\circ}$q}&\ctrl[style={green}]{2}&\targ[style={red}]{}&&\ctrl[style={red}]{2}\slice[style={black}]{$\prod T_3^o\prod T_2^o$}&\ctrl[style={red}]{2}&\targ[style={red}]{}&\ctrl[style={green}]{2}&\\
    \lstick{$2^{\circ}$q}&&&\targ[style={green}]{}&&&&&\\
    \lstick{$3^{\circ}$q}&\targ[style={green}]{}&\ctrl[style={red}]{-2}&\ctrl[style={green}]{-1}&\targ[style={red}]{}&\targ[style={red}]{}&\ctrl[style={red}]{-2}&\targ[style={green}]{}&
    \end{quantikz}
\end{subfigure}}
\caption{Simplification scheme for $\prod T_x^o$ factors in the $n=3$ case; only the pair of edges $\prod T_3^o\prod T_2^o$ shows simplifications.}
\label{fig:prodTo_n3}
\end{figure}

Proposition~\ref{prop_prodT_scalable} identifies a scalable method for the $\Psi/\Phi$-factor design, optimized in terms of CNOT gates thanks to simple binary matrices, as illustrated via pseudocode in Algorithm~\ref{alg2}, \ref{alg3}, \ref{alg4} and \ref{alg5}.

\begin{algorithm}
\caption{\textbf{[PERMUTATION]} Circuit for a single $\prod T_x^{e/o}$ factor, $1\leqslant x\leqslant2^{n-1}-1$, parity = \{even, odd\}}
\label{alg2}
\begin{algorithmic}[1]
\STATE\textbf{Provided:} $CNOT(c,t)$ where $c$ stands for control and $t$ for target; the usual notation where the first qubit is labeled with 0.
\STATE\textbf{Input:} binary array corresponding to the $\prod T_x^{e/o}$ factor (see Figure~\ref{fig:prodTe_n3}); position $x$ in the sequence of $\prod T_x^{e/o}$ factors; array with $k$-index (see Figure~\ref{fig:prodTo_n3}). 
\STATE\textbf{Output:} circuit of the $\prod T_x^{e/o}$ factor for a particular position $x$.
\IF{parity = even}
\FOR{element $i$ in binary array}
\IF{element $i$ equals 1}
\STATE $CNOT(n-1,i)$
\ELSE
\STATE Pass
\ENDIF
\ENDFOR
\ELSE
\STATE Consider the element $x$ of the array with $k$-index, i.e., $k_x$
\STATE $CNOT(k_x,n-1)$
\STATE Run instructions for parity = even
\STATE $CNOT(k_x,n-1)$
\ENDIF
\end{algorithmic}
\end{algorithm}

\begin{algorithm}
\caption{\textbf{[$\Psi$-FACTOR-MAIN]} Scalable scheme for CNOT-simplifications within the alternating sequence of $\prod T_x^e$ and $M_x^e$ factors, $1\leqslant x\leqslant2^{n-1}-1$}
\label{alg3}
\begin{algorithmic}[1]
\STATE\textbf{Provided:} \emph{CNOT-remaining-Gray} (see Algorithm $Z$-FACTOR); an algorithm that generates the binary matrix, obtained by comparing two rows at a time of the \emph{reverse} binary matrix that orders $\prod T_x^e$ factors (see Figure~\ref{fig:prodTe_n3}), where the 1s indicate the positions in which the bits have different values (called \emph{CNOT-remaining-EvenEdges}); a scalable algorithm for $M_x^e$ circuits, called \emph{Me-factor-circuit}.
\STATE\textbf{Input:} set of parameters $\theta_a\in\Theta_\Psi$; output of \emph{CNOT-remaining-Gray}; output of \emph{CNOT-remaining-EvenEdges}; output of \emph{Me-factor-circuit}.
\STATE\textbf{Output:} circuit for the 2nd sub-factor of $\Psi$, in the correct logic order but without $\prod T_{x_{max}}^e$, $\prod T_1^e$ and $M_{x_{max}}^e$, optimized in terms of CNOT gates.
\STATE Consider the output of \emph{CNOT-remaining-EvenEdges}, i.e., a binary matrix
\FOR{row $i$ in binary matrix}
\FOR{column $j$ in row $i$}\label{instr6}
\IF{element $(i,j)$ equals 1}
\STATE $CNOT(n-1,j)$
\ELSE
\STATE Pass
\ENDIF
\ENDFOR\label{instr12}
\STATE Run \emph{Me-factor-circuit}\label{instr_to_eliminate}
\ENDFOR
\end{algorithmic}
\end{algorithm}

\begin{algorithm}
\caption{\textbf{[$\Phi$-FACTOR-MAIN]} Scalable scheme for CNOT-simplifications within the alternating sequence of $\prod T_x^o$ and $M_x^o$ factors, $1\leqslant x\leqslant2^{n-1}-1$}
\label{alg4}
\begin{algorithmic}[1]
\STATE\textbf{Provided:} \emph{CNOT-remaining-Gray} (see Algorithm $Z$-FACTOR); an algorithm that generates the matrix, obtained by comparing two rows at a time of the \emph{reverse} binary matrix that orders $\prod T_x^o$ factors (see Figure~\ref{fig:prodTo_n3}), where the 1s indicate the positions in which the bits have different values, the 2s the positions in which the bits have value equals to 1, the 0s the remaining positions (called \emph{CNOT-remaining-OddEdges}); a scalable algorithm for $M_x^o$ circuits, called \emph{Mo-factor-circuit}; Algorithm $\Psi$-FACTOR-MAIN.
\STATE\textbf{Input:} set of parameters $\theta_a\in\Theta_\Phi$; output of \emph{CNOT-remaining-Gray}; output of \emph{CNOT-remaining-OddEdges}; output of \emph{Mo-factor-circuit}; array with $k$-index (see Figure~\ref{fig:prodTo_n3}); output of Algorithm \ref{alg3} $\Psi$-FACTOR-MAIN.
\STATE\textbf{Output:} circuit for the $\Phi$ factor, in the correct logic order but without $\prod T_{x_{max}}^o$, $\prod T_1^o$ and $M_{x_{max}}^o$, optimized in terms of CNOT gates.
\STATE Consider the output of \emph{CNOT-remaining-OddEdges, i.e., a matrix with entries \{0,1,2\}}
\FOR{row $i$ in matrix}
\FOR{column $j$ in row $i$}
\IF{element $(i,j)$ equals 2}
\STATE Consider the element $i$ of the array with $k$-index, i.e., $k_i$
\STATE $CNOT(k_i,n-1)$
\STATE Run instructions from \ref{instr6} to \ref{instr12} of Algorithm $\Psi$-FACTOR-MAIN (with the output of \emph{CNOT-remaining-EvenEdges})
\STATE $CNOT(k_i,n-1)$
\STATE Run \emph{Mo-factor-circuit}
\STATE Break
\ELSIF{no element $(i,j)$ equals 2}
\STATE This is an edge without CNOT-simplifications
\STATE Run Algorithm PERMUTATION for $x=i$ and parity = odd
\STATE Run Algorithm PERMUTATION for $x=i+1$ and parity = odd
\STATE Run \emph{Mo-factor-circuit}
\STATE Break
\ENDIF
\ENDFOR
\ENDFOR
\end{algorithmic}
\end{algorithm}

\begin{algorithm}
\caption{\textbf{[$\Psi/\Phi$-FACTORS]} Scalable CNOT-optimized $\Psi(\Theta_\Psi)/\Phi(\Theta_\Phi)$-factor circuit, parity = \{even, odd\}}
\label{alg5}
\begin{algorithmic}
\STATE\textbf{Provided:} Algorithm PERMUTATION, Algorithm $\Psi$-FACTOR-MAIN and Algorithm $\Phi$-FACTOR-MAIN.
\STATE\textbf{Input:} set of parameters $\theta_a\in\{\Theta_\Psi\cup\Theta_\Phi\}$; $x_{max}=2^{n-1}-1$; outputs of Algorithm PERMUTATION, Algorithm $\Psi$-FACTOR-MAIN and Algorithm $\Phi$-FACTOR-MAIN.
\STATE\textbf{Output:} $\Psi$-factor and $\Phi$-factor circuits, in the correct logic order and optimized in terms of CNOT gates.
\IF{parity = even}
\STATE Run Algorithm PERMUTATION for $x=x_{max}$ and parity = even
\STATE Run \emph{Me-factor-circuit} (for $x=x_{max})$
\STATE Run Algorithm $\Psi$-FACTOR-MAIN
\STATE Run Algorithm PERMUTATION for $x=1$ and parity = even
\STATE Run \emph{Me-factor-circuit} (for the 1st sub-factor of $\Psi$)
\ELSE
\STATE Run Algorithm PERMUTATION for $x=x_{max}$ and parity = odd
\STATE Run \emph{Mo-factor-circuit} (for $x=x_{max}$)
\STATE Run Algorithm $\Phi$-FACTOR-MAIN
\STATE Run Algorithm \ref{alg2} PERMUTATION (for $x=1$ and parity = odd)
\ENDIF
\end{algorithmic}
\end{algorithm}

\subsubsection{Solution for n=4}
In the case of 4-qubit systems, there are exactly $2^{n-1}-1=7$ factors of type $\prod T_x^{e}$ inside $\Psi(\Theta_\Psi)$ and the same number of type $\prod T_x^{o}$ inside $\Phi(\Theta_\Phi)$. As described in Section~\ref{sec:srbb}, they are defined as the product of disjoint 2-cycles from the set of permutations $P_{16}^{even}$ and $P_{16}^{odd}$, respectively. However, for the purposes of this Section, it is more useful to use their representation in terms of CNOTs\footnote{As usual, the first label of the pair stands for the control qubit and second one for the target.}:
{\footnotesize
\begin{equation}
\begin{aligned}
    \prod T_1^e&=CNOT_{(4,3)}&\prod T_1^o&=CNOT_{(3,4)}\prod T_1^e\,CNOT_{(3,4)}\\
    \prod T_2^e&=CNOT_{(4,2)}&\prod T_2^o&=CNOT_{(2,4)}\prod T_2^e\,CNOT_{(2,4)}\\
    \prod T_3^e&=CNOT_{(4,2)}CNOT_{(4,3)}&\prod T_3^o&=CNOT_{(2,4)}\prod T_3^e\,CNOT_{(2,4)}\\
    \prod T_4^e&=CNOT_{(4,1)}&\prod T_4^o&=CNOT_{(1,4)}\prod T_4^e\,CNOT_{(1,4)}\\
    \prod T_5^e&=CNOT_{(4,1)}CNOT_{(4,3)}&\prod T_5^o&=CNOT_{(1,4)}\prod T_5^e\,CNOT_{(1,4)}\\
    \prod T_6^e&=CNOT_{(4,1)}CNOT_{(4,2)}&\prod T_6^o&=CNOT_{(1,4)}\prod T_6^e\,CNOT_{(1,4)}\\
    \prod T_7^e&=CNOT_{(4,1)}CNOT_{(4,2)}CNOT_{(4,3)}&\prod T_7^o&=CNOT_{(1,4)}\prod T_7^e\,CNOT_{(1,4)}
\end{aligned}
\end{equation}}
To find between which edges it is possible to simplify CNOT pairs, the table described in the previous Section is constructed (note that the table should be read from the bottom and in the odd case it has an extra column for the index $k$). The resulting gates are shown in Figure~\ref{fig:prodTe_n4} for the even case and in Figure~\ref{fig:prodTo_n4} for the odd one. Since for $n=4$ the edges suitable for simplifications are four in both cases, the rows of the tables are partitioned for clarity.
\begin{figure}[htbp]
\begin{subfigure}[b]{0.4\textwidth}
    \centering
    \resizebox{0.9\textwidth}{!}{
    \begin{tabular}{|c||c|c|c|c|}
    \hline
    $x$&$1^{\circ}\,$q&$2^{\circ}\,$q&$3^{\circ}\,$q&sub-factor\\
    \hline
    \hline
    1&0&0&1&$\prod T_1^e$\\
    2&0&\textcolor{red}{1}&0&$\prod T_2^e$\\
    3&0&\textcolor{red}{1}&\textcolor{green}{1}&$\prod T_3^e$\\
    $\vdots$&$\vdots$&$\vdots$&$\vdots$&$\vdots$\\
    \hline
    \end{tabular}}
\end{subfigure}
\begin{subfigure}[b]{0.55\textwidth}
    \centering
    \resizebox{0.45\textwidth}{!}{
    \begin{quantikz}
    \lstick{$1^{\circ}$q}&&\slice[style={black}]{$\prod T_3^e\prod T_2^e$}&&\\
    \lstick{$2^{\circ}$q}&\targ[style={red}]{}&&\targ[style={red}]{}&\\
    \lstick{$3^{\circ}$q}&&\targ[style={green}]{}&&\\
    \lstick{$4^{\circ}$q}&\ctrl[style={red}]{-2}&\ctrl[style={green}]{-1}&\ctrl[style={red}]{-2}&
    \end{quantikz}}
\end{subfigure}
\begin{subfigure}[b]{0.4\textwidth}
    \centering
    \resizebox{0.9\textwidth}{!}{
    \begin{tabular}{|c||c|c|c|c|}
    \hline
    $x$&$1^{\circ}\,$q&$2^{\circ}\,$q&$3^{\circ}\,$q&sub-factor\\
    \hline
    \hline
    $\vdots$&$\vdots$&$\vdots$&$\vdots$&$\vdots$\\
    4&\textcolor{red}{1}&0&0&$\prod T_4^e$\\
    5&\textcolor{red}{1}&0&\textcolor{green}{1}&$\prod T_5^e$\\
    $\vdots$&$\vdots$&$\vdots$&$\vdots$&$\vdots$\\
    \hline
    \end{tabular}}
\end{subfigure}
\begin{subfigure}[b]{0.55\textwidth}
    \centering
    \resizebox{0.45\textwidth}{!}{
    \begin{quantikz}
    \lstick{$1^{\circ}$q}&\targ[style={red}]{}&\slice[style={black}]{$\prod T_5^e\prod T_4^e$}&\targ[style={red}]{}&\\
    \lstick{$2^{\circ}$q}&&&&\\
    \lstick{$3^{\circ}$q}&&\targ[style={green}]{}&&\\
    \lstick{$4^{\circ}$q}&\ctrl[style={red}]{-3}&\ctrl[style={green}]{-1}&\ctrl[style={red}]{-3}&
    \end{quantikz}}
\end{subfigure}
\begin{subfigure}[b]{0.4\textwidth}
    \centering
    \resizebox{0.9\textwidth}{!}{
    \begin{tabular}{|c||c|c|c|c|}
    \hline
    $x$&$1^{\circ}\,$q&$2^{\circ}\,$q&$3^{\circ}\,$q&sub-factor\\
    \hline
    \hline
    $\vdots$&$\vdots$&$\vdots$&$\vdots$&$\vdots$\\
    5&\textcolor{red}{1}&0&\textcolor{green}{1}&$\prod T_5^e$\\
    6&\textcolor{red}{1}&\textcolor{green}{1}&0&$\prod T_6^e$\\
    $\vdots$&$\vdots$&$\vdots$&$\vdots$&$\vdots$\\
    \hline
    \end{tabular}}
\end{subfigure}
\begin{subfigure}[b]{0.55\textwidth}
    \centering
    \resizebox{0.45\textwidth}{!}{
    \begin{quantikz}
    \lstick{$1^{\circ}$q}&\targ[style={red}]{}&\slice[style={black}]{$\prod T_6^e\prod T_5^e$}&\targ[style={red}]{}&&\\
    \lstick{$2^{\circ}$q}&&\targ[style={green}]{}&&&\\
    \lstick{$3^{\circ}$q}&&&&\targ[style={green}]{}&\\
    \lstick{$4^{\circ}$q}&\ctrl[style={red}]{-3}&\ctrl[style={green}]{-2}&\ctrl[style={red}]{-3}&\ctrl[style={green}]{-1}&
    \end{quantikz}}
\end{subfigure}
\begin{subfigure}[b]{0.4\textwidth}
    \centering
    \resizebox{0.9\textwidth}{!}{
    \begin{tabular}{|c||c|c|c|c|}
    \hline
    $x$&$1^{\circ}\,$q&$2^{\circ}\,$q&$3^{\circ}\,$q&sub-factor\\
    \hline
    \hline
    $\vdots$&$\vdots$&$\vdots$&$\vdots$&$\vdots$\\
    5&1&0&1&$\prod T_5^e$\\
    6&\textcolor{red}{1}&\textcolor{red}{1}&0&$\prod T_6^e$\\
    7&\textcolor{red}{1}&\textcolor{red}{1}&\textcolor{green}{1}&$\prod T_7^e$\\
    \hline
    \end{tabular}}
\end{subfigure}\qquad
\begin{subfigure}[b]{0.55\textwidth}
    \centering
    \resizebox{0.55\textwidth}{!}{
    \begin{quantikz}
    \lstick{$1^{\circ}$q}&\targ[style={red}]{}&&\slice[style={black}]{$\prod T_7^e\prod T_6^e$}&\targ[style={red}]{}&&\\
    \lstick{$2^{\circ}$q}&&\targ[style={red}]{}&&&\targ[style={red}]{}&\\
    \lstick{$3^{\circ}$q}&&&\targ[style={green}]{}&&&\\
    \lstick{$4^{\circ}$q}&\ctrl[style={red}]{-3}&\ctrl[style={red}]{-2}&\ctrl[style={green}]{-1}&\ctrl[style={red}]{-3}&\ctrl[style={red}]{-2}&
    \end{quantikz}}
\end{subfigure}
\caption{Simplification scheme for $\prod T_x^e$ factors for $n=4$ case.}
\label{fig:prodTe_n4}
\end{figure}

\begin{figure}[ht]
\begin{subfigure}[b]{0.4\textwidth}
    \centering
    \resizebox{0.9\textwidth}{!}{
    \begin{tabular}{|c||c|c|c|c|c|}
    \hline
    $x$&$1^{\circ}\,$q&$2^{\circ}\,$q&$3^{\circ}\,$q&sub-factor&$k$\\
    \hline
    \hline
    1&0&0&1&$\prod T_1^o$&3\\
    2&0&\textcolor{red}{1}&0&$\prod T_2^o$&2\\
    3&0&\textcolor{red}{1}&\textcolor{green}{1}&$\prod T_3^o$&2\\
    $\vdots$&$\vdots$&$\vdots$&$\vdots$&$\vdots$&$\vdots$\\
    \hline
    \end{tabular}}
\end{subfigure}
\begin{subfigure}[b]{0.55\textwidth}
    \centering
    \resizebox{0.65\textwidth}{!}{
    \begin{quantikz}
    \lstick{$1^{\circ}$q}&&&&\slice[style={black}]{$\prod T_3^o\prod T_2^o$}&&&&\\
    \lstick{$2^{\circ}$q}&\ctrl[style={green}]{2}&\targ[style={red}]{}&&\ctrl[style={red}]{2}&\ctrl[style={red}]{2}&\targ[style={red}]{}&\ctrl[style={green}]{2}&\\
    \lstick{$3^{\circ}$q}&&&\targ[style={green}]{}&&&&&\\
    \lstick{$4^{\circ}$q}&\targ[style={green}]{}&\ctrl[style={red}]{-2}&\ctrl[style={green}]{-1}&\targ[style={red}]{}&\targ[style={red}]{}&\ctrl[style={red}]{-2}&\targ[style={green}]{}&
    \end{quantikz}}
\end{subfigure}
\begin{subfigure}[b]{0.4\textwidth}
    \centering
    \resizebox{0.9\textwidth}{!}{
    \begin{tabular}{|c||c|c|c|c|c|}
    \hline
    $x$&$1^{\circ}\,$q&$2^{\circ}\,$q&$3^{\circ}\,$q&sub-factor&$k$\\
    \hline
    \hline
    $\vdots$&$\vdots$&$\vdots$&$\vdots$&$\vdots$&$\vdots$\\
    4&\textcolor{red}{1}&0&0&$\prod T_4^e$&1\\
    5&\textcolor{red}{1}&0&\textcolor{green}{1}&$\prod T_5^e$&1\\
    $\vdots$&$\vdots$&$\vdots$&$\vdots$&$\vdots$&$\vdots$\\
    \hline
    \end{tabular}}
\end{subfigure}
\begin{subfigure}[b]{0.55\textwidth}
    \centering
    \resizebox{0.65\textwidth}{!}{
    \begin{quantikz}
    \lstick{$1^{\circ}$q}&\ctrl[style={green}]{3}&\targ[style={red}]{}&&\ctrl[style={red}]{3}\slice[style={black}]{$\prod T_5^o\prod T_4^o$}&\ctrl[style={red}]{3}&\targ[style={red}]{}&\ctrl[style={green}]{3}&\\
    \lstick{$2^{\circ}$q}&&&&&&&&\\
    \lstick{$3^{\circ}$q}&&&\targ[style={green}]{}&&&&&\\
    \lstick{$4^{\circ}$q}&\targ[style={green}]{}&\ctrl[style={red}]{-3}&\ctrl[style={green}]{-1}&\targ[style={red}]{}&\targ[style={red}]{}&\ctrl[style={red}]{-3}&\targ[style={green}]{}&
    \end{quantikz}}
\end{subfigure}
\begin{subfigure}[b]{0.4\textwidth}
    \centering
    \resizebox{0.9\textwidth}{!}{
    \begin{tabular}{|c||c|c|c|c|c|}
    \hline
    $x$&$1^{\circ}\,$q&$2^{\circ}\,$q&$3^{\circ}\,$q&sub-factor&$k$\\
    \hline
    \hline
    $\vdots$&$\vdots$&$\vdots$&$\vdots$&$\vdots$&$\vdots$\\
    5&\textcolor{red}{1}&0&\textcolor{green}{1}&$\prod T_5^o$&1\\
    6&\textcolor{red}{1}&\textcolor{green}{1}&0&$\prod T_6^o$&1\\
    $\vdots$&$\vdots$&$\vdots$&$\vdots$&$\vdots$&$\vdots$\\
    \hline
    \end{tabular}}
\end{subfigure}
\begin{subfigure}[b]{0.55\textwidth}
    \centering
    \resizebox{0.65\textwidth}{!}{
    \begin{quantikz}
    \lstick{$1^{\circ}$q}&\ctrl[style={green}]{3}&\targ[style={red}]{}&&\ctrl[style={red}]{3}\slice[style={black}]{$\prod T_6^o\prod T_5^o$}&\ctrl[style={red}]{3}&\targ[style={red}]{}&&\ctrl[style={green}]{3}&\\
    \lstick{$2^{\circ}$q}&&&\targ[style={green}]{}&&&&&&\\
    \lstick{$3^{\circ}$q}&&&&&&&\targ[style={green}]{}&&\\
    \lstick{$4^{\circ}$q}&\targ[style={green}]{}&\ctrl[style={red}]{-3}&\ctrl[style={green}]{-2}&\targ[style={red}]{}&\targ[style={red}]{}&\ctrl[style={red}]{-3}&\ctrl[style={green}]{-1}&\targ[style={green}]{}&
    \end{quantikz}}
\end{subfigure}
\begin{subfigure}[b]{0.4\textwidth}
    \centering
    \resizebox{0.9\textwidth}{!}{
    \begin{tabular}{|c||c|c|c|c|c|}
    \hline
    $x$&$1^{\circ}\,$q&$2^{\circ}\,$q&$3^{\circ}\,$q&sub-factor&$k$\\
    \hline
    \hline
    $\vdots$&$\vdots$&$\vdots$&$\vdots$&$\vdots$&$\vdots$\\
    5&1&0&1&$\prod T_5^o$&1\\
    6&\textcolor{red}{1}&\textcolor{red}{1}&0&$\prod T_6^o$&1\\
    7&\textcolor{red}{1}&\textcolor{red}{1}&\textcolor{green}{1}&$\prod T_7^o$&1\\
    \hline
    \end{tabular}}
\end{subfigure}\qquad
\begin{subfigure}[b]{0.55\textwidth}
    \centering
    \resizebox{0.75\textwidth}{!}{
    \begin{quantikz}
    \lstick{$1^{\circ}$q}&\ctrl[style={green}]{3}&\targ[style={red}]{}&&&\ctrl[style={red}]{3}\slice[style={black}]{$\prod T_7^o\prod T_6^o$}&\ctrl[style={red}]{3}&\targ[style={red}]{}&&\ctrl[style={green}]{3}&\\
    \lstick{$2^{\circ}$q}&&&\targ[style={red}]{}&&&&&\targ[style={red}]{}&&\\
    \lstick{$3^{\circ}$q}&&&&\targ[style={green}]{}&&&&&&\\
    \lstick{$4^{\circ}$q}&\targ[style={green}]{}&\ctrl[style={red}]{-3}&\ctrl[style={red}]{-2}&\ctrl[style={green}]{-1}&\targ[style={red}]{}&\targ[style={red}]{}&\ctrl[style={red}]{-3}&\ctrl[style={red}]{-2}&\targ[style={green}]{}&
    \end{quantikz}}
\end{subfigure}
\caption{Simplification scheme for $\prod T_x^o$ factors for $n=4$ case.}
\label{fig:prodTo_n4}
\end{figure}

Finally, the simplification algorithms proposed above can be merged into an overall algorithm (see Algorithm~\ref{alg6}) for the scalable (and partially recursive) construction of the CNOT-optimized VQC for the approximation of $SU(2^n)$ operators.

\begin{algorithm}
\caption{\textbf{[SU-SYNTHESIS]} CNOT-optimized circuit to approximate $SU(2^n)$ operators within SRBB-framework}
\label{alg6}
\begin{algorithmic}
\STATE\textbf{Provided:} CNOT-optimized circuit for $n=2$, Algorithm $\Psi/\Phi$-FACTORS and Algorithm $Z$-FACTOR.
\STATE\textbf{Input:} number of qubits $n\geqslant2$; set of parameters $\theta_i\in\{\Theta_Z\cup\Theta_\Psi\cup\Theta_\Phi\}$; outputs of Algorithm $\Psi/\Phi$-FACTORS and Algorithm $Z$-FACTOR.
\STATE\textbf{Output:} complete circuit, optimized in terms of CNOT gates, to approximate a given special unitary matrix.
\STATE Partition of the entire set of parameters into three subsets, those for $\Theta_Z$, $\Theta_\Psi$ and $\Theta_\Phi$
\IF{n = 2}
\STATE $CNOT$-optimized circuit for $n=2$ (see Figure \ref{fig:n2circuit})
\ELSE
\STATE Run Algorithm $\Psi/\Phi$-FACTORS for parity = odd
\STATE Run Algorithm $\Psi/\Phi$-FACTORS for parity = even
\STATE Run Algorithm $Z$-FACTOR
\ENDIF
\end{algorithmic}
\end{algorithm}

\subsection{CNOT-reduced VQC to approximate SU operators with SRBB-decomposition}
The scalable simplification algorithm, valid for any $n\geqslant3$, has been successfully incorporated into the original scalability scheme~\cite{sarkar2024quantum}, reducing the number of CNOTs in the overall circuit.
\begin{prop}
The following formulas are obtained for gates count after the scalable CNOT-simplification procedure on the approximate SRBB-based synthesis algorithm when $n\geqslant3$:
\begin{equation}\label{eqn:gate_counts_general}
\begin{aligned}
    N_{CNOT} &= 2^{2n + 1} - 5 \cdot 2^{n - 1} + 2n -4
    \\
    N_{Rot} &= 2^{2n + 1} - 5 \cdot 2^{n - 1} + 1
    \\
    CNOT_{red} &= 2^n(2n-5)-2n+8 
\end{aligned}
\end{equation}
where the first and second equations represent, respectively, the exponential trend of the total number of CNOT-gates and rotations (both z- and y-rotations), while the last one calculates the reduction of CNOT-gates compared to the original structure proposed by the literature. For the special case $n=2$, instead, it is found:
\begin{equation}\label{eqn:gate_counts_n2}
    \begin{aligned}
        N_{CNOT}&=\frac{3}{2}\cdot4^n+\left(n-\frac{3}{2}\right)2^n-8=18\\
        N_{Rot}&=\frac{3}{2}\cdot4^n-2^{n-1}-1=21\\
        CNOT_{red}&=4
    \end{aligned}
\end{equation}
\end{prop}
For the length and modularity features of proofs, these trends are demonstrated in the next Sections~\ref{sec:proof_gate_counts} and \ref{sec:anomaly_n2}.

\subsubsection{Proof of gate count formulas}\label{sec:proof_gate_counts}
For general $n$, defining the overall function that describes the number of CNOT gates after the simplification process means adding all the relevant contributions, as they appear in the original structure, and then subtracting the simplified quantity\footnote{In this way, the formula describing the difference in terms of CNOTs with respect to the original structure will be obtained accordingly.}:
\begin{equation}
\begin{split}
    N_{CNOT}&=N_{CNOT}(Z)+N_{CNOT}(1\mbox{st subfactor of }\Psi)+N_{CNOT}(M_x^e)+\\
    &+N_{CNOT}\left(\prod T_x^e\right)+N_{CNOT}(M_x^o)+N_{CNOT}\left(\prod T_x^o\right)-\\
    &-N_{CNOT}\left(\mbox{simplif. }\prod T_x^e\right)-N_{CNOT}\left(\mbox{simplif. }\prod T_x^o\right)
\end{split}
\end{equation}
where $N_{CNOT}\left(\mbox{simplif. }\prod T_x^{e/o}\right)$ represents the number of simplified CNOTs inside factors $\prod T_x^{e/o}$, and the simplifications within the $Z$-factor are not included because, from $n\geqslant3$, the term $N_{CNOT}(Z)$ already takes them into account. In the following, each term will be analyzed separately and then added only at the end.
\begin{enumerate}
    \item[$\square$]$N_{CNOT}(Z)$\\As anticipated, this term represents the number of resulting CNOTs after the simplifications only for $n\geqslant3$. With reference to Section~\ref{sec:gray_code}, the number of CNOT gates for each $Z$-factor corresponds to the $2^{n-1}$ jumps among the rows ending with 1 in the cyclic Gray Code matrix. Considering the recursive pattern within the diagonal contributions, the overall number of CNOTs for the case $n$ in question is equivalent to the series:
    \begin{equation}
        \sum_{m=2}^{n}2^{m-1}=\left(\sum_{m=0}^{n}2^{m-1}\right)-\frac{3}{2}=2^n-2
    \end{equation}
    To calculate the effect of the simplification process with respect to the original structure, it is also very useful to calculate the number of CNOTs of the diagonal factor without the simplifications. The latter corresponds to double the number of 1-valued bits within the columns of the control qubits (all except the last one, associated with parameters $m'$) in the binary tables that build the $Z$-factor circuit. This method analyzes the binary tables by columns and takes advantage of the fact that in each column there is the same number of 1s, even if they are grouped differently. In formulas, it is worth noting that:
\begin{equation}
        2\sum_{n=2}^{M}2^{M-n}(2^{n-1}-1)=2^M(M+2)+2-2^{M+2}=2^M(M-2)+2
    \end{equation}
    where $M=n_{max}$, $2^{M-n}$ is the number of subgroups in each column and $2^{n-1}-1$ is the number of 1s associated with $m'$ in each subgroup. This result allows for directly calculating the number of CNOT-simplifications within the $Z$-factor for $n\geqslant3$ ($M=n_{max}\rightarrow n$):
    \begin{equation}
        [2^n(n-2)+2]-(2^n+2)=2^n(n-3)+4
    \end{equation}
    \item[$\square$]$N_{CNOT}(1\mbox{st subfactor of }\Psi)+N_{CNOT}(M_x^e)$\\These two factors are analyzed together because the number of CNOTs contained in them depends on the $ZYZ$-decomposition property that they both satisfy. As explained in~\cite{sarkar2024quantum}, a $M_nZYZ$-type matrix can be fully decomposed and implemented through $3\cdot2^{n-1}-2$ CNOT gates. Thus, considering the $2^{n-1}-1$ factors of type $M_x^e$ plus the first sub-factor of $\Psi(\Theta_\Psi)$ that fall into this decomposition, the total number of CNOTs is $2^{n-1}(3\cdot2^{n-1}-2)$.
    \item[$\square$]$N_{CNOT}(M_x^o)$\\Due to their algebraic properties, these sub-factors cannot be implemented with a pure $ZYZ$-decomposition, but multi-controlled $z$-rotations must be added before and after a $ZYZ$-type kernel~\cite{sarkar2024quantum}. For this reason, the total number of CNOTs is $(2^{n-1}-1)(5\cdot2^{n-1}-6)$. 
    \item[$\square$]$N_{CNOT}\left(\prod T_x^e\right)$\\With reference to the binary tables used to built the corresponding quantum circuits for these factors, the total number of CNOTs is equivalent to double the number of 1-valued bits, remembering that target columns are $n-1$. Thus, in each column there are exactly $2^{n-2}$ 1-valued bits and the formula for the total number of CNOTs within these factors is $2(n-1)2^{n-2}$.
    \item[$\square$]$N_{CNOT}\left(\prod T_x^o\right)$\\For these factors the same rules apply as for the even case with only one difference due to their very definition (presence of the parameter $k$): the total number of CNOTs is equal to the number of CNOTs for the even case plus the number of pairs CNOT$_{(k,n)}$. The final formula is $2[(n-1)2^{n-2}+2(2^{n-1}-1)]$.  
    \item[$\square$]$N_{CNOT}\left(\mbox{simplif. }\prod T_x^e\right)$\\To calculate the number of simplifications inside $\prod T_x^e$ factors it is absolutely important to understand the recursive pattern with which pairs of 1s (simplifications) appear in the binary tables that encode these factors. In Figure~\ref{tab:prodTe_simplifications}, this pattern is illustrated for $n=3,4,5$ to simplify the explanation. Firstly, the 1s highlighted in red form the pairs indicating the simplifications according to the rule already explained in Section~\ref{sec:simplifications_evenodd}; secondly, the entire table is divided into two blocks by a horizontal line that partitions it exactly in half: it is immediate to observe that the number of simplifications of case $n$ corresponds to the number of simplifications present in the upper half of the table of case $n+1$. Finally, the number of simplifications (pairs of 1) of the lower half equals $2^{n-2}-1$ and a fraction of these, more precisely all except those of the first column, are equal to those of the upper half. Thus, introducing the function $f_{max}(i)=2^{n_{max}-2-i}-1$, it is possible to count the total number of pairs of (vertically) adjacent 1-valued bits (CNOT-simplifications) through the formula
    \begin{equation}
        \sum_{i=0}^{n_{max}-3}2^i\cdot f_{max}(i)
    \end{equation}
    and, multiplying by 2, obtain the corresponding number of CNOT gates  ($n_{max}\rightarrow n$):
    \begin{equation}
    \begin{split}
         N_{CNOT}\left(\mbox{simplif. }\prod T_x^e\right)&=2\left(\sum_{i=0}^{n-3}2^i\cdot f_{max}(i)\right)=\sum_{i=0}^{n-3}2^{i+1}(2^{n-2-i}-1)=\\
         &=\sum_{i=0}^{n-3}2^{n-1}-\sum_{i=0}^{n-3}2^{i+1}=2^{n-1}(n-3)+2
    \end{split}
    \end{equation}

    \begin{figure}
    \centering
    \resizebox{0.7\textwidth}{!}{
    \begin{subfigure}[b]{0.3\textwidth}
        \centering
        \begin{tabular}{|c|c|}
        \hline
        \multicolumn{2}{|c|}{$n=3$}\\
        \hline
        \hline
        0&0\\
        0&1\\
        \hline
        \textcolor{red}{1}&0\\
        \textcolor{red}{1}&1\\
        \hline
        \end{tabular}
    \end{subfigure}$\longrightarrow$
    \begin{subfigure}[b]{0.3\textwidth}
        \centering
        \begin{tabular}{|c|c|c|}
        \hline
        \multicolumn{3}{|c|}{$n=4$}\\
        \hline
        \hline
        0&0&0\\
        0&0&1\\
        0&\textcolor{red}{1}&0\\
        0&\textcolor{red}{1}&1\\
        \hline
        \textcolor{red}{1}&0&0\\
        \textcolor{red}{1}&0&1\\
        \textcolor{red}{1}&\textcolor{red}{1}&0\\
        \textcolor{red}{1}&\textcolor{red}{1}&1\\
        \hline
        \end{tabular}
    \end{subfigure}$\longrightarrow$
    \begin{subfigure}[b]{0.3\textwidth}
        \centering
        \begin{tabular}{|c|c|c|c|}
        \hline
        \multicolumn{4}{|c|}{$n=5$}\\
        \hline
        \hline
        0&0&0&0\\
        0&0&0&1\\
        0&0&\textcolor{red}{1}&0\\
        0&0&\textcolor{red}{1}&1\\
        0&\textcolor{red}{1}&0&0\\
        0&\textcolor{red}{1}&0&1\\
        0&\textcolor{red}{1}&\textcolor{red}{1}&0\\
        0&\textcolor{red}{1}&\textcolor{red}{1}&1\\
        \hline
        \textcolor{red}{1}&0&0&0\\
        \textcolor{red}{1}&0&0&1\\
        \textcolor{red}{1}&0&\textcolor{red}{1}&0\\
        \textcolor{red}{1}&0&\textcolor{red}{1}&1\\
        \textcolor{red}{1}&\textcolor{red}{1}&0&0\\
        \textcolor{red}{1}&\textcolor{red}{1}&0&1\\
        \textcolor{red}{1}&\textcolor{red}{1}&\textcolor{red}{1}&0\\
        \textcolor{red}{1}&\textcolor{red}{1}&\textcolor{red}{1}&1\\
        \hline
        \end{tabular}
    \end{subfigure}}
    \caption{Recursive scheme for simplification within $\prod T^e_x$ factors.}
    \label{tab:prodTe_simplifications}
    \end{figure}
    
    \item[$\square$]$N_{CNOT}\left(\mbox{simplif. }\prod T_x^o\right)$\\For the odd case the same type of analysis applies, but a contribution must be added for the simplifications that come from the pairs of 1s added by the parameter $k$. Remembering the definition of the parameter $k$, which in Figure~\ref{tab:prodTe_simplifications} corresponds to calculating the number of pairs that can be created with the leftmost 1s in each row, the function $f_{max}(i)$ introduced for the even case can be used to derive the number of simplifications even in the odd case:
    \begin{equation}
    \begin{split}
        N_{CNOT}\left(\mbox{simplif. }\prod T_x^o\right)&=N_{CNOT}\left(\mbox{simplif. }\prod T_x^e\right)+2\sum_{i=3}^{n}f_i(0)=\\
        &=2^{n-1}(n-3)+2+2\sum_{i=3}^{n}(2^{i-2}-1)=\\
        &=2^{n-1}(n-3)+2-2n+4+\frac{1}{2}\left(\sum_{i=0}^{n}2^i-\sum_{i=0}^{2}2^i\right)=\\
        &=2^n-2n
    \end{split}
    \end{equation}
\end{enumerate}
Putting everything together, the number of CNOTs expressed by the first of equations~(\ref{eqn:gate_counts_general}) is proven.

To deduce the total number of rotations $N_{Rot}$, the second of equations~(\ref{eqn:gate_counts_general}), the following contributions must be considered:
\begin{equation}
    N_{Rot}=N_{Rot}(Z)+N_{Rot}(1\mbox{st subfactor of }\Psi)+N_{Rot}(M_x^e)+N_{Rot}(M_x^o)
\end{equation}
Below, each term will be analyzed separately and then added only at the end.
\begin{enumerate}
    \item[$\square$]$N_{Rot}(Z)$\\In the $Z$-factor, the number of rotations is very simple since it corresponds to the number of diagonal elements involved, that is $2^n-1$.
    \item[$\square$]$N_{Rot}(1\mbox{st subfactor of }\Psi)$\\Thanks to its algebraic properties that allow a complete $ZYZ$-decomposition, this factor has exactly $3\cdot2^{n-1}$ rotations.
    \item[$\square$]$N_{Rot}(M_x^e)$\\Again due to decomposition properties, each factor of type $M_x^e$ has $3\cdot2^{n-1}$ rotations. Thus, the total number is simply $3(2^{n-1}-1)2^{n-1}$.
    \item[$\square$]$N_{Rot}(M_x^o)$\\As explained in~\cite{sarkar2023scalable}, a block-diagonal unitary matrix with $U(2)$ blocks can be decomposed by appropriately modifying the structure of a $ZYZ$-type matrix. This modification involves the addition of $n-2$ multi-controlled $z$-rotations and 1 $z$-rotations, before and after the central $ZYZ$-core. Considering that, as for the even case, there are $2^{n-1}-1$ factors of this type, the total number of rotations in the odd case is expressed by the equation:
    \begin{equation}
        (2^{n-1}-1)\left[2\sum_{m=0}^{n-2}2^m+3\cdot2^{n-1}\right]=(2^{n-1}-1)(5\cdot2^{n-1}-2)
    \end{equation}
    where the total number of $y$-rotations is $(2^{n-1}-1)\cdot2^{n-1}$ since in the $ZYZ$-core there is only one multi-controlled $y$-rotation with $2^{n-1}$ controls and the total number of $z$-rotations is instead
    \begin{equation}
        (2^{n-1}-1)\cdot\left(2\sum_{m=0}^{n-1}2^m\right)=2(2^n-1)(2^{n-1}-1)
    \end{equation}
\end{enumerate}
Putting everything together, the number of rotations expressed by the second of equations~(\ref{eqn:gate_counts_general}) is proven.

Finally, the last of equations~(\ref{eqn:gate_counts_general}) concerning the reduction in terms of CNOT gates with respect to the original structure proposed in~\cite{sarkar2024quantum} is easily derived from the contributions already computed:
\begin{equation}
\begin{split}
    CNOT_{red}&=N_{CNOT}(\mbox{simplif. }Z)+N_{CNOT}\left(\mbox{simplif. }\prod T_x^e\right)+N_{CNOT}\left(\mbox{simplif. }\prod T_x^o\right)=\\
    &=2^n(n-3)+4+2[(n-3)2^{n-1}+2]+2^n-2n=\\
    &=2^n(2n-5)-2n+8
\end{split}
\end{equation}
valid from $n\geqslant3$.

\subsubsection{The special case n=2}\label{sec:anomaly_n2}
Section~\ref{sec:srbb} and Section~\ref{sec:2qubits} provide all the elements to understand why the case $n=2$ represents a particular instance of the scalability pattern. There are essentially two reasons: the $M_1^o$ sub-factor is a $M_2ZYZ$-type matrix only for $n=2$, and only with 2 qubits does the $Z(\Theta_Z)$-factor admit simplifications with the $\Psi(\Theta_\Psi)$-factor. In the first case, the exception to scalability is from an implementation perspective, which means that from $n=3$ onwards, the $M_x^o$ sub-factors can all be implemented as block diagonal unitary matrices with unitary blocks (not $SU(2)$ blocks as for $M_nZYZ$-type matrices). Their implementation requires a different pattern, which starts and ends with a rotation gate, preventing simplifications with the adjacent CNOTs of other sub-factors~\cite{10821064}. In the second case, the explanation is simpler: by implementing both $\Psi(\Theta_\Psi)$- and $Z(\Theta_Z)$-factors through cyclic Gray Code\footnote{The first because it is composed of $M_nZYZ$ matrices and the second to reach the smallest number of CNOTs, as previously described.}, $\Psi(\Theta_\Psi)$ ends with a $\mbox{CNOT}_{(1,n)}$, where the first subscript stands for control, while $Z(\Theta_Z)$ starts with a $\mbox{CNOT}_{(n-1,n)}$, preventing simplification from $n\geqslant3$~\cite{10821064}.

To complete the analysis of this particular case, Equations~(\ref{eqn:gate_counts_n2}) are demonstrated below. Since the underlying approximation algorithm is expressed by the same key Equation~(\ref{eqn:U_approx}), for each of the Equations~(\ref{eqn:gate_counts_n2}) the same contributions of the $n\geqslant3$ cases were taken into consideration. What changes are the algebraic properties of these contributions; therefore, their implementation and the consequent possible simplifications.
\begin{enumerate}
    \item[$\square$]$N_{CNOT}(Z)$\\In this case, the formula $2^n-2=2$ counts the precise number of CNOT gates \emph{before} the simplification procedure (and not after as for $n\geqslant3$ cases).
    \item[$\square$]$N_{CNOT}(1^\circ\,\mbox{subfactor of }\Psi)+N_{CNOT}(M_x^{e/o})$\\As explained above, only for $n=2$ the $M_1^o$ factor has a complete $ZYZ$-decomposition too, so that the structure of the overall circuit is characterized by sequences of $ZYZ$-type matrices. In addition to the $1^\circ$ sub-factor of $\Psi$, the circuit counts $2^{n-1}-1=1$ contribution for the even case and $2^{n-1}-1=1$ contribution for the odd case. Thus, a total of $(2^n-1)(3\cdot2^{n-1}-2)=12$.
    \item[$\square$]$N_{CNOT}\left(\prod T_x^{e/o}\right)$\\For both the even and odd cases, the same trends obtained for $n\geqslant3$ apply to these factors: $2(n-1)2^{n-2}=2$ for the even case and $2[(n-1)2^{n-2}+2(2^{n-1}-1)]=6$ for the odd case.
    \item[$\square$]$CNOT_{red}$\\From the circuit illustrated in Figure~\ref{fig:n2circuit}, it is possible to count the 4 simplified CNOTs for this particular case. 
\end{enumerate}
Putting everything together, the number of CNOT gates expressed by the first of Equations~(\ref{eqn:gate_counts_n2}) is proven.

As regards the total number of rotations in the $n=2$ case, the calculation is very simple given the repetition of $2^n-1$ $ZYZ$-decompositions in addition to $2^n-1$ diagonal elements:
\begin{equation}
    (3\cdot2^{n-1})(2^n-1)+2^n-1=21
\end{equation}

\section{Implementation and Results}\label{sec:impl}
The proposed scalable VQC has been implemented using the PennyLane library~\cite{bergholm2018pennylane}. The majority of the tests discussed in this section have been performed using the simulator provided by the same library, executed on a Linux machine equipped with an AMD EPYC 7282 CPU and 256 GB RAM. A subset of experiments, detailed in Section~\ref{sec:RealHW}, was instead carried out on IBM quantum computers, namely IBM Brisbane~\cite{ibm_quantum_brisbane} and IBM Fez~\cite{ibm_quantum_fez}.

The QNN is utilized to approximate predefined ideal circuits, from which the associated unitary matrices are retrieved, and random unitary operators.
Since the approximating algorithm operates with special unitary matrices, the set of $d=2^n$ $SU(2^n)$ matrices associated with each ideal circuit must be computed by dividing the ideal unitary matrix by the $d$-th roots of its determinant, as previously explained in Section~\ref{sec:problem_phase}.

The parameters of the VQC are updated to minimize the cost function. In this work, different losses are tested:
\begin{itemize}
    \item Frobenius norm: $\|A \|_{F} = \sqrt{\operatorname{tr}(A A^\dagger)}$, where $A = SU_{ideal} - SU_{VQC}$
    \item Trace distance: $
            \| \rho - \sigma \|_1 = \operatorname{tr}\sqrt{(\rho - \sigma)^\dagger(\rho - \sigma)}$
    \item Fidelity: $
            \operatorname{F}(\rho, \sigma) = \left(\operatorname{tr}\sqrt{\sqrt{\rho}\sigma\sqrt{\rho}}\right)^2$
\end{itemize}
In the previous definitions, $\rho$ and $\sigma$ are the density matrices of the states obtained by the ideal $SU$ evolution, which will be approximated, and by the VQC, respectively, given the same input state.

When the Frobenius norm is used as a loss function, the VQC is constructed to obtain its matrix representation and the Frobenius distance is calculated to get the loss value.
The Frobenius norm is applied to compare the $SU(2^n)$ matrix associated with the circuit that the QNN should approximate and the $SU(2^n)$ matrix associated with the VQC. 
Instead, when the fidelity or the trace distance are used, a training set with $1000$ random states is created to perform the minimization of the loss function, in order to decrease (increase) the value of the trace (fidelity) of the states obtained by the ideal matrix (direct algebraic computation) and the approximated matrix (synthesis process).

However, a problem arises when the network is trained using fidelity or trace distance. Since density matrices make use of the square moduli of the amplitudes, the contribution of the phases that establishes the relationship between the special unitary group and the unitary one disappears, as previously explained in Section~\ref{sec:problem_phase}.
Thus, to recover the phase contribution, the following procedure has been implemented:
\begin{itemize}
    \item the $2^n$ $\mathbb{C}$-roots of the determinant are calculated;
    \item the roots are associated with the $2^n$ possible $SU(2^n)$ matrices;
    \item after training, the distance between the approximated $SU(2^n)$ matrix returned by the network and each possible $SU(2^n)$ matrix is computed;
    \item finally, the approximated $U(2^n)$ matrix is computed by multiplying the approximated $SU(2^n)$ matrix by the root associated with the ideal one that has the smallest distance. 
\end{itemize}

The samples of the dataset are mapped to quantum states by means of amplitude encoding. This scheme encodes a classical data vector $x$ in a quantum state by associating the normalized features with the probability amplitudes of the quantum state. Given $x = (x_1, x_2, ..., x_N)^T$, with $N = 2^n$, the quantum state is as follows:
\begin{equation}
    U_\phi(x): x \in \mathbb{R}^N \longrightarrow |\phi(x)\rangle = \frac{1}{||x||}\sum_{i = 1}^{N}{x_i|i\rangle},
\end{equation}
where $|i\rangle$ is the $i$-th computational basis.

The optimization is performed using two different optimizers: 
\begin{itemize}
    \item Adam, which is a classical optimizer based on gradient descent; the learning rate is $0.01$, and when the trace or fidelity is used, the number of epochs is $20$ and the batch size is $64$;
    \item Nelder Mead, which is used only with the Frobenius loss.
\end{itemize}

\subsection{Tests with Quantum Circuit Simulation}
\label{sec:simulation}

The QNN was tested with different numbers of qubits, from 2 to 6. Initially, it was tested with predefined circuits, which are depicted in Appendix~\ref{circuitsTested}.
Then, the QNN was also tested with dense random $SU(2^n)$ matrices. The results obtained with Adam are shown in Table~\ref{tab:results}, where "circuits" indicates the predefined circuits, while "random" indicates the $SU(2^n)$ matrices that are generated randomly. The execution time was measured with respect to the Frobenius loss.
For 6 qubits, the results of fidelity and trace are not included because the algorithm takes days to complete.
In general, the algorithm achieved a low value for the loss function up to 5 qubits, indicating suitable approximations of the ideal matrices. The results show that the precision decreases as the number of qubits increases. This may result from the optimization process and the increased number of gates, which can create a more complex parametric landscape with additional local minima. Furthermore, for the same number of qubits, the performance worsens when using random matrices. This suggests that having sparse matrices or matrices with specific patterns simplifies the approximation for the QNN. In any case, the best performance was always achieved after just 10 epochs and with only one single layer of the QNN.

\begin{table}[htbp]
\centering
\resizebox{0.6\textwidth}{!}{
    \begin{tabular}{ |c|c||c|c|c|c|  }
     \hline
     \multicolumn{6}{|c|}{Loss results} \\
     \hline
     \hline
     n & Ideal matrix &Trace & Fidelity & Frobenius & Execution time \\
     \hline
     2 & circuits & $10^{-3}$ &$10^{-8}$ &$10^{-3}$ &$<8 s$ \\
     \hline
     2 & random & $10^{-3}$  & $10^{-5}$ &$10^{-3}$ &$<6 s$\\
     \hline
     3 & circuits & $10^{-3}$  & $10^{-7}$ & $10^{-2}$ &$<40 s$\\
     \hline
     3 & random  &  $10^{-2}$& $10^{-3}$ & $10^{-1}\sim$~\newline$10^{-2}$ &$<$~$75 s$\\
     \hline
     4 & circuits &    $10^{-2}$  & $10^{-5}$ & $10^{-2}$ & $<$~$227 s$\\
     \hline
     4 & random    &  $0.13$    &$10^{-2}$ & $0.28$&$<$~$222 s$\\
     \hline
     5 & circuits  &  $10^{-2}$& $10^{-3}$ & $0.3$ &$<$~$21 m$\\
     \hline
     5 & random  & $0.19$& $10^{-2}$ & $0.67$&$<$~$25 m$\\
     \hline
     6 & circuits &    $-$  & $-$ & $0.7$ &$<$~$2 h$\\
     \hline
     6 & random &    $-$  & $-$ & $1.3$&$<$~$2 h$\\
     \hline
    \end{tabular}}
\caption{Loss results divided according to the number of qubits and the ideal matrix used. }
\label{tab:results}
\end{table}

An example of the execution with the Adam optimizer of an approximate predefined circuit and an approximate random unitary matrix is provided for 2 qubits in appendix~\ref{a_exampleMatrices}.

In Table~\ref{tab:compare2qubits}, the results for 2-qubit circuits are presented, comparing those of our method, those from~\cite{sarkar2023scalable} and those obtained using~\cite{vidal2004universal} with our method (QNN optimized by using Adam or Nelder Mead). When the Adam optimizer was used, although the network achieved a lower accuracy compared to the results reported by~\cite{sarkar2023scalable}, it demonstrated significantly reduced training times. This is a notable outcome, as it can help to obtain a coarse approximation of larger matrices in a reasonable time. Furthermore, when the Nelder Mead optimizer was used (as in~\cite{sarkar2023scalable}), the QNN achieved better results in both approximation error and time taken. The network can produce significant results in a short amount of time. Moreover, our method was also applied to the two-qubit optimal circuit of Vidal~\cite{vidal2004universal}, where it achieved slightly worse performance. Additionally, the last row shows the results obtained with a random unitary. The network is able to achieve notable results in about 1 minute with the Nelder-Mead optimizer and, in this case, the performance achieved using the circuit introduced in~\cite{vidal2004universal} is much worse.

\begin{table}[htbp]
    \centering
    \resizebox{1\textwidth}{!}{
    \begin{tabular}{|>{\centering\arraybackslash}m{1.6cm}||>{\centering\arraybackslash}m{1.4cm}|>{\centering\arraybackslash}m{1.4cm}|>{\centering\arraybackslash}m{1cm}|>{\centering\arraybackslash}m{1.2cm}|>{\centering\arraybackslash}m{2.3cm}|>{\centering\arraybackslash}m{2cm}|>{\centering\arraybackslash}m{2cm}|>{\centering\arraybackslash}m{1.5cm}|}
        \hline
        Ideal matrix & Time taken by our method + Adam & Time taken by our method + NM & Time taken in~\cite{sarkar2023scalable} & Error from our method + Adam & Error from our method + NM & Error from~\cite{sarkar2023scalable} & Error from~\cite{vidal2004universal} + our method + Adam & Error from~\cite{vidal2004universal} + our method + NM \\
        \hline
        CNOT & \multirow{18}{*}{$8s$} & \multirow{18}{*}{$11\sim$~$24s$}& $90s$ & \multirow{18}{*}{$10^{-3}$} & \multirow{18}{*}{$10^{-15}\sim$~$10^{-16}$} & $7.977 \times 10^{-14}$ & \multirow{18}{*}{$10^{-1}\sim$~$10^{-3}$} & \multirow{18}{*}{$10^{-15}$}\\
        \cline{1-1} \cline{4-4} \cline{7-7}
        Grover$_2$ & & & $124s$ & & & $1.256 \times 10^{-15}$ & &\\
        \cline{1-1} \cline{4-4} \cline{7-7}
        XX & & & $20s$ & & & $6.226 \times 10^{-12}$ & &\\
        \cline{1-1} \cline{4-4} \cline{7-7}
        YY & & & $240s$ & & & $3.223 \times 10^{-15}$ & &\\
        \cline{1-1} \cline{4-4} \cline{7-7}
        ZZ & & & $90s$ & & & $1.363 \times 10^{-17}$ & &\\
        \cline{1-1} \cline{4-4} \cline{7-7}
        SWAP & & & $63s$ & & & $1.839 \times 10^{-13}$ & &\\
        \cline{1-1} \cline{4-4} \cline{7-7}
        XZ & & & $150s$ & & & $3.580 \times 10^{-13}$ & &\\
        \cline{1-1} \cline{4-4} \cline{7-7}
        ZX & & & $129s$ & & & $5.438 \times 10^{-13}$ & &\\
        \cline{1-1} \cline{4-4} \cline{7-7}
        ZY & & & $121s$ & & & $3.188 \times 10^{-12}$ & &\\
        \cline{1-1} \cline{4-4} \cline{7-7}
        CNOT(2, 1) & & & $45s$ & & & $1.058 \times 10^{-13}$ & &\\
        \cline{1-1} \cline{4-4} \cline{7-7}
        DCNOT & & & $29s$ & & & $4.020 \times 10^{-13}$ & &\\
        \cline{1-1} \cline{4-4} \cline{7-7}
        XNOR & & & $23s$ & & & $3.166 \times 10^{-13}$ & &\\
        \cline{1-1} \cline{4-4} \cline{7-7}
        iSWAP & & & $183s$ & & & $3.003 \times 10^{-14}$ & &\\
        \cline{1-1} \cline{4-4} \cline{7-7}
        fSWAP & & & $93s$ & & & $2.037 \times 10^{-13}$ & &\\
        \cline{1-1} \cline{4-4} \cline{7-7}
        C-Phase & & & $15s$ & & & $7.666 \times 10^{-15}$ & &\\
        \cline{1-1} \cline{4-4} \cline{7-7}
        XX + YY & & & $124s$ & & & $1.665 \times 10^{-12}$ & &\\
        \cline{1-1} \cline{4-4} \cline{7-7}
        $\sqrt{\text{SWAP}}$ & & & $97s$ & & & $1.686 \times 10^{-13}$ & &\\
        \cline{1-1} \cline{4-4} \cline{7-7}
        $\sqrt{\text{iSWAP}}$ & & & $10s$ & & & $1.106 \times 10^{-13}$ & &\\
        \cline{1-1} \cline{4-4} \cline{7-7}
        QFT$_2$ & & & $31s$ & & & $3.215 \times 10^{-13}$ & &\\
        \hline
        random & $6s$ & $70s$ & $-$ & $10^{-3}$ & $10^{-14}$ & $-$ & $10^{-1}$ & $10^{-2}$ \\
        \hline
    \end{tabular}}
    \caption{Comparison between different methods, with 2-qubit circuits and Frobenius loss. NM stands for Nelder-Mead optimizer.}
    \label{tab:compare2qubits}
\end{table}

In Table~\ref{tab:compare3qubits}, a comparison of the results for 3-qubit circuits is presented, considering those obtained with our method, those from~\cite{sarkar2023scalable} and those from~\cite{younis2020qfast}. Also in this case, the network achieved superior results with Nelder-Mead. The training times with Nelder-Mead are significantly higher compared to the 2-qubit scenario, but the network can achieve better results than the other two methods. However, when Adam is used, the training time remains short, providing a fast way to achieve a coarse approximation of circuits.

\begin{table}[htbp]
    \centering
    \resizebox{1\textwidth}{!}{
    \begin{tabular}{|>{\centering\arraybackslash}m{1.5cm}||>{\centering\arraybackslash}m{1.2cm}|>{\centering\arraybackslash}m{1.2cm}|>{\centering\arraybackslash}m{1.2cm}|>{\centering\arraybackslash}m{2.3cm}|>{\centering\arraybackslash}m{2cm}|>{\centering\arraybackslash}m{2cm}|>{\centering\arraybackslash}m{2cm}|}
        \hline
        Ideal matrix & Time taken by our method + Adam & Time taken by our method + NM & Error from our method + Adam & Error from our method + NM & Error from \cite{sarkar2023scalable} & QFAST + KAK \cite{younis2020qfast} & UniversalQ \cite{younis2020qfast} \\
        \hline
        Toffoli & \multirow{4}{*}{$40s$} & \multirow{4}{*}{$<1h$} & \multirow{4}{*}{$10^{-2}$} & \multirow{4}{*}{$10^{-10}\sim$~$10^{-11}$} & $4.48 \times 10^{-9}$ & $1.5 \times 10^{-6}$ & $2.6 \times 10^{-8}$ \\
        \cline{1-1} \cline{6-8}
        Fredkin & & & & & $1.6 \times 10^{-8}$& $2.2 \times 10^{-6}$ & $0$ \\
        \cline{1-1} \cline{6-8}
        Grover$_3$ & & & & & $4.60 \times 10^{-9}$& $8.1 \times 10^{-7}$ & $0$ \\
        \cline{1-1} \cline{6-8}
        Peres & & & & & $2 \times 10^{-8}$& $6.8 \times 10^{-7}$ & $2.1 \times 10^{-8}$ \\
        \cline{1-1} \cline{6-8}
        QFT$_3$ & & & & & $3.1 \times 10^{-9}$& $3 \times 10^{-7}$ & $3 \times 10^{-8}$ \\
        \hline
    \end{tabular}}
    \caption{Comparison between different methods, with 3-qubit circuits and Frobenius loss. NM stands for Nelder-Mead optimizer.}
    \label{tab:compare3qubits}
\end{table}

In Table~\ref{tab:compare4qubits}, a comparison of the results for 4-qubit circuits is presented, considering those obtained with our method, those from~\cite{sarkar2023scalable} and those from~\cite{younis2020qfast}. The network achieved the same or even better accuracy than the other methods, except for the Grover algorithm, but its training time does not scale efficiently. Using the Nelder-Mead optimizer, given sufficient time, it can achieve the same performance as the other methods; therefore, the focus should be on comparing the training time required to reach that approximation.

\begin{table}[htbp]
    \centering
    \resizebox{1\textwidth}{!}{
    \begin{tabular}{|>{\centering\arraybackslash}m{1.6cm}||>{\centering\arraybackslash}m{1.5cm}|>{\centering\arraybackslash}m{1.5cm}|>{\centering\arraybackslash}m{1.5cm}|>{\centering\arraybackslash}m{1.5cm}|>{\centering\arraybackslash}m{2cm}|>{\centering\arraybackslash}m{2cm}|>{\centering\arraybackslash}m{2cm}|}
        \hline
        Ideal matrix & Time taken by our method + Adam & Time taken by our method + NM & Error from our method + Adam & Error from our method + NM & Error from \cite{sarkar2023scalable} & QFAST + KAK \cite{younis2020qfast} & UniversalQ \cite{younis2020qfast} \\
        \hline
        $C_1C_1C_1X$ & \multirow{4}{*}{$227s$} & $<2 weeks$ & \multirow{4}{*}{$10^{-1}$} & $10^{-9}$ & $1.97 \times 10^{-8}$ & $2.2 \times 10^{-5}$ & $1.3 \times 10^{-6}$ \\
        \cline{1-1} \cline {3-3} \cline{5-8}
        Grover$_4$ & & $\sim3 weeks$ & & $10^{-5}$ & $2.12 \times 10^{-9}$& $-$ & $-$ \\
        \cline{1-1} \cline {3-3} \cline{5-8}
        QFT$_4$ & & $<2 weeks$ & & $10^{-9}$ & $9.331 \times 10^{-8}$& $-$ & $-$ \\
        \hline
    \end{tabular}}
     \caption{Comparison between different methods, with 4-qubit circuits and Frobenius loss. NM stands for Nelder-Mead optimizer.}
    \label{tab:compare4qubits}
\end{table}

In Table~\ref{tab:diamondResults}, the diamond norm results achieved when the QNN is trained using the Frobenius norm are presented.
The diamond norm measures the difference between two quantum channels and is defined as~\cite{regula2021operational}
\begin{equation}
    ||\Phi||_\diamond = \max_{\rho \in \mathbb{D}(A \otimes A)} ||I_A \otimes \Phi(\rho)||_1,
\end{equation}
where $\mathbb{D}(A \otimes A)$ denotes the states acting in a bipartite Hilbert space composed of the space $A$ and another space isomorphic thereto. The diamond norm can be computed efficiently via semidefinite programming~\cite{watrous2012simpler}.
The results demonstrate the effectiveness of the Frobenius norm as a training objective for approximating unitary matrices, as it yields strong performance when evaluated in terms of the diamond norm.

\begin{table}[htbp]
    \centering
    \begin{tabular}{|c||c|c|}
        \hline
        N qubit & Error with NM & Diamond norm error \\
        \hline
        2 & $10^{-15}\sim$~$10^{-16}$& $10^{-27}$\\
        \hline
        3 & $10^{-10}\sim$~$10^{-11}$ &$10^{-24}\sim$~$10^{-25}$\\
        \hline
    \end{tabular}
     \caption{Diamond norm results, when the QNN is trained with the frobenius norm.}
    \label{tab:diamondResults}
\end{table}

Table~\ref{tab:opFid} presents the results when the QNN is trained using the operator fidelity loss, an operator norm that is defined as follows:
\begin{equation}
    F(U, V) = \frac{1}{d^2}|Tr(U^{\dagger}V)|^2
\end{equation}
where $U$ and $V$ are unitary operators.
The results highlight that this loss can achieve better outcomes than fidelity and trace distance, but its performance is worse than that of the frobenius norm; it is more volatile.

\begin{table}[htbp]
    \centering
    \begin{tabular}{|c||c|}
        \hline
        N qubit & Operator fidelity error \\
        \hline\
        2 & $10^{-6}\sim10^{-13}$\\
        \hline
        3 & $10^{-3}\sim10^{-6}$\\
        \hline
        4 & $10^{-3}$\\
        \hline
    \end{tabular}
    \caption{Results achieved when the QNN is trained via operator fidelity loss.}
    \label{tab:opFid}
\end{table}

Table~\ref{tab:resultsEvolution} presents the results of the verification of the density matrix evolution theorem through the Liouville-Von Neumann equation:
\begin{equation}
    \frac{\partial\rho}{\partial t} = \frac{-i}{\hbar} [H, \rho]
\end{equation}
whose solution is:
\begin{equation}\label{eq:Evolution}
    \rho(t) = U\rho(0)U^\dagger
\end{equation}

From an implementation viewpoint, equation~(\ref{eq:Evolution}) is used to compute the ideal density matrix resulting from any ideal circuit (Appendix~\ref{circuitsTested}), where $\rho(0)$ is the amplitude encoding of the input state.
The latter is compared to the approximate density matrix produced by the QNN via the trace distance.

\begin{table}[htbp]
\centering
\resizebox{0.5\textwidth}{!}{
    \begin{tabular}{ |c|c||c|c|c|  }
     \hline
     \multicolumn{5}{|c|}{Evolution} \\
     \hline
     \hline
     n & Ideal matrix &Trace & Fidelity & Frobenius \\
     \hline
     \multirow{2}{*}{2} & circuits & $10^{-3}$ &$10^{-6}$ &$10^{-3}$ \\
     \cline{2-5}
     & random & $10^{-3}$  & $10^{-3}$ &$10^{-3}$\\
     \hline
     \multirow{2}{*}{3} & circuits & $10^{-3}$  & $10^{-5}$ & $10^{-3}$\\
     \cline{2-5}
     & random  &  $10^{-2}$& $10^{-2}$ & $10^{-3}$\\
     \hline
     \multirow{2}{*}{4} & circuits &    $10^{-3}$  & $10^{-5}$ & $10^{-2}$ \\
     \cline{2-5}
     & random    &  $10^{-1}$    &$10^{-1}$ & $10^{-2}$\\
     \hline
     \multirow{2}{*}{5} & circuits  &  $10^{-2}$& $10^{-2}$ & $10^{-2}$ \\
     \cline{2-5}
     & random  & $0.17$& $0.16$ & $0.12$\\
     \hline
     \multirow{2}{*}{6} & circuits &    $-$  & $-$ & $10^{-2}$ \\
     \cline{2-5}
     & random &    $-$  & $-$ & $0.12$\\
     \hline
    \end{tabular}}
    \caption{Density matrix evolution results divided based on the number of qubits and the ideal matrix used.}
    \label{tab:resultsEvolution}
\end{table} 

In Figure~\ref{fig:testError}, the error on a test set of 500 samples is plotted for both the Adam and Nelder-Mead optimizers. When Fidelity or Trace Distance is used, the test error is computed between the output state produced by the trained QNN and the ideal state obtained from the ideal circuit that the network should approximate. The plot shows how the fidelity can produce better results with a low number of qubits, whereas, as the number of qubits increases, the three loss functions behave in a similar manner. The Nelder–Mead trend yields better results than those obtained with Adam; however, starting from 4 qubits, the training takes days.

Furthermore, the 2-qubit QNN is evaluated using two layers to analyze the effect of increasing the circuit depth. The QNN is tested using two optimizer-loss combinations: Adam with the fidelity loss and Nelder Mead with the Frobenius loss. When using the Nelder Mead optimizer, the QNN achieved the same performance as in the single-layer case, reaching a test error of $10^{-15}$ on the $500$-sample test set. In contrast, when using the Adam optimizer, as shown in Figure~\ref{fig:testError}, the single-layer QNN achieved a test error of $10^{-12}$, while the two-layer QNN reached a slightly lower error of $10^{-13}$. However, this marginal improvement comes at the cost of doubling the number of trainable parameters. Therefore, increasing the number of layers provides only a limited performance gain when using Adam and no improvement when using Nelder Mead, likely because the latter already attains the best achievable approximation with a single layer.

\begin{figure}[htbp]
    \centering
    \begin{tikzpicture}
        \begin{axis}[
            xlabel={Number of Qubits},
            ylabel={Test Set Error},
            ymode=log,
            grid=both,
            minor tick num=1,
            major grid style={line width=.2pt,draw=gray!50},
            minor grid style={line width=.1pt,draw=gray!10},
            width=10cm,
            height=6cm,
            legend pos=outer north east,
            legend style={font=\small},
            ytick={1e-16, 1e-10, 1e-8, 1e-5, 1e-3, 1e-2, 1e-1},
            yticklabels={\(10^{-16}\),\(10^{-10}\), \(10^{-8}\), \(10^{-5}\), \(10^{-3}\), \(10^{-2}\), \(10^{-1}\)},
            xtick={2, 3, 4, 5, 6}
        ]
        
        \addplot[
            color=blue,
            mark=*,
            line width=1.5pt,
            smooth
        ] coordinates {
            (2, 1e-12)
            (3, 1e-4)
            (4, 1e-3)
            (5, 1e-2)
        };
        \addlegendentry{Fidelity (Adam)}

        \addplot[
            color=red,
            mark=square*,
            line width=1.5pt,
            smooth
        ] coordinates {
            (2, 1e-3)
            (3, 1e-3)
            (4, 1e-2)
            (5, 1e-2)
        };
        \addlegendentry{Trace (Adam)}

        \addplot[
            color=green,
            mark=triangle*,
            line width=1.5pt,
            smooth
        ] coordinates {
            (2, 1e-3)
            (3, 1e-3)
            (4, 1e-2)
            (5, 1e-2)
            (6, 1e-1)
        };
        \addlegendentry{Frobenius (Adam)}

        \addplot[
            color=purple,
            mark=*,
            line width=1.5pt,
            dotted,
            smooth
        ] coordinates {
            (2, 1e-16)
            (3, 1e-8)
            (4, 1e-2)
        };
        \addlegendentry{Frobenius (Nelder-Mead)}
        \end{axis}
    \end{tikzpicture}
    \caption{Test set error of the trained QNN with the three loss functions.}
    \label{fig:testError}
\end{figure}

\subsection{Tests on Real Quantum Devices}
\label{sec:RealHW}

In Table~\ref{tab:hellinger}, the results obtained from tests performed using the IBM Brisbane quantum computer~\cite{ibm_quantum_brisbane} are shown. In particular, the Hellinger distance between the ideal probability distribution and the approximated one produced by the QNN is computed, when the proposed scalable VQC is applied to the state $|0\rangle^{\otimes n}$ with 1024 shots.

\begin{table}[htbp]
\centering
    \begin{tabular}{ |c|c|  }
     \hline
     n  & Hellinger distance \\
     \hline
     2 & $\sim$~$0.15$ \\
     \hline
     3 & $\sim$~$0.30$\\
     \hline
     4 & $\sim$~$0.35$\\
     \hline
    \end{tabular}
    \caption{Hellinger distance between the ideal probability and the one obtained using the IBM Brisbane quantum computer when the VQC is applied to the state $|0\rangle^{\otimes n}$ with 1024 shots.}
    \label{tab:hellinger}
\end{table}

Table~\ref{tab:timeRealHW} shows the execution time of a circuit on IBM Brisbane. When the number of qubits is small, the execution times for 1024 and 4096 shots are similar. However, when the number of qubits increases, the time taken with 4096 shots grows significantly.

\begin{table}[htbp]
\centering
    \begin{tabular}{ |c|c|c|  }
     \hline
     n  & Time [s] for 1024 shots & Time [s] for 4096 shots \\
     \hline
     2 & $\sim$~$3$ & $\sim$~$3$ \\
     \hline
     3 & $\sim$~$3$ & $\sim$~$3$\\
     \hline
     4 & $\sim$~$4$ & $\sim$~$4$\\
     \hline
     5 & $\sim$~$5-6$ & $\sim$~$11$\\
     \hline
     6 & $\sim$~$11$ & $\sim$~$40-45$\\
     \hline
    \end{tabular}
    \caption{Time in seconds taken for the execution on using the IBM Brisbane quantum computer when the VQC is applied to the state $|0\rangle^{\otimes n}$.}
    \label{tab:timeRealHW}
\end{table}

In Table~\ref{tab:cnotHW}, the fidelity results on IBM quantum computers Brisbane and Fez \cite{ibm_quantum_brisbane, ibm_quantum_fez} are presented, when the QNN is trained with Adam to approximate a CNOT gate. In this case, the test set consists of only 10 random states, which limits its generalizability. However, it still achieved a good approximation.
An interesting observation arises when comparing the compiled circuits of the QNN trained with Adam and Nelder-Mead on real IBM hardware. On IBM Brisbane, the amplitude-encoded input state and Nelder-Mead QNN circuits are compiled into approximately 45 gates, while the Adam-trained circuits are compiled into around 60 gates. This likely occurs because the circuit produced by Nelder-Mead is more accurate, allowing the compiler to better recognize the ideal gate, rather than the sequence of RZ, RY, and CNOT gates used in the scalable algorithm.

\begin{table}[htbp]
\centering
    \begin{tabular}{ |c|c|c|  }
     \hline
     Device & Time [s] for 1024 shots & Hellinger distance \\
     \hline
     IBM Brisbane & $\sim$~$3$ & $\sim$~$0.06$ \\
     \hline
     IBM Fez & $\sim$~$6$ & $\sim$~$0.07$ \\
     \hline
    \end{tabular}
    \caption{Results of an approximated CNOT on real IBM hardware.}
    \label{tab:cnotHW}
\end{table}

\subsection{Insights on loss functions}

From an application viewpoint, the interpretations of the loss functions vary. When using the Frobenius loss, the $\|A \|_{F}$ compares the two matrices algebraically, allowing it to find the ideal matrix. On the other hand, when using the trace distance or fidelity losses, these metrics compare the density matrices of the output states. Therefore, there may be scenarios where the QNN produces the correct output probability distribution when measured, but the approximated matrix differs from the ideal one. For instance, in Appendix~\ref{circuitsTested}, the circuit $C_0^1C_1^0RY(\frac{\pi}{4})$ for 3 qubits shows discrepancies in some entries of the approximated matrix compared to the ideal one, as shown below.

Ideal unitary matrix \( U \):~
\begin{center}
\scriptsize
\[
\begin{bmatrix}
1+0j & 0j & 0j & 0j & 0j & 0j & 0j & 0j \\
0j & 1+0j & 0j & 0j & 0j & 0j & 0j & 0j \\
0j & 0j & 1+0j & 0j & 0j & 0j & 0j & 0j \\
0j & 0j & 0j & 1+0j & 0j & 0j & 0j & 0j \\
0j & 0j & 0j & 0j & \textbf{0.92388+0j} & \textbf{-0.38268-0j} & 0j & 0j \\
0j & 0j & 0j & 0j & \textbf{0.38268+0j} & \textbf{0.92388+0j} & 0j & 0j \\
0j & 0j & 0j & 0j & 0j & 0j & 1+0j & 0j \\
0j & 0j & 0j & 0j & 0j & 0j & 0j & 1+0j
\end{bmatrix}
\]
\end{center}

Approximated unitary matrix \( U \) (Fidelity/Trace):\\\\
\resizebox{1\textwidth}{!}{\(
\begin{bmatrix}
1 - 10^{-5}j & 10^{-5} + 10^{-5}j & 10^{-5} + 10^{-5}j & 10^{-5} + 0j & 10^{-5}j & 10^{-5} - 10^{-5}j & 10^{-5} + 0j & 10^{-5} + 10^{-5}j \\
10^{-5} + 10^{-5}j & 1 + 10^{-5}j & 10^{-5} + 10^{-5}j & 10^{-5} + 10^{-5}j & -10^{-5} - 10^{-5}j & 10^{-5} + 0j & 10^{-5} + 0j & 10^{-5} + 10^{-5}j \\
-10^{-5} + 10^{-5}j & -10^{-5} + 10^{-5}j & 1 - 10^{-5}j & -10^{-5} - 10^{-5}j & -10^{-4} + 10^{-5}j & -10^{-4} - 10^{-4}j & 0 & 10^{-5} - 10^{-5}j \\
-10^{-5} + 0j & -10^{-5} + 10^{-5}j & 10^{-5} - 10^{-5}j & 1 + 10^{-5}j & 10^{-5} + 10^{-4}j & -10^{-5} - 10^{-5}j & 10^{-5} - 10^{-5}j & 10^{-5} - 0j \\
10^{-5} + 10^{-5}j & 10^{-5} - 10^{-5}j & -10^{-5} + 10^{-4}j & -10^{-5} + 10^{-4}j & \textcolor{red}{0.70722 - 10^{-4}j} & \textcolor{red}{-0.70699 + 10^{-4}j} & -0 +10^{-4}j & -10^{-4}j \\
-10^{-5} - 10^{-5}j & 10^{-5} - 10^{-5}j & 10^{-4} - 10^{-4}j & -10^{-5} + 10^{-4}j & \textcolor{red}{0.70699 + 10^{-4}j} & \textcolor{red}{0.70722 + 10^{-4}j} & 10^{-5} - 10^{-4}j & -10^{-5} - 10^{-4}j \\
10^{-5} + 0j & 10^{-5} + 0j & 0 & -10^{-5} - 10^{-5}j & 0 + 10^{-5}j & -10^{-5} - 10^{-4}j & 1 + 0j & 0 + 10^{-5}j \\
10^{-5} + 10^{-5}j & -10^{-5} + 10^{-5}j & -10^{-5} - 10^{-5}j & -10^{-5} + 0j & 10^{-5} - 10^{-4}j & 10^{-5} + 0j & 10^{-5}j & 1 - 0j
\end{bmatrix}
\)}\vspace{0.5cm}

On the other hand, when the Frobenius norm is used as loss function, the neural network can find a correct approximation of the unitary matrices. For instance, the approximated unitary matrix for $C_0^1C_1^0RY(\frac{\pi}{4})$ will result:\\\\
\resizebox{1\textwidth}{!}{\(
\begin{bmatrix}
0.99989 + 0.01376j & 10^{-5} - 0j & -10^{-3} + 10^{-3}j & 10^{-5}j & -10^{-4} + 10^{-5}j & 10^{-5} - 10^{-5}j & -10^{-4} - 10^{-4}j & 10^{-4} - 10^{-5}j \\
-10^{-5} - 0j & 0.99999 - 0.00215j & -10^{-5} + 0j & 10^{-3} - 10^{-3}j & 10^{-5} - 10^{-5}j & 10^{-4} + 10^{-4}j & 10^{-5} - 10^{-4}j & -10^{-4} - 10^{-4}j \\
10^{-3} + 10^{-3}j & 10^{-5} + 0j & 0.9999 - 0.01262j & -10^{-5}j & 10^{-4} + 10^{-4}j & -10^{-4} + 0j & 10^{-4} - 10^{-4}j & 10^{-4} + 10^{-4}j \\
10^{-5}j & -10^{-3} - 10^{-3}j & -10^{-5}j & 0.99999 + 10^{-4}j & -10^{-4} + 10^{-5}j & 10^{-4} + 10^{-4}j & 10^{-5} + 10^{-5}j & 10^{-4} + 10^{-4}j \\
10^{-4} + 10^{-5}j & 10^{-5} - 10^{-5}j & -10^{-4} + 10^{-4}j & 10^{-4} - 10^{-4}j & \textcolor{red}{0.92386 + 0.00185j} & \textcolor{red}{-0.38273 - 10^{-5}j} & 10^{-4} + 10^{-5}j & -10^{-4} + 10^{-4}j \\
-10^{-5} - 10^{-5}j & -10^{-4} + 10^{-4}j & 10^{-5} + 10^{-4}j & -10^{-4} + 10^{-4}j & \textcolor{red}{0.38273 + 10^{-5}j} & \textcolor{red}{0.92386 + 10^{-5}j} & -10^{-4} - 10^{-5}j & 10^{-4} + 10^{-5}j \\
10^{-5} - 10^{-4}j & -10^{-5} - 10^{-4}j & -10^{-4} - 10^{-4}j & -10^{-5} + 10^{-5}j & -10^{-5} - 10^{-4}j & 10^{-4} - 10^{-5}j & 1 - 10^{-3}j & 10^{-5} + 10^{-4}j \\
-10^{-4} - 10^{-5}j & 10^{-4} - 10^{-4}j & -10^{-4} + 10^{-4}j & -10^{-4} + 10^{-4}j & 10^{-4} + 10^{-5}j & -10^{-4} - 10^{-4}j & -10^{-5} + 10^{-4}j & 1 - 10^{-5}j
\end{bmatrix}
\)}\vspace{0.5cm}

For practical applications focused solely on output probability distribution given an input state to the circuit, fidelity and trace can be employed, up to 5 qubits. However, if precise matrix representation is crucial, the Frobenius norm stands as the only suitable loss function, thereby mitigating the risk of obtaining incorrect values.

Another important aspect is achieved using fidelity or trace distance. Except for certain specific cases, such as $C_0^1C_1^0RY(\frac{\pi}{4})$, it has been demonstrated that this algorithm can also be used for state preparation. This is achieved because, when the algorithm is trained using these 2 loss functions, $1000$ random states are used. Then, if the approximated unitary matrix is correct, it means that the QNN can also successfully perform approximate state preparation.

\section{Conclusions}\label{sec:conclusions}
In this work, a novel single layer QNN is proposed to approximate any unitary operator through the Standard Recursive Block Basis (SRBB) decomposition, which is capable of representing unitary matrices in a scalable parametric form by exploiting Lie algebras. Particular attention is paid to the relationship between $SU(2^n)$ and $U(2^n)$ Lie groups, given that the algorithm works only with special unitary matrices; by exploiting a training process based on a phase correction approach (unitary scaling), it is possible to generalize the approximation technique to generic unitary matrices. Firstly, the $n=2$ case is found to be a peculiar scenario of the scalability scheme proposed by the literature only from a theoretical point of view. The latter concerns the ordering and grouping of the elements of the algebra, based on their respective algebraic properties, but not the possible and theoretically proposed optimizations (simplifications of CNOT gates). The reason for the peculiarity of the $n=2$ case lies in the properties of the matrix algebra under permutation via 2-cycles (or transpositions) and gives rise to specific simplifications for this case. Secondly, a novel CNOT-reduced scaling scheme is discovered and first implemented from $n\geqslant3$, leading to new gate count formulas; it makes use of a Gray-Code-type ordering and the properties of the binary representation to find and prove the scalable scheme of the simplifications of CNOT gates. Indeed, in the 2-qubit case, other simplifications that cannot be incorporated into this new scalable scheme occur to reach the optimal number of CNOTs. For $n\geqslant6$, the fraction of simplified gates resulting from this new scalable scheme is not yet suitable for practical use of the QNN, but it is a solid starting point to reach the optimal depth in the SRBB-based synthesis algorithm, leveraging the connection between scalable VQCs
and unitary parameterizations through Lie algebras. Thirdly, with the aim of developing a Python library, a revised version of the recursive algorithm that constructs the SRBB is derived, thus simplifying the management and manipulation of high cardinality algebras (exponential
growth). The revision concerns the particular ordering, only partially increasing, of the matrix elements for the second and third methods of construction of the algebra; this detail allows for a clear match of the properties under permutation to the algebra sub-groupings, thus structuring scalability from an implementation perspective.

The scalable CNOT-reduced VQC was implemented using the PennyLane library, and the tests were performed using the simulator provided by the same library; some simple tests with 2 qubits were performed on IBM real hardware to assess the algorithm's usability. The QNN was tested for different numbers of qubits, from 2 to 6, with only one single layer of approximation; first with predefined circuits designed with increasing complexity and then with dense random $SU(2^n)$ matrices.

The network was trained using two different optimizers, Adam and Nelder-Mead, while the performance was evaluated with different loss functions, such as trace distance, fidelity, and Frobenius norm, each with a different application based on the specific metric criterion within their definitions. In general, the algorithm achieved a low loss value of up to 5 qubits. 
As expected, the precision decreases as the number of qubits increases, indicating a possibly too complex parametric landscape with additional local minima. Furthermore, the performance worsens when using random matrices, suggesting the simpler approximation of sparse matrices or matrices with a specific pattern.
Furthermore, while the Nelder-Mead optimizer can achieve a good quality of approximation, the Adam optimizer results in worse performance but requires significantly less execution time. Moreover, the trained network was tested on real IBM hardware to approximate a CNOT gate, achieving a high level of fidelity as measured by the Hellinger distance.

In the future, the proposed framework could be used to design efficient algorithms for Quantum State Preparation \cite{belli2025srbb} or classification tasks. In state preparation, the objective is to achieve a specific quantum state, and the proposed VQC could generate an approximated matrix capable of reaching each desired state. Similarly, in classification tasks, this framework could be used to discover the appropriate unitary matrix to maximize the probability of finding the correct classes. Additionally, it would be interesting to use this algorithm in Hamiltonian simulation with the aim of automating the design of a versatile simulator of machine-learned quantum evolutions. As the number of qubits increases, the depth and the number of gates in the QNN increase rapidly, leading to overly deep circuits and significantly longer training times. Therefore, an important improvement would be to find a way to reduce the number of rotations in the VQC while still maintaining scalability. One strategy could be to define approximating sub-algebras, thus reducing the number of Lie parameters needed. Alternatively, one could consider more efficient or ad hoc optimized implementations, for example, based on ancillary qubits, for the 3 main factors defined by uniformly controlled gates. Finally, it would be interesting to explore the development of variational parameter optimization processes (learning) that are more dependent on the Lie geometry and the topology of the unitary group.

\appendix

\section{Transposition factors and CNOT sequences}\label{sec:cnot_sequences}
The permutations considered by definition~(\ref{eqn:def_prodT}) are called 2-cycles or transpositions and can be implemented through CNOT sequences. In the range $1\leqslant x\leqslant2^{n-1}-1$, the binary representation of $x$ is obtained according to the formula $\sum_{i=1}^{n-1}\,2^{n-i-1}\,x_i$, thus enumerating all $\prod T_x^e$ factors. Then, for each $x$, a CNOT gate is associated with the control in the $n$-th qubit and the target in the $i$-th qubit, where $x_i=1$. For the odd case, the same scheme is used with only one variation: the presence of the parameter $k$, which is the lowest index for which $x_k=1$ appears in the binary string. Thus, $\prod T_x^o$ factors correspond to the circuit $CNOT_{(k,n)}\prod T_x^e\,CNOT_{(k,n)}$.

This recipe allows us to easily derive the quantum circuit associated with each $\prod T_x^{e/o}$ factor, but the link with permutation matrices remains somewhat hidden. Firstly, the binary representation previously described considers only $n-1$ qubits, i.e., only the qubits suitable for targets (the control qubit is always the $n$-th). In this way, considering the binary representation for each $x$ is equivalent to considering all possible combinations between target lines and the $n$-th qubit. Secondly, from a matrix representation viewpoint, the states of the $n$-qubit quantum register are ordered according to the usual sequence of kets:
\begin{equation*}
\begin{split}
    |0...00\rangle&=|0\rangle\otimes...\otimes|0\rangle\otimes|0\rangle=(1,0,...,0,0)^T\\
    |0...01\rangle&=|0\rangle\otimes...\otimes|0\rangle\otimes|1\rangle=(0,1,...,0,0)^T\\
    \cdots\;&=\;\cdots\\
    |1...11\rangle&=|1\rangle\otimes...\otimes|1\rangle\otimes|1\rangle=(0,0,...,0,1)^T
\end{split}
\end{equation*}
in which only the ones ending with a 1-valued qubit activate the CNOT. The latter are exactly half of the total number of states, and thanks to the properties of binary representation, they are always coupled with respect to the only one qubit that changes between each other. See for instance the $\prod T_1^e$ factor for $n=2$ ($x=1$):
\begin{equation*}
    x=1\equiv2^0\cdot x_1\iff x_1=1\quad\longrightarrow\quad CNOT_{(2,1)}
\end{equation*}
\begin{equation*}
   \prod T_1^e=P_{(2,4)}=
   \begin{pmatrix}
       1&0&0&0\\
       0&0&0&1\\
       0&0&1&0\\
       0&1&0&0
   \end{pmatrix}
   \begin{tabular}{cc}
        $\rightarrow$&$|00\rangle=(1,0,0,0)^T$\\
        $\rightarrow$&$|01\rangle=(0,1,0,0)^T$\\
        $\rightarrow$&$|10\rangle=(0,0,1,0)^T$\\
        $\rightarrow$&$|11\rangle=(0,0,0,1)^T$
   \end{tabular}
\end{equation*}
The remaining half of the states, the ones ending with a 0-valued qubit, do not activate the CNOT and represent the so-called fixed points of the transposition. As shown in the example, the pair of indices $(\alpha,\beta)=(2,4)$ corresponds to the second and fourth vectors (coupled by the first qubit that represents the two possible values of the target line) that activate the CNOT and are then permuted. Finally, note that the permutation matrix $P_{(2,4)}$ corresponds to exchanging rows 2 and 4 if multiplied from the left, or exchanging columns 2 and 4 if multiplied from the right. Of course, the same link between 2-cycle (or transposition) matrices and CNOT-sequences applies for any value of $n$. In this Appendix, an intuitive and circuit-oriented explanation of the relationship between blocks $\prod T_x^{e/o}$, transposition matrices, and CNOT-gate sequences has been provided. In~\cite[Theorem 5.1]{sarkar2024quantum} this relationship is formalized rigorously via a closed form of permutation strings.

\section{Quantum circuit to approximate 3-qubit systems}\label{sec:3qubits}
According to Section~\ref{sec:recursiveSRBB}, the design of the VQC valid for any 3-qubit operator requires the matrix algebra of order 8. However, unlike what was shown in Section~\ref{sec:2qubits} for 2-qubit evolutions, since the basis of interest has 64 matrices of order 8, it is really impractical to list each of its elements or write the bases of order from 4 to 7 that build it. For these reasons, only some general characteristics of the basis are provided below and useful shortcuts are proposed to manage all the elements of the algebra.  

The basis $\mathcal{B}^{(8)}$ is composed of 64 algebraic elements that describe the general evolution of an arbitrary 3-qubit quantum system. The set $\{B_j^{(8)}:1\leqslant j\leqslant63\}$ is a basis for the $su(8)$ matrix algebra whose elements satisfy the properties of hermiticity, unitarity, and zero trace. The last element of the basis is, as usual, $B_{64}^{(8)}=\mathbb{I}_8$. Following the general Equation~(\ref{eqn:U_approx}) with $l=1$, there are exactly:
\begin{enumerate}
    \item[i)]8 diagonal elements belonging to the set\\$\mathcal{J}=\mathcal{J}_Z\cup\{64\}=\{m^2-1,\;2\leqslant m\leqslant8\}\cup\{64\}=\{3,8,15,24,35,48,63,64\}$;
    \item[ii)]8 elements with a $ZYZ$-decomposition, divided into pairs, belonging to the set\\$\mathcal{J}_\Psi=\{[(2m-1)^2,(4m^2-2m)],\;1\leqslant m\leqslant4\}=\{(1,2),(9,12),(25,30),(49,56)\}$;
    \item[iii)]24 elements divided into 6 quadruples, each of which is associated with a precise element of the set of permutations $P_8^{even}=\{P_{(2,4)},P_{(2,6)},P_{(2,8)},P_{(4,6)},P_{(4,8)},P_{(6,8)}\}$, as shown in Table~\ref{tab:evenpairs_element_n3};
    \item[iv)]24 elements divided into 6 quadruples, each of which is associated with a precise element of the set of permutations $P_8^{odd}=\{P_{(2,3)},P_{(2,5)},P_{(2,7)},P_{(4,5)},P_{(4,7)},P_{(6,7)}\}$, as shown in Table~\ref{tab:oddpairs_element_n3}.
\end{enumerate}
\begin{table}[htbp]
    \centering
    \begin{subtable}[t]{0.45\textwidth}
    \centering
        \begin{tabular}{||c|c||}
        \hline
        Pair&Elements\\
        \hline\hline
        (2,4)&10,13,4,6\\
        (2,6)&26,31,18,22\\
        (2,8)&50,57,38,44\\
        (4,6)&28,33,16,20\\
        (4,8)&52,59,40,46\\
        (6,8)&54,61,36,42\\
        \hline
        \end{tabular}
    \caption{Even case.}
    \label{tab:evenpairs_element_n3}
    \end{subtable}
    \quad
    \begin{subtable}[t]{0.45\textwidth}
    \centering
        \begin{tabular}{||c|c||}
        \hline
        Pair&Elements\\
        \hline\hline
        (2,3)&5,7,11,14\\
        (2,5)&17,21,27,32\\
        (2,7)&37,43,51,58\\
        (4,5)&19,23,29,34\\
        (4,7)&39,45,53,60\\
        (6,7)&41,47,55,62\\
        \hline
        \end{tabular}
    \caption{Odd case.}
    \label{tab:oddpairs_element_n3}
    \end{subtable}
\caption{Association between algebraic elements and even/odd (2-cycles) permutations.}
\label{tab:pairs_element_n3}
\end{table}
Each element of the basis has a very specific role captured by the shortcuts described in the following subsections and is taken only once (note that the last element does not participate in the generation of $SU(8)$ matrices). The transition to the SRBB, namely, the replacement of the diagonal elements belonging to the set $\mathcal{J}$ with other Hermitian unitary diagonal matrices, coming from Pauli strings, is illustrated in Table~\ref{tab:SRBB_diagonals_n3} (refer to Table~\ref{tab:SRBB_diagonals_n2} for the explanation). 
\begin{table}[htbp]
\centering
\begin{tabular}{||c|c|c|c|c|c||}
\hline
Decimal&Binary&String&Matrix&Element&Replaced\\
\hline\hline
0&000&$\mathbb{I}_2\otimes\mathbb{I}_2\otimes\mathbb{I}_2$&$\mathbb{I}_8$&64&no\\
1&001&$\mathbb{I}_2\otimes\mathbb{I}_2\otimes\sigma_3$&$diag(1,-1,1,-1,1,-1,1,-1)$&3&no\\
2&010&$\mathbb{I}_2\otimes\sigma_3\otimes\mathbb{I}_2$&$diag(1,1,-1,-1,1,1,-1,-1)$&8&yes\\
3&011&$\mathbb{I}_2\otimes\sigma_3\otimes\sigma_3$&$diag(1,-1,-1,1,1,-1,-1,1)$&15&yes\\
4&100&$\sigma_3\otimes\mathbb{I}_2\otimes\mathbb{I}_2$&$diag(1,1,1,1,-1,-1,-1,-1)$&24&yes\\
5&101&$\sigma_3\otimes\mathbb{I}_2\otimes\sigma_3$&$diag(1,-1,1,-1,-1,1,-1,1)$&35&yes\\
6&110&$\sigma_3\otimes\sigma_3\otimes\mathbb{I}_2$&$diag(1,1,-1,-1,-1,-1,1,1)$&48&yes\\
7&111&$\sigma_3\otimes\sigma_3\otimes\sigma_3$&$diag(1,-1,-1,1,-1,1,1,-1)$&63&yes\\
\hline
\end{tabular}
\caption{The new sequence of diagonal basis elements that marks the transition to SRBB.}
\label{tab:SRBB_diagonals_n3}
\end{table}
Therefore, the complete SRBB of order 8 is identified as the set $\mathcal{U}^{(8)}=\{U_j^{(8)}:\;1\leqslant j\leqslant64\}$ and is defined by
\begin{equation}
U^{(8)}_j=\left\{
\begin{array}{l}
\mbox{new diagonal elements if }j\in\mathcal{J}=\{3,8,15,24,35,48,63,64\}\\
B^{(8)}_j\mbox{ otherwise}
\end{array}\right.
\end{equation}

\subsection{Diagonal contributions}
The $Z$-factor~(\ref{eqn:def_Z}) responsible for diagonal contributions in the case $n=3$ is:
\begin{equation}
    Z(\Theta_Z)=\prod_{j\in\mathcal{J}_Z}\exp\{i\,\theta_jU_j^{(8)}\}
\end{equation}
where $\mathcal{J}_Z=\{m^2-1,\;2\leqslant m\leqslant8\}$ and $\Theta_Z$ collect the associated Lie parameters. According to Section~\ref{sec:diagonals_n2} and~\cite{Sarkar:2024bax}, Figure~\ref{fig:zeta_diagram_n3} shows the values of the parametric pairs $(m,m')$, representing the first essential step for designing the circuits.
\begin{figure}[htbp]
\centering
\begin{subfigure}[b]{0.4\textwidth}
\centering
    \begin{tabular}{||c|c|c|c||}
    \hline
    Decimal&Binary&String&Element\\
    \hline\hline
    1&001&$\mathbb{I}_2\otimes\mathbb{I}_2\otimes\sigma_3$&3\\
    2&010&$\mathbb{I}_2\otimes\sigma_3\otimes\mathbb{I}_2$&8\\
    3&011&$\mathbb{I}_2\otimes\sigma_3\otimes\sigma_3$&15\\
    4&100&$\sigma_3\otimes\mathbb{I}_2\otimes\mathbb{I}_2$&24\\
    5&101&$\sigma_3\otimes\mathbb{I}_2\otimes\sigma_3$&35\\
    6&110&$\sigma_3\otimes\sigma_3\otimes\mathbb{I}_2$&48\\
    7&111&$\sigma_3\otimes\sigma_3\otimes\sigma_3$&63\\
    \hline
    \end{tabular}
\end{subfigure}\hfill
\begin{subfigure}[b]{0.4\textwidth}
\resizebox{0.98\textwidth}{!}{
\centering
\tikzset{phase label/.append style={label position=above}}
    \begin{quantikz}
    &\push{\;\mathbb{I}_2\;}\gategroup[7,steps=1,style={thin}]{$1^{\circ}\,qubit\,(0)$}&&&\push{\;\mathbb{I}_2\;}\gategroup[7,steps=1,style={thin}]{$2^{\circ}\,qubit\,(1)$}&&&\push{\;\sigma_3\;}\gategroup[7,steps=1,style={thin}]{$3^{\circ}\,qubit\,(2)$}&\rstick{$m=2$}\\
    &\push{\;\mathbb{I}_2\;}&&&\push{\;\sigma_3\;}&&&\push{\;\mathbb{I}_2\;}&\rstick{$m=1$}\\
    &\push{\;\mathbb{I}_2\;}&&&\push{\;\sigma_3\;}&&&\push{\;\sigma_3\;}&\rstick{$m=2,\,m'=1$}\\
    &\push{\;\sigma_3\;}&&&\push{\;\mathbb{I}_2\;}&&&\push{\;\mathbb{I}_2\;}&\rstick{$m=0$}\\
    &\push{\;\sigma_3\;}&&&\push{\;\mathbb{I}_2\;}&&&\push{\;\sigma_3\;}&\rstick{$m=2,\,m'=0$}\\
    &\push{\;\sigma_3\;}&&&\push{\;\sigma_3\;}&&&\push{\;\mathbb{I}_2\;}&\rstick{$m=1,\,m'=0$}\\
    &\push{\;\sigma_3\;}&&&\push{\;\sigma_3\;}&&&\push{\;\sigma_3\;}&\rstick{$m=2,\,m'=0,1$}
    \end{quantikz}}
\end{subfigure}
\caption{Diagram to find the position of CNOTs and rotation gates.}
\label{fig:zeta_diagram_n3}
\end{figure}
Then, the second step of the shortcut consists in associating each line of the previous diagram (and therefore each associated exponential factor) with one or more merged gates, following the scheme already explained in Section~\ref{sec:diagonals_n2} (see Figure~\ref{fig:zeta_components_n3}). Finally, Figure~\ref{fig:zeta_circuit_n3} shows the sequence of diagonal elements that minimizes\footnote{Section~\ref{sec:scaling} explains how the CNOT-optimized circuit emerges from the scalable algorithm of simplifications.} the number of CNOTs within the $Z$-factor circuit, which is built following Proposition~\ref{prop_Z_scalable} and Algorithm~\ref{alg1} of Section~\ref{sec:scaling}.
\begin{figure}[htbp]
\centering
    \begin{subfigure}[b]{0.3\textwidth}
        \centering
        \begin{quantikz}
        \lstick{0}&&&&\\
        \lstick{1}&&&&\\
        \lstick{2}&&\gate{R_z(\theta_3)}&&
        \end{quantikz}
        \caption{$\exp\{i\,\theta_3U_3\}$}
    \end{subfigure}
    \hfill
    \begin{subfigure}[b]{0.3\textwidth}
        \centering
        \begin{quantikz}
        \lstick{0}&&&&\\
        \lstick{1}&&\gate{R_z(\theta_8)}&&\\
        \lstick{2}&&&&   
        \end{quantikz} 
        \caption{$\exp\{i\,\theta_8U_8\}$}
    \end{subfigure}
    \hfill
    \begin{subfigure}[b]{0.3\textwidth}
        \centering
        \begin{quantikz}
        \lstick{0}&&&&\\
        \lstick{1}&\ctrl{1}&&\ctrl{1}&\\
        \lstick{2}&\targ{}&\gate{R_z(\theta_{15})}&\targ{}&
        \end{quantikz}     
        \caption{$\exp\{i\,\theta_{15}U_{15}\}$}
    \end{subfigure}
    \begin{subfigure}[b]{0.3\textwidth}
        \centering
        \begin{quantikz}
        \lstick{0}&&\gate{R_z(\theta_{24})}&&\\
        \lstick{1}&&&&\\
        \lstick{2}&&&&
        \end{quantikz}     
        \caption{$\exp\{i\,\theta_{24}U_{24}\}$}
    \end{subfigure}
    \hfill
    \begin{subfigure}[b]{0.3\textwidth}
        \centering
        \begin{quantikz}
        \lstick{0}&\ctrl{2}&&\ctrl{2}&\\
        \lstick{1}&&&&\\
        \lstick{2}&\targ{}&\gate{R_z(\theta_{35})}&\targ{}&
        \end{quantikz}     
        \caption{$\exp\{i\,\theta_{35}U_{35}\}$}
    \end{subfigure}
    \hfill
    \begin{subfigure}[b]{0.3\textwidth}
        \centering
        \begin{quantikz}
        \lstick{0}&\ctrl{1}&&\ctrl{1}&\\
        \lstick{1}&\targ{}&\gate{R_z(\theta_{48})}&\targ{}&\\
        \lstick{2}&&&&   
        \end{quantikz} 
        \caption{$\exp\{i\,\theta_{48}U_{48}\}$}
    \end{subfigure}
    \begin{subfigure}[b]{0.5\textwidth}
        \centering
        \begin{quantikz}
        \lstick{0}&\ctrl{2}&&&&\ctrl{2}&\\
        \lstick{1}&&\ctrl{1}&&\ctrl{1}&&\\
        \lstick{2}&\targ{}&\targ{}&\gate{R_z(\theta_{63})}&\targ{}&\targ{}&
        \end{quantikz}     
        \caption{$\exp\{i\,\theta_{63}U_{63}\}$}
    \end{subfigure}
\caption{Seven little blocks that make up the $Z$-factor circuit for $n=3$.}
\label{fig:zeta_components_n3}
\end{figure}
\begin{figure}[ht]
\resizebox{1\textwidth}{!}{
\centering
    \begin{quantikz}
    \lstick{$1^\circ$q}&\gategroup[3,steps=12,style={dashed,rounded corners},label style={label position=below,anchor=north,yshift=-0.3cm}]{Section with simplifications (scalable)}&&&\ctrl{2}&&&&\ctrl[style={red}]{2}&\ctrl[style={red}]{2}&&\ctrl{2}&&\ctrl{1}\gategroup[2,steps=4,style={dashed,rounded corners},label style={label position=above,anchor=north,yshift=+0.3cm}]{Circuit for $n=2$ (recursive)}&&\ctrl{1}&\gate{R_z(\theta_{24})}&\\
    \lstick{$2^\circ$q}&\ctrl{1}&&\ctrl[style={red}]{1}&&\ctrl[style={red}]{1}&&\ctrl{1}&&&&&&\targ{}&\gate{R_z(\theta_{48})}&\targ{}&\gate{R_z(\theta_8)}&\\
    \lstick{$3^\circ$q}&\targ{}&\gate{R_z(\theta_{15})}&\targ[style={red}]{}&\targ{}&\targ[style={red}]{}&\gate{R_z(\theta_{63})}&\targ{}&\targ[style={red}]{}&\targ[style={red}]{}&\gate{R_z(\theta_{35})}&\targ{}&\gate{R_z(\theta_3)}&&&&&
    \end{quantikz}}
\caption{$Z$-factor quantum circuit for $n=3$; the highlighted gates simplify each other.}
\label{fig:zeta_circuit_n3}
\end{figure}

\subsection{Even contributions}\label{sec:even_3qubits}
As described in Section~\ref{sec:recursiveSRBB}, the $\Psi$-factor~(\ref{eqn:def_Psi}) for a general 3-qubit system,
\begin{equation}
    \Psi(\Theta_\Psi)=\left[\prod_{(i,j)\in\mathcal{A}_\Psi}\exp\{i\,\theta_iU_i\}\exp\{i\,\theta_jU_j\}\right]\prod_{x=1}^{3}\left(\prod T_x^e\right)M_x^e\left(\prod T_x^e\right)
\end{equation}
is built from two different type of subsets of SRBB elements:
\begin{itemize}
\item a set of pairs called $\mathcal{A}_{\Psi}$ makes up the first sub-factor of type A enclosed in squared brackets,
\begin{equation}
\mathcal{A}_{\Psi}=\{[(2m-1)^2,(4m^2-2m)],\;1\leqslant m\leqslant4)\}=\{(1,2),(9,12),(25,30),(49,56)\}
\end{equation}
\item a set of quadruples called $\mathcal{B}_\Psi$, divided into three subsets $\mathcal{B}_{\Psi,x}$, makes up the second sub-factor of type $B$,
\begin{equation}
\mathcal{B}_\Psi=\left\{
\begin{array}{l}
    \mathcal{B}_{\Psi,1}=\left\{
    \begin{array}{l}
        \left[h_4(1),f_4(1),h_3(2),f_3(2)\right]=(10,13,4,6)\\
        \left[h_8(5),f_8(5),h_7(6),f_7(6)\right]=(54,61,36,42)
    \end{array}\right.\\
    \mathcal{B}_{\Psi,2}=\left\{
    \begin{array}{l}
        \left[(h_6(1),f_6(1),h_5(2),f_5(2)\right]=(26,31,18,22)\\
        \left[h_8(3),f_8(3),h_7(4),f_7(4)\right]=(52,59,40,46)
    \end{array}\right.\\
    \mathcal{B}_{\Psi,3}=\left\{
    \begin{array}{l}
        \left[h_8(1),f_8(1),h_7(2),f_7(2)\right]=(50,57,38,44)\\
        \left[h_6(3),f_6(3),h_5(4),f_5(4)\right]=(28,33,16,20)
    \end{array}\right.
\end{array}\right.
\end{equation}
\end{itemize}
Furthermore, from Table~\ref{tab:evenpairs_element_n3}, it is possible to identify the $2^{n-1}-1\big|_{n=3}=3$ subsets $\prod T_x^e$ for $x=1,2,3$ of disjoint transpositions, each containing $t=2^{n-2}\big|_{n=3}=2$ elements, as stated by Equation~\ref{eqn:def_prodT}, using a simple selection algorithm. First, the elements of the permutation group $P_8^{even}$ are arranged in a tabular representation, as shown in Table~\ref{tab:2cyclesP8even}, where each row is indexed by the first element of the index pair. Then, to ensure that the elements selected for each $\prod T_x^e$ are disjoint from each other, it is enough to consider the following instructions:
\begin{itemize}
    \item[i)] for $1\leqslant x\leqslant3$, the first element $P_{(\alpha,\beta)}$ of the $t$-tuple is chosen from the first row of Table~\ref{tab:2cyclesP8even}, read from left to right, skipping any elements considered in previous cycles;
    \item[ii)] the remaining elements $P_{(\alpha',\beta')}$ of the $t$-tuple are chosen starting from the bottom of the table, this time reading the rows from right to left, and taking care to consider $\alpha'\neq\beta$. The selected items will no longer be available for subsequent cycles;
    \item[iii)] repeat instructions i) and ii) for all values of $x$.
\end{itemize}
In \cite[Theorem 5.1]{sarkar2024quantum}, a closed form for the permutation strings defining the disjoint partitions of $P_{2^n}^{even/odd}$ is rigorously defined; however, here, a logic preparatory to the implementation is preferred, and the above selection recipe is an equivalent alternative version. The same logic, in fact, aligns with the intent of the \emph{new formulation} to manage increasing complexity with an implementation-oriented approach, and it has been used to implement the construction of the main factors $Z,\Psi,\Phi$ and the verification of their algebraic properties in the repository \cite{srbb-syn}.
\begin{table}[htbp]
    \centering
    \begin{tabular}{|c|c|c|}
    \hline
    $P_{(2,4)}$&$P_{(2,6)}$&$P_{(2,8)}$\\
    \hline
    &$P_{(4,6)}$&$P_{(4,8)}$\\
    \hline
    &&$P_{(6,8)}$\\
    \hline
    \end{tabular}
    \caption{Elements of the 2-cycle group $P_8^{even}$.}
    \label{tab:2cyclesP8even}
\end{table}
Accordingly,
\begin{equation}
    \prod T_1^e=P_{(2,4)}P_{(6,8)}\;,\quad\prod T_2^e=P_{(2,6)}P_{(4,8)}\;,\quad\prod T_3^e=P_{(2,8)}P_{(4,6)}
\end{equation}
Moreover,
\begin{equation}
\begin{aligned}
    M_1^e&=\prod T_1^e\left[\exp\{i\,\theta\,U_{h_4(1)}\}\exp\{i\,\theta\,U_{f_4(1)}\}\exp\{i\,\theta\,U_{h_3(2)}\}\exp\{i\,\theta\,U_{f_3(2)}\}\right.\cdot\\
    &\cdot\left.\exp\{i\,\theta\,U_{h_8(5)}\}\exp\{i\,\theta\,U_{f_8(5)}\}\exp\{i\,\theta\,U_{h_7(6)}\}\exp\{i\,\theta\,U_{f_7(6)}\}\right]\prod T_1^e
\end{aligned}
\end{equation}
\begin{equation}
\begin{aligned}
    M_2^e&=\prod T_2^e\left[\exp\{i\,\theta\,U_{h_6(1)}\}\exp\{i\,\theta\,U_{f_6(1)}\}\exp\{i\,\theta\,U_{h_5(2)}\}\exp\{i\,\theta\,U_{f_5(2)}\}\right.\cdot\\
    &\cdot\left.\exp\{i\,\theta\,U_{h_8(3)}\}\exp\{i\,\theta\,U_{f_8(3)}\}\exp\{i\,\theta\,U_{h_7(4)}\}\exp\{i\,\theta\,U_{f_7(4)}\}\right]\prod T_2^e
\end{aligned}
\end{equation}
\begin{equation}
\begin{aligned}
    M_3^e&=\prod T_3^e\left[\exp\{i\,\theta\,U_{h_8(1)}\}\exp\{i\,\theta\,U_{f_8(1)}\}\exp\{i\,\theta\,U_{h_7(2)}\}\exp\{i\,\theta\,U_{f_7(2)}\}\right.\cdot\\
    &\cdot\left.\exp\{i\,\theta\,U_{h_6(3)}\}\exp\{i\,\theta\,U_{f_6(3)}\}\exp\{i\,\theta\,U_{h_5(4)}\}\exp\{i\,\theta\,U_{f_5(4)}\}\right]\prod T_3^e
\end{aligned}    
\end{equation}
so that, the first sub-factor of type $A$ and the factors called $M_x^e$ into the second sub-factor of type $B$ become $(2\times2)$-block diagonal matrices,
\begin{equation}
    \prod_{j\in\mathcal{A}_\Psi}\exp\{i\,\theta_jU_j\}=
    \begin{pmatrix}
        A_{\Psi}^1&0&0&0\\
        0&A_{\Psi}^2&0&0\\
        0&0&A_{\Psi}^3&0\\
        0&0&0&A_{\Psi}^4
    \end{pmatrix}
\end{equation}
\begin{equation}
    M_x^{e}=\prod T_x^e\cdot\prod_{j\in\mathcal{B}_{\Psi,x}}\exp\{i\,\theta_jU_j\}\cdot\prod T_x^e=
    \begin{pmatrix}
        B_{\Psi,x}^1&0&0&0\\
        0&B_{\Psi,x}^2&0&0\\
        0&0&B_{\Psi,x}^3&0\\
        0&0&0&B_{\Psi,x}^4
    \end{pmatrix}
\end{equation}
where each block belongs to $SU(2)$, as can be quickly verified, thus composing as a whole two unitary matrices belonging to $M_3ZYZ$. The sub-circuit for the first sub-factor of type $A$ and for each $M_x^e$ of the second sub-factor of type $B$ is given by the complete decomposition of a $M_3ZYZ$-type matrix, illustrated in Figure~\ref{fig:zyz_circuit_n3}, where rotation parameters are renamed after the decomposition. Finally, Figure~\ref{fig:prodTxe_diagram_n3} shows the quantum circuits for the factors $\prod T_x^e$ (see Appendix~\ref{sec:cnot_sequences}).

\begin{figure}[ht]
\centering
    \resizebox{1\textwidth}{!}{
    \begin{subfigure}[b]{1\textwidth}
    \centering
        \begin{quantikz}
            \lstick{0}&&&&\ctrl{2}&&&&\rstick{...}\\
            \lstick{1}&&\ctrl{1}&&&&\ctrl{1}&&\rstick{...}\\
            \lstick{2}&\gate{R_z(\theta_1^*)}&\targ{}&\gate{R_z(\theta_2^*)}&\targ{}&\gate{R_z(\theta_3^*)}&\targ{}&\gate{R_z(\theta_4^*)}&\rstick{...}
        \end{quantikz}
    \end{subfigure}}\vspace{3mm}
    \resizebox{1\textwidth}{!}{
    \begin{subfigure}[b]{1\textwidth}
    \centering
        \begin{quantikz}
            \lstick{...}&&&&\ctrl{2}&&&&&\rstick{...}\\
            \lstick{...}&&\ctrl{1}&&&&\ctrl{1}&&&\rstick{...}\\
            \lstick{...}&\gate{R_y(\theta_5^*)}&\targ{}&\gate{R_y(\theta_6^*)}&\targ{}&\gate{R_y(\theta_7^*)}&\targ{}&\gate{R_y(\theta_8^*)}&&\rstick{...}
        \end{quantikz}
    \end{subfigure}}\vspace{3mm}
    \resizebox{1\textwidth}{!}{
    \begin{subfigure}[b]{1\textwidth}
    \centering
        \begin{quantikz}
            \lstick{...}&&&&\ctrl{2}&&&&\ctrl{2}&\\
            \lstick{...}&&\ctrl{1}&&&&\ctrl{1}&&&\\
            \lstick{...}&\gate{R_z(\theta_9^*)}&\targ{}&\gate{R_z(\theta_{10}^*)}&\targ{}&\gate{R_z(\theta_{11}^*)}&\targ{}&\gate{R_z(\theta_{12}^*)}&\targ{}&
        \end{quantikz}
    \end{subfigure}}
    \caption{The fully decomposed $M_3ZYZ$ quantum circuit.}
    \label{fig:zyz_circuit_n3}
\end{figure}\vspace{5cm}
\begin{figure}[ht]
    \begin{subfigure}[c]{1\textwidth}
    \centering
        \begin{tabular}{||c|c|c|c||}
        \hline
        x&Bits&Binary&Control-Target\\
        \hline\hline
        1&$x_0x_1$&01&(2,1)\\
        2&$x_0x_1$&10&(2,0)\\
        3&$x_0x_1$&11&(2,0),(2,1)\\
        \hline
        \end{tabular}   
    \end{subfigure}\vspace{3mm}
    \begin{subfigure}[c]{0.3\textwidth}
    \centering
        \begin{quantikz}
        \lstick{0}&&&&\\
        \lstick{1}&&\targ{}&&\\
        \lstick{2}&&\ctrl{-1}&&
        \end{quantikz}
        \caption{$\prod T_1^e$}
    \end{subfigure}
    \begin{subfigure}[c]{0.3\textwidth}
    \centering
        \begin{quantikz}
        \lstick{0}&&\targ{}&&\\
        \lstick{1}&&&&\\
        \lstick{2}&&\ctrl{-2}&&
        \end{quantikz}
        \caption{$\prod T_2^e$}
    \end{subfigure}
    \begin{subfigure}[c]{0.3\textwidth}
    \centering
        \begin{quantikz}
        \lstick{0}&\targ{}&&&\\
        \lstick{1}&&\targ{}&&\\
        \lstick{2}&\ctrl{-2}&\ctrl{-1}&&
        \end{quantikz}
        \caption{$\prod T_3^e$}
    \end{subfigure}
\caption{Diagram to find the control-target pair for the CNOTs representing $\prod T_x^e$.}
\label{fig:prodTxe_diagram_n3}
\end{figure}

\subsection{Odd contributions}\label{sec:odd_3qubits}
In the general case $n=3$, the $\Phi$-factor~(\ref{eqn:def_Phi}),
\begin{equation}
    \Phi(\Theta_\Phi)=\prod_{x=1}^3\left(\prod T_x^o\right)M_x^o\left(\prod T_x^o\right)
\end{equation}
is built from a subset $\mathcal{C}_\Phi$ of quadruples of SRBB elements:
\begin{equation}
\mathcal{C}_\Phi=\left\{
\begin{array}{l}
    \mathcal{C}_{\Phi,1}=\left\{
    \begin{array}{l}
        \left[h_3(1),f_3(1),h_4(2),f_4(2)\right]=(5,7,11,14)\\
        \left[h_7(5),f_7(5),h_8(6),f_8(6)\right]=(41,47,55,62)
    \end{array}\right.\\
    \mathcal{C}_{\Phi,2}=\left\{
    \begin{array}{l}
        \left[(h_5(1),f_5(1),h_6(2),f_6(2)\right]=(17,21,27,32)\\
        \left[h_7(3),f_7(3),h_8(4),f_8(4)\right]=(39,45,53,60)
    \end{array}\right.\\
    \mathcal{C}_{\Phi,3}=\left\{
    \begin{array}{l}
        \left[h_7(1),f_7(1),h_8(2),f_8(2)\right]=(37,43,51,58)\\
        \left[h_5(3),f_5(3),h_6(4),f_6(4)\right]=(19,23,29,34)
    \end{array}\right.
\end{array}\right.
\end{equation}
Besides, following the selection procedure described for the even contributions (Section~\ref{sec:even_3qubits}) and referring to Table~\ref{tab:2cyclesP8odd},
\begin{table}[htbp]
    \centering
    \begin{tabular}{|c|c|c|}
    \hline
    $P_{(2,3)}$&$P_{(2,5)}$&$P_{(2,7)}$\\
    \hline
    &$P_{(4,5)}$&$P_{(4,7)}$\\
    \hline
    &&$P_{(6,7)}$\\
    \hline
    \end{tabular}
    \caption{Elements of the 2-cycle group $P_8^{odd}$.}
    \label{tab:2cyclesP8odd}
\end{table}
\begin{equation}
    \prod T_1^o=P_{(2,3)}P_{(6,7)}\;,\quad\prod T_2^o=P_{(2,5)}P_{(4,7)}\;,\quad\prod T_3^o=P_{(2,7)}P_{(4,5)}\;,
\end{equation}
and
\begin{equation}\label{eqn:M1odd_n3}
\begin{aligned}
    M_1^o&=\prod T_1^o\left[\exp\{i\,\theta\,U_{h_3(1)}\}\exp\{i\,\theta\,U_{f_3(1)}\}\exp\{i\,\theta\,U_{h_4(2)}\}\exp\{i\,\theta\,U_{f_4(2)}\}\right.\cdot\\
    &\cdot\left.\exp\{i\,\theta\,U_{h_7(5)}\}\exp\{i\,\theta\,U_{f_7(5)}\}\exp\{i\,\theta\,U_{h_8(6)}\}\exp\{i\,\theta\,U_{f_8(6)}\}\right]\prod T_1^o
\end{aligned}
\end{equation}
\begin{equation}\label{eqn:M2odd_n3}
\begin{aligned}
    M_2^o&=\prod T_2^o\left[\exp\{i\,\theta\,U_{h_5(1)}\}\exp\{i\,\theta\,U_{f_5(1)}\}\exp\{i\,\theta\,U_{h_6(2)}\}\exp\{i\,\theta\,U_{f_6(2)}\}\right.\cdot\\
    &\cdot\left.\exp\{i\,\theta\,U_{h_7(3)}\}\exp\{i\,\theta\,U_{f_7(3)}\}\exp\{i\,\theta\,U_{h_8(4)}\}\exp\{i\,\theta\,U_{f_8(4)}\}\right]\prod T_2^o=
\end{aligned}
\end{equation}
\begin{equation}\label{eqn:M3odd_n3}
\begin{aligned}
    M_3^o&=\prod T_3^o\left[\exp\{i\,\theta\,U_{h_7(1)}\}\exp\{i\,\theta\,U_{f_7(1)}\}\exp\{i\,\theta\,U_{h_8(2)}\}\exp\{i\,\theta\,U_{f_8(2)}\}\right.\cdot\\
    &\cdot\left.\exp\{i\,\theta\,U_{h_5(3)}\}\exp\{i\,\theta\,U_{f_5(3)}\}\exp\{i\,\theta\,U_{h_6(4)}\}\exp\{i\,\theta\,U_{f_6(4)}\}\right]\prod T_3^o=
\end{aligned}    
\end{equation}
Accordingly, the factors called $M_x^o$ become $2\times2$-block diagonal matrices,
\begin{equation}
    M_x^{o}=\prod T_x^o\cdot\prod_{i\in\mathcal{C}_{\Phi,x}}\exp\{i\,\theta_iU_i\}\cdot\prod T_x^o=
    \begin{pmatrix}
        C_{\Phi,x}^1&0&0&0\\
        0&C_{\Phi,x}^2&0&0\\
        0&0&C_{\Phi,x}^3&0\\
        0&0&0&C_{\Phi,x}^4
    \end{pmatrix}
\end{equation}
where each block no longer belongs to $SU(2)$, but only to $U(2)$, thus composing as a whole a unitary matrix that can be decomposed only in part through a $M_3ZYZ$~\cite{sarkar2023scalable}. The circuit for $M_x^o$ is given in Figure~\ref{fig:Mxo_circuit_n3}.

\begin{figure}[htbp]
    \begin{subfigure}[b]{1\textwidth}
        \centering
        \begin{quantikz}
            \lstick{0}&\gate{R_z(\theta_{1,x}^*)}\gategroup[3,steps=4,style={dashed,rounded corners},label style={label position=above,anchor=north,yshift=+0.3cm}]{pre-scaling}&\ctrl{1}&&\ctrl{1}&\gate[3]{M_3ZYZ}&\ctrl{1}\gategroup[3,steps=4,style={dashed,rounded corners},label style={label position=above,anchor=north,yshift=+0.3cm}]{post-scaling}&&\ctrl{1}&\gate{R_z(\theta_{6,x}^*)}&\\
            \lstick{1}&\gate{R_z(\theta_{2,x}^*)}&\targ{}&\gate{R_z(\theta_{3,x}^*)}&\targ{}&&\targ{}&\gate{R_z(\theta_{4,x}^*)}&\targ{}&\gate{R_z(\theta_{5,x}^*)}&\\
            \lstick{2}&&&&&&&&&&
        \end{quantikz}
    \end{subfigure}    
    \caption{Quantum circuit for $M_x^o$ factors; a $ZYZ$-core is surrounded by multi-controlled rotations.}
    \label{fig:Mxo_circuit_n3}
\end{figure}
Therefore, for $n=3$, the scalability scheme proposed by~\cite{sarkar2024quantum} begins to hold, justified by the algebraic properties under permutation shown by equations~(\ref{eqn:M1odd_n3}), (\ref{eqn:M2odd_n3}), (\ref{eqn:M3odd_n3}). The remaining three pairs of factors $\prod T_x^o$ (see Figure~\ref{fig:prodTxo_diagram_n3}) can be implemented thanks to the recipe described in Appendix~\ref{sec:cnot_sequences} and formalized in~\cite[Theorem 5.1]{sarkar2024quantum}.
\begin{figure}[htbp]
    \begin{subfigure}[c]{1\textwidth}
    \centering
        \begin{tabular}{||c|c|c|c|c||}
        \hline
        x&Bits&Binary&Control-Target&$k$-index\\
        \hline\hline
        1&$x_0x_1$&01&(2,1)&1\\
        2&$x_0x_1$&10&(2,0)&0\\
        3&$x_0x_1$&11&(2,0),(2,1)&0\\
        \hline
        \end{tabular}   
    \end{subfigure}\vspace{3mm}
    \begin{subfigure}[c]{0.3\textwidth}
    \centering
        \begin{quantikz}
        \lstick{0}&&&&\\
        \lstick{1}&\ctrl{1}&\targ{}&\ctrl{1}&\\
        \lstick{2}&\targ{}&\ctrl{-1}&\targ{}&
        \end{quantikz}
        \caption{$\prod T_1^o$}
    \end{subfigure}
    \begin{subfigure}[c]{0.3\textwidth}
    \centering
        \begin{quantikz}
        \lstick{0}&\ctrl{2}&\targ{}&\ctrl{2}&\\
        \lstick{1}&&&&\\
        \lstick{2}&\targ{}&\ctrl{-2}&\targ{}&
        \end{quantikz}
        \caption{$\prod T_2^o$}
    \end{subfigure}
    \begin{subfigure}[c]{0.3\textwidth}
    \centering
        \begin{quantikz}
        \lstick{0}&\ctrl{2}&\targ{}&&\ctrl{2}&\\
        \lstick{1}&&&\targ{}&&\\
        \lstick{2}&\targ{}&\ctrl{-2}&\ctrl{-1}&\targ{}&
        \end{quantikz}
        \caption{$\prod T_3^o$}
    \end{subfigure}
\caption{Diagram to find the control-target pair and the $k$-index for the CNOTs representing $\prod T_x^o$.}
\label{fig:prodTxo_diagram_n3}
\end{figure}

\subsection{3-qubit circuit to approximate SU(8)}
The overall quantum circuit to approximate any special unitary operator according to the equation~(\ref{eqn:U_approx}) within a QNN framework can be assembled from the sub-circuits of the previous subsections. For one single layer, the approximating operator can be rewritten in the following way:
\begin{equation}
\begin{split}
    &\mathcal{U}_{approx}(\Theta)=Z_1(\Theta_{Z})\,\Psi_1(\Theta_\Psi)\,\Phi_1(\Theta_\Phi)=\\
    &=\prod_{i\in\mathcal{J}_Z}\exp\{i\;\theta_iU_i\}\prod_{j\in\mathcal{A}_\Psi}\exp\{i\;\theta_jU_j\}\;\prod_{x=1}^3\left[\left(\prod T_x^e\right)\left(\prod T_x^e\right)\prod_{k_e\in\mathcal{B}_{\Psi,x}}\exp\{i\;\theta_{k_e}U_{k_e}\}\left(\prod T_x^e\right)\left(\prod T_x^e\right)\right]\cdot\\
    &\cdot\prod_{x=1}^3\left[\left(\prod T_x^o\right)\left(\prod T_x^o\right)\prod_{k_o\in\mathcal{C}_{\Phi,x}}\exp\{i\;\theta_{k_o}U_{k_o}\}\left(\prod T_x^o\right)\left(\prod T_x^o\right)\right]=\\
    &=\prod_{i\in\mathcal{J}_Z}\exp\{i\;\theta_iU_i\}\prod_{j\in\mathcal{A}_\Psi}\exp\{i\;\theta_jU_j\}\prod_{k_e\in\mathcal{B}_{\Psi,1}}\exp\{i\;\theta_{k_e}U_{k_e}\}\prod_{l_e\in\mathcal{B}_{\Psi,2}}\exp\{i\;\theta_{l_e}U_{l_e}\}\prod_{m_e\in\mathcal{B}_{\Psi,3}}\exp\{i\;\theta_{m_e}U_{m_e}\}\cdot\\
    &\cdot\prod_{k_o\in\mathcal{C}_{\Phi,1}}\exp\{i\;\theta_{k_o}U_{k_o}\}\prod_{l_o\in\mathcal{C}_{\Phi,2}}\exp\{i\;\theta_{l_o}U_{l_o}\}\prod_{m_o\in\mathcal{C}_{\Phi,3}}\exp\{i\;\theta_{m_o}U_{m_o}\}
\end{split}
\end{equation}
The last equality is due to the properties of permutation matrices and reveals the crucial role of grouping (and also ordering). Finally, in Figure~\ref{fig:n3circuit} the overall quantum circuit is depicted from left to right, line by line, taking care to label its sub-components.
\begin{figure}[htbp]
    \begin{subfigure}[b]{1\textwidth}
        \centering
            \begin{quantikz}
            \lstick{0}&\ctrl{2}\gategroup[3,steps=4,style={dashed,rounded corners},label style={label position=below,anchor=north,yshift=-0.3cm}]{$\prod T_3^o$}&\targ{}&&\ctrl{2}&\gate[3]{M_3^o-circuit}&\ctrl{2}\gategroup[3,steps=4,style={dashed,rounded corners},label style={label position=below,anchor=north,yshift=-0.3cm}]{$\prod T_3^o$}&\targ[style={red}]{}&&\ctrl[style={red}]{2}&\ctrl[style={red}]{2}\gategroup[3,steps=3,style={dashed,rounded corners},label style={label position=below,anchor=north,yshift=-0.3cm}]{$\prod T_2^o$}&\targ[style={red}]{}&\ctrl{2}&\rstick{...}\\
            \lstick{1}&&&\targ{}&&&&&\targ{}&&&&&\rstick{...}\\
            \lstick{2}&\targ{}&\ctrl{-2}&\ctrl{-1}&\targ{}&&\targ{}&\ctrl[style={red}]{-2}&\ctrl{-1}&\targ[style={red}]{}&\targ[style={red}]{}&\ctrl[style={red}]{-2}&\targ{}&\rstick{...}
            \end{quantikz}        
    \end{subfigure}\vspace{5mm}
    \begin{subfigure}[b]{1\textwidth}
        \centering
            \begin{quantikz}
            \lstick{...}&\gate[3]{M_2^o-circuit}&\ctrl{2}\gategroup[3,steps=3,style={dashed,rounded corners},label style={label position=below,anchor=north,yshift=-0.3cm}]{$\prod T_2^o$}&\targ{}&\ctrl{2}&\gategroup[3,steps=3,style={dashed,rounded corners},label style={label position=below,anchor=north,yshift=-0.3cm}]{$\prod T_1^o$}&&&\gate[3]{M_1^o-circuit}&\gategroup[3,steps=3,style={dashed,rounded corners},label style={label position=below,anchor=north,yshift=-0.3cm}]{$\prod T_1^o$}&&&\rstick{...}\\
            \lstick{...}&&&&&\ctrl{1}&\targ{}&\ctrl{1}&&\ctrl{1}&\targ{}&\ctrl{1}&\rstick{...}\\
            \lstick{...}&&\targ{}&\ctrl{-2}&\targ{}&\targ{}&\ctrl{-1}&\targ{}&&\targ{}&\ctrl{-1}&\targ{}&\rstick{...}
            \end{quantikz}        
    \end{subfigure}\vspace{5mm}
    \begin{subfigure}[b]{1\textwidth}
        \centering
            \begin{quantikz}
            \lstick{...}&\targ{}\gategroup[3,steps=2,style={dashed,rounded corners},label style={label position=below,anchor=north,yshift=-0.3cm}]{$\prod T_3^e$}&&\gate[3]{M_3^e-circuit}&\targ[style={red}]{}\gategroup[3,steps=2,style={dashed,rounded corners},label style={label position=below,anchor=north,yshift=-0.3cm}]{$\prod T_3^e$}&&\targ[style={red}]{}\gategroup[3,steps=1,style={dashed,rounded corners},label style={label position=below,anchor=north,yshift=-0.3cm}]{$\prod T_2^e$}&\gate[3]{M_2^e-circuit}&\targ{}\gategroup[3,steps=1,style={dashed,rounded corners},label style={label position=below,anchor=north,yshift=-0.3cm}]{$\prod T_2^e$}&\gategroup[3,steps=1,style={dashed,rounded corners},label style={label position=below,anchor=north,yshift=-0.3cm}]{$\prod T_1^e$}&\rstick{...}\\
            \lstick{...}&&\targ{}&&&\targ{}&&&&\targ{}&\rstick{...}\\
            \lstick{...}&\ctrl{-2}&\ctrl{-1}&&\ctrl[style={red}]{-2}&\ctrl{-1}&\ctrl[style={red}]{-2}&&\ctrl{-2}&\ctrl{-1}&\rstick{...}
            \end{quantikz}        
    \end{subfigure}\vspace{5mm}
    \begin{subfigure}[b]{1\textwidth}
        \centering
            \begin{quantikz}
            \lstick{...}&\gate[3]{M_1^e-circuit}&\gategroup[3,steps=1,style={dashed,rounded corners},label style={label position=below,anchor=north,yshift=-0.3cm}]{$\prod T_1^e$}&&\gate[3]{M_3ZYZ}\gategroup[3,steps=1,style={dashed,rounded corners},label style={label position=below,anchor=north,yshift=-0.3cm}]{sub-factor A}&&\gate[3]{Z(\Theta_Z)}&\rstick{...}\\
            \lstick{...}&&\targ{}&&&&&\rstick{...}\\
            \lstick{...}&&\ctrl{-1}&&&&&\rstick{...}
            \end{quantikz}        
    \end{subfigure}
\caption{Quantum circuit for general 3-qubit special unitary operators; red gates are simplified gates.}
\label{fig:n3circuit}
\end{figure}

\section{Quantum circuit to approximate 4-qubit systems}\label{sec:4qubits}
As required by the recursive SRBB (Section~\ref{sec:recursiveSRBB}), to build the quantum circuit for approximating any 4-qubit operator, the matrix algebra of order 16 is necessary. Since for $n\geqslant3$ the properties of the algebraic elements guarantee the same scalable structure in terms of grouping and ordering, for $n=4$ only some intermediate steps will be shown to correctly manage the approximation formula in light of the increased number of its elements. It is important to note that these shortcuts summarize the circuit design steps for any value of $n$, making the new formulation extremely practical, as shown below.

The basis $\mathcal{B}^{(16)}$ is composed of 256 algebraic elements capable of parameterizing an arbitrary 4-qubit unitary operator. The set $\{B_j^{(16)}:1\leqslant j\leqslant255\}$ is a basis for the $su(16)$ matrix algebra whose elements satisfy the properties of hermiticity, unitarity, and zero trace. The last element of the basis is, as usual, $B_{256}^{(16)}=\mathbb{I}_{16}$ and it does not participate in the $su(16)$ matrix algebra. The general equation~(\ref{eqn:U_approx}) with $l=1$ groups the elements according to this list:
\begin{enumerate}
    \item[i)]16 diagonal elements belonging to the set $\mathcal{J}=\{m^2-1,\;2\leqslant m\leqslant16\}\cup\{256\}=\mathcal{J}_Z\cup\{256\}$;
    \item[ii)]16 elements with a $ZYZ$-decomposition, divided into pairs, belonging to the set\\$\mathcal{J}_\Psi=\{[(2m-1)^2,(4m^2-2m)],\;1\leqslant m\leqslant8\}$;
    \item[iii)]112 elements divided into 28 quadruples, each of which is associated with a precise 2-cycle of the set of permutations $P_{16}^{even}$, as shown in Table \ref{tab:2cyclesP16even}. For brevity, only the elements associated with the first row of Table~\ref{tab:2cyclesP16even} are reported in Table~\ref{tab:evenpairs_element_n4}.
    \item[iv)]112 elements divided into 28 quadruples, each of which is associated with a precise element of the set of permutations $P_{16}^{odd}$, as shown in Table \ref{tab:2cyclesP16odd}. Again, only the elements associated with the first row of Table~\ref{tab:2cyclesP16odd} are reported in Table~\ref{tab:oddpairs_element_n4}.
\end{enumerate}
\begin{table}[htbp]
    \centering
    \begin{tabular}{|c|c|c|c|c|c|c|}
    \hline
    $P_{(2,4)}$&$P_{(2,6)}$&$P_{(2,8)}$&$P_{(2,10)}$&$P_{(2,12)}$&$P_{(2,14)}$&$P_{(2,16)}$\\
    \hline
    &$P_{(4,6)}$&$P_{(4,8)}$&$P_{(4,10)}$&$P_{(4,12)}$&$P_{(4,14)}$&$P_{(4,16)}$\\
    \hline
    &&$P_{(6,8)}$&$P_{(6,10)}$&$P_{(6,12)}$&$P_{(6,14)}$&$P_{(6,16)}$\\
    \hline
    &&&$P_{(8,10)}$&$P_{(8,12)}$&$P_{(8,14)}$&$P_{(8,16)}$\\
    \hline
    &&&&$P_{(10,12)}$&$P_{(10,14)}$&$P_{(10,16)}$\\
    \hline
    &&&&&$P_{(12,14)}$&$P_{(12,16)}$\\
    \hline
    &&&&&&$P_{(14,16)}$\\
    \hline
    \end{tabular}
    \caption{Elements of the 2-cycle group $P_{16}^{even}$.}
    \label{tab:2cyclesP16even}
\end{table}
\begin{table}[htbp]
    \centering
    \begin{tabular}{||c|c|c||}
    \hline
    Pair&Functions&Elements\\
    \hline\hline
    (2,4)&$h_4(1)$, $f_4(1)$, $h_3(2)$, $f_3(2)$&10,13,4,6\\
    (2,6)&$h_6(1)$, $f_6(1)$, $h_5(2)$, $f_5(2)$&26,31,18,22\\
    (2,8)&$h_8(1)$, $f_8(1)$, $h_7(2)$, $f_7(2)$&50,57,38,44\\
    (2,10)&$h_{10}(1)$, $f_{10}(1)$, $h_9(2)$, $f_9(2)$&82,91,66,74\\
    (2,12)&$h_{12}(1)$, $f_{12}(1)$, $h_{11}(2)$, $f_{11}(2)$&122,133,102,112\\
    (2,14)&$h_{14}(1)$, $f_{14}(1)$, $h_{13}(2)$, $f_{13}(2)$&170,184,146,158\\
    (2,16)&$h_{16}(1)$, $f_{16}(1)$, $h_{15}(2)$, $f_{15}(2)$&226,241,198,212\\
    \hline
    \end{tabular}
    \caption{Association between algebraic elements and even (2-cycles) permutations.}
    \label{tab:evenpairs_element_n4}
\end{table}
\begin{table}[htbp]
    \centering
    \begin{tabular}{|c|c|c|c|c|c|c|}
    \hline
    $P_{(2,3)}$&$P_{(2,5)}$&$P_{(2,7)}$&$P_{(2,9)}$&$P_{(2,11)}$&$P_{(2,13)}$&$P_{(2,15)}$\\
    \hline
    &$P_{(4,5)}$&$P_{(4,7)}$&$P_{(4,9)}$&$P_{(4,11)}$&$P_{(4,13)}$&$P_{(4,15)}$\\
    \hline
    &&$P_{(6,7)}$&$P_{(6,9)}$&$P_{(6,11)}$&$P_{(6,13)}$&$P_{(6,15)}$\\
    \hline
    &&&$P_{(8,9)}$&$P_{(8,11)}$&$P_{(8,13)}$&$P_{(8,15)}$\\
    \hline
    &&&&$P_{(10,11)}$&$P_{(10,13)}$&$P_{(10,15)}$\\
    \hline
    &&&&&$P_{(12,13)}$&$P_{(12,15)}$\\
    \hline
    &&&&&&$P_{(14,15)}$\\
    \hline
    \end{tabular}
    \caption{Elements of the 2-cycle group $P_{16}^{odd}$.}
    \label{tab:2cyclesP16odd}
\end{table}
\begin{table}[htbp]
    \centering
    \begin{tabular}{||c|c|c||}
    \hline
    Pair&Functions&Elements\\
    \hline\hline
    (2,3)&$h_3(1)$, $f_3(1)$, $h_4(2)$, $f_4(2)$&5,7,11,14\\
    (2,5)&$h_5(1)$, $f_5(1)$, $h_6(2)$, $f_6(2)$&17,21,27,32\\
    (2,7)&$h_7(1)$, $f_7(1)$, $h_8(2)$, $f_8(2)$&37,43,51,58\\
    (2,9)&$h_{9}(1)$, $f_{9}(1)$, $h_{10}(2)$, $f_{10}(2)$&65,73,83,92\\
    (2,11)&$h_{11}(1)$, $f_{11}(1)$, $h_{12}(2)$, $f_{12}(2)$&101,111,123,134\\
    (2,13)&$h_{13}(1)$, $f_{13}(1)$, $h_{14}(2)$, $f_{14}(2)$&145,157,171,184\\
    (2,15)&$h_{15}(1)$, $f_{15}(1)$, $h_{16}(2)$, $f_{16}(2)$&197,211,227,242\\
    \hline
    \end{tabular}
    \caption{Association between algebraic elements and odd (2-cycles) permutations.}
    \label{tab:oddpairs_element_n4}
\end{table}
The transition to the SRBB is illustrated in Table~\ref{tab:SRBB_diagonals_n4} which shows the list of new diagonal elements generated from Pauli strings. 
\begin{table}[ht]
\centering
\begin{tabular}{||c|c|c|c||}
\hline
Decimal&Binary&String&Element\\
\hline\hline
0&0000&$\mathbb{I}_2\otimes\mathbb{I}_2\otimes\mathbb{I}_2\otimes\mathbb{I}_2$&256\\
1&0001&$\mathbb{I}_2\otimes\mathbb{I}_2\otimes\mathbb{I}_2\otimes\sigma_3$&3\\
2&0010&$\mathbb{I}_2\otimes\mathbb{I}_2\otimes\sigma_3\otimes\mathbb{I}_2$&8\\
3&0011&$\mathbb{I}_2\otimes\mathbb{I}_2\otimes\sigma_3\otimes\sigma_3$&15\\
4&0100&$\mathbb{I}_2\otimes\sigma_3\otimes\mathbb{I}_2\otimes\mathbb{I}_2$&24\\
5&0101&$\mathbb{I}_2\otimes\sigma_3\otimes\mathbb{I}_2\otimes\sigma_3$&35\\
6&0110&$\mathbb{I}_2\otimes\sigma_3\otimes\sigma_3\otimes\mathbb{I}_2$&48\\
7&0111&$\mathbb{I}_2\otimes\sigma_3\otimes\sigma_3\otimes\sigma_3$&63\\
8&1000&$\sigma_3\otimes\mathbb{I}_2\otimes\mathbb{I}_2\otimes\mathbb{I}_2$&80\\
9&1001&$\sigma_3\otimes\mathbb{I}_2\otimes\mathbb{I}_2\otimes\sigma_3$&99\\
10&1010&$\sigma_3\otimes\mathbb{I}_2\otimes\sigma_3\otimes\mathbb{I}_2$&120\\
11&1011&$\sigma_3\otimes\mathbb{I}_2\otimes\sigma_3\otimes\sigma_3$&143\\
12&1100&$\sigma_3\otimes\sigma_3\otimes\mathbb{I}_2\otimes\mathbb{I}_2$&168\\
13&1101&$\sigma_3\otimes\sigma_3\otimes\mathbb{I}_2\otimes\sigma_3$&195\\
14&1110&$\sigma_3\otimes\sigma_3\otimes\sigma_3\otimes\mathbb{I}_2$&224\\
15&1111&$\sigma_3\otimes\sigma_3\otimes\sigma_3\otimes\sigma_3$&255\\
\hline
\end{tabular}
\caption{The new sequence of diagonal basis elements that marks the transition to SRBB.}
\label{tab:SRBB_diagonals_n4}
\end{table}
Therefore, the complete SRBB of order 16 is identified as the set $\mathcal{U}^{(16)}=\{U_j^{(16)}:\;1\leqslant j\leqslant256\}$ and defined by
\begin{equation}
U^{(16)}_j=\left\{
\begin{array}{l}
\mbox{new diagonal elements if }j\in\mathcal{J}\\
B^{(16)}_j\mbox{ otherwise}
\end{array}\right.
\end{equation}
In the following subsections, the design of the quantum circuit for $n=4$ is explained starting from the analysis of its three main factors. In order to show how the scaling characteristics help in circuit design for any value of $n$, the analysis will follow the same procedure already described in Appendix~\ref{sec:3qubits}.

\subsection{Diagonal contributions}
In this subsection, the $Z$-factor~(\ref{eqn:def_Z}) responsible for diagonal contributions in the $n=4$ case is briefly analyzed to show the most relevant aspects of its design:
\begin{equation}
    Z(\Theta_Z)=\prod_{j\in\mathcal{J}_Z}\exp\left\{i\,\theta_jU_j^{(16)}\right\}
\end{equation}
where $\mathcal{J}_Z=\{m^2-1,\;2\leqslant m\leqslant16\}$ and $\Theta_Z$ collects the associated Lie parameters. With the aim of designing the corresponding circuit, the scheme already explained in Section~\ref{sec:diagonals_n2} is followed. With reference to Table~\ref{tab:SRBB_diagonals_n4}, Figure~\ref{fig:zeta_diagram_n4} shows the values of the parametric pairs $(m,m')$, where each line must be associated with one or more merged gates, following the scheme already explained in Section~\ref{sec:2qubits}: the result is depicted in Figure~\ref{fig:zeta_components_n4}. The sequence of diagonal elements that minimizes\footnote{Section~\ref{sec:scaling} explains how the CNOT-optimized circuit emerges from the scalable algorithm of simplifications.} the number of CNOTs is illustrated in Figure~\ref{fig:zeta_circuit_n4}, where red gates highlight the possible simplifications of gates.
\begin{figure}[htbp]
\centering
\tikzset{phase label/.append style={label position=above}}
\resizebox{0.5\textwidth}{!}{
    \begin{quantikz}
    &\push{\;\mathbb{I}_2\;}\gategroup[15,steps=1,style={thin}]{$1^{\circ}\,qubit\,(0)$}&&&\push{\;\mathbb{I}_2\;}\gategroup[15,steps=1,style={thin}]{$2^{\circ}\,qubit\,(1)$}&&&\push{\;\mathbb{I}_2\;}\gategroup[15,steps=1,style={thin}]{$3^{\circ}\,qubit\,(2)$}&&&\push{\;\sigma_3\;}\gategroup[15,steps=1,style={thin}]{$4^{\circ}\,qubit\,(3)$}&\rstick{$m=3$}\\
    &\push{\;\mathbb{I}_2\;}&&&\push{\;\mathbb{I}_2\;}&&&\push{\;\sigma_3\;}&&&\push{\;\mathbb{I}_2\;}&\rstick{$m=2$}\\
    &\push{\;\mathbb{I}_2\;}&&&\push{\;\mathbb{I}_2\;}&&&\push{\;\sigma_3\;}&&&\push{\;\sigma_3\;}&\rstick{$m=3,\,m'=2$}\\
    &\push{\;\mathbb{I}_2\;}&&&\push{\;\sigma_3\;}&&&\push{\;\mathbb{I}_2\;}&&&\push{\;\mathbb{I}_2\;}&\rstick{$m=1$}\\
    &\push{\;\mathbb{I}_2\;}&&&\push{\;\sigma_3\;}&&&\push{\;\mathbb{I}_2\;}&&&\push{\;\sigma_3\;}&\rstick{$m=3,\,m'=1$}\\
    &\push{\;\mathbb{I}_2\;}&&&\push{\;\sigma_3\;}&&&\push{\;\sigma_3\;}&&&\push{\;\mathbb{I}_2\;}&\rstick{$m=2,\,m'=1$}\\
    &\push{\;\mathbb{I}_2\;}&&&\push{\;\sigma_3\;}&&&\push{\;\sigma_3\;}&&&\push{\;\sigma_3\;}&\rstick{$m=3,\,m'=1,2$}\\
    &\push{\;\sigma_3\;}&&&\push{\;\mathbb{I}_2\;}&&&\push{\;\mathbb{I}_2\;}&&&\push{\;\mathbb{I}_2\;}&\rstick{$m=0$}\\
    &\push{\;\sigma_3\;}&&&\push{\;\mathbb{I}_2\;}&&&\push{\;\mathbb{I}_2\;}&&&\push{\;\sigma_3\;}&\rstick{$m=3,\,m'=0$}\\
    &\push{\;\sigma_3\;}&&&\push{\;\mathbb{I}_2\;}&&&\push{\;\sigma_3\;}&&&\push{\;\mathbb{I}_2\;}&\rstick{$m=2,\,m'=0$}\\
    &\push{\;\sigma_3\;}&&&\push{\;\mathbb{I}_2\;}&&&\push{\;\sigma_3\;}&&&\push{\;\sigma_3\;}&\rstick{$m=3,\,m'=0,2$}\\
    &\push{\;\sigma_3\;}&&&\push{\;\sigma_3\;}&&&\push{\;\mathbb{I}_2\;}&&&\push{\;\mathbb{I}_2\;}&\rstick{$m=1,\,m'=0$}\\
    &\push{\;\sigma_3\;}&&&\push{\;\sigma_3\;}&&&\push{\;\mathbb{I}_2\;}&&&\push{\;\sigma_3\;}&\rstick{$m=3,\,m'=0,1$}\\
    &\push{\;\sigma_3\;}&&&\push{\;\sigma_3\;}&&&\push{\;\sigma_3\;}&&&\push{\;\mathbb{I}_2\;}&\rstick{$m=2,\,m'=0,1$}\\
    &\push{\;\sigma_3\;}&&&\push{\;\sigma_3\;}&&&\push{\;\sigma_3\;}&&&\push{\;\sigma_3\;}&\rstick{$m=3,\,m'=0,1,2$}
    \end{quantikz}}
\caption{Diagram to find the position of CNOTs and rotation gates.}
\label{fig:zeta_diagram_n4}
\end{figure}

\begin{figure}[htbp]
\resizebox{0.19\textwidth}{!}{
    \begin{subfigure}[b]{0.2\textwidth}
    \centering
        \begin{quantikz}
        \lstick{0}&&\\
        \lstick{1}&&\\
        \lstick{2}&&\\
        \lstick{3}&\gate{R_z(\theta_3)}&
        \end{quantikz}
    \caption{$\exp\{i\,\theta_3U_3\}$}
    \end{subfigure}}
\hfill
\resizebox{0.19\textwidth}{!}{
    \begin{subfigure}[b]{0.2\textwidth}
    \centering
        \begin{quantikz}
        \lstick{0}&&\\
        \lstick{1}&&\\
        \lstick{2}&\gate{R_z(\theta_8)}&\\
        \lstick{3}&&
        \end{quantikz}
    \caption{$\exp\{i\,\theta_8U_8\}$}
    \end{subfigure}}
\hfill
\resizebox{0.28\textwidth}{!}{
    \begin{subfigure}[b]{0.3\textwidth}
    \centering
        \begin{quantikz}
        \lstick{0}&&&&\\
        \lstick{1}&&&&\\
        \lstick{2}&\ctrl{1}&&\ctrl{1}&\\
        \lstick{3}&\targ{}&\gate{R_z(\theta_{15})}&\targ{}&
        \end{quantikz}     
    \caption{$\exp\{i\,\theta_{15}U_{15}\}$}
    \end{subfigure}}
\hfill
\resizebox{0.19\textwidth}{!}{
    \begin{subfigure}[b]{0.2\textwidth}
    \centering
        \begin{quantikz}
        \lstick{0}&&\\
        \lstick{1}&\gate{R_z(\theta_{24})}&\\
        \lstick{2}&&\\
        \lstick{3}&&
        \end{quantikz}     
    \caption{$\exp\{i\,\theta_{24}U_{24}\}$}
    \end{subfigure}}\\\vspace{3mm}
\resizebox{0.28\textwidth}{!}{ 
    \begin{subfigure}[b]{0.3\textwidth}
    \centering
        \begin{quantikz}
        \lstick{0}&&&&\\
        \lstick{1}&\ctrl{2}&&\ctrl{2}&\\
        \lstick{2}&&&&\\
        \lstick{3}&\targ{}&\gate{R_z(\theta_{35})}&\targ{}&
        \end{quantikz}     
    \caption{$\exp\{i\,\theta_{35}U_{35}\}$}
    \end{subfigure}}
\hfill
\resizebox{0.28\textwidth}{!}{
    \begin{subfigure}[b]{0.3\textwidth}
    \centering
        \begin{quantikz}
        \lstick{0}&&&&\\
        \lstick{1}&\ctrl{1}&&\ctrl{1}&\\
        \lstick{2}&\targ{}&\gate{R_z(\theta_{48})}&\targ{}&\\
        \lstick{3}&&&&   
        \end{quantikz} 
    \caption{$\exp\{i\,\theta_{48}U_{48}\}$}
    \end{subfigure}}
\hfill
\resizebox{0.28\textwidth}{!}{
    \begin{subfigure}[b]{0.3\textwidth}
    \centering
        \begin{quantikz}
        \lstick{0}&\ctrl{2}&&\ctrl{2}&\\
        \lstick{1}&&&&\\
        \lstick{2}&\targ{}&\gate{R_z(\theta_{120})}&\targ{}&\\
        \lstick{3}&&&&   
        \end{quantikz} 
    \caption{$\exp\{i\,\theta_{120}U_{120}\}$}
    \end{subfigure}}\\\vspace{3mm}
\resizebox{0.36\textwidth}{!}{  
    \begin{subfigure}[b]{0.4\textwidth}
    \centering
        \begin{quantikz}
        \lstick{0}&&&&&&\\
        \lstick{1}&\ctrl{2}&&&&\ctrl{2}&\\
        \lstick{2}&&\ctrl{1}&&\ctrl{1}&&\\
        \lstick{3}&\targ{}&\targ{}&\gate{R_z(\theta_{63})}&\targ{}&\targ{}&
        \end{quantikz}     
    \caption{$\exp\{i\,\theta_{63}U_{63}\}$}
    \end{subfigure}}
\hfill
\resizebox{0.19\textwidth}{!}{
    \begin{subfigure}[b]{0.2\textwidth}
    \centering
        \begin{quantikz}
        \lstick{0}&\gate{R_z(\theta_{80})}&\\
        \lstick{1}&&\\
        \lstick{2}&&\\
        \lstick{3}&&
        \end{quantikz}
    \caption{$\exp\{i\,\theta_{80}U_{80}\}$}
    \end{subfigure}}
\hfill
\resizebox{0.28\textwidth}{!}{
    \begin{subfigure}[b]{0.3\textwidth}
    \centering
        \begin{quantikz}
        \lstick{0}&\ctrl{3}&&\ctrl{3}&\\
        \lstick{1}&&&&\\
        \lstick{2}&&&&\\
        \lstick{3}&\targ{}&\gate{R_z(\theta_{99})}&\targ{}&
        \end{quantikz}     
    \caption{$\exp\{i\,\theta_{99}U_{99}\}$}
    \end{subfigure}}\\\vspace{3mm}
\resizebox{0.44\textwidth}{!}{
    \begin{subfigure}[b]{0.5\textwidth}
    \centering
        \begin{quantikz}
        \lstick{0}&\ctrl{3}&&&&\ctrl{3}&\\
        \lstick{1}&&&&&&\\
        \lstick{2}&&\ctrl{1}&&\ctrl{1}&&\\
        \lstick{3}&\targ{}&\targ{}&\gate{R_z(\theta_{143})}&\targ{}&\targ{}&
        \end{quantikz}     
    \caption{$\exp\{i\,\theta_{143}U_{143}\}$}
    \end{subfigure}}
\resizebox{0.36\textwidth}{!}{
    \begin{subfigure}[b]{0.4\textwidth}
    \centering
        \begin{quantikz}
        \lstick{0}&\ctrl{1}&&\ctrl{1}&\\
        \lstick{1}&\targ{}&\gate{R_z(\theta_{168})}&\targ{}&\\
        \lstick{2}&&&&\\
        \lstick{3}&&&&
        \end{quantikz} 
    \caption{$\exp\{i\,\theta_{168}U_{168}\}$}
    \end{subfigure}}\\\vspace{3mm}
\resizebox{0.44\textwidth}{!}{
    \begin{subfigure}[b]{0.5\textwidth}
    \centering
        \begin{quantikz}
        \lstick{0}&\ctrl{3}&&&&\ctrl{3}&\\
        \lstick{1}&&\ctrl{2}&&\ctrl{2}&&\\
        \lstick{2}&&&&&&\\
        \lstick{3}&\targ{}&\targ{}&\gate{R_z(\theta_{195})}&\targ{}&\targ{}&
        \end{quantikz}     
    \caption{$\exp\{i\,\theta_{195}U_{195}\}$}
    \end{subfigure}}
\resizebox{0.44\textwidth}{!}{
    \begin{subfigure}[b]{0.5\textwidth}
    \centering
        \begin{quantikz}
        \lstick{0}&\ctrl{2}&&&&\ctrl{2}&\\
        \lstick{1}&&\ctrl{1}&&\ctrl{1}&&\\
        \lstick{2}&\targ{}&\targ{}&\gate{R_z(\theta_{224})}&\targ{}&\targ{}&\\
        \lstick{3}&&&&&&
        \end{quantikz}     
    \caption{$\exp\{i\,\theta_{224}U_{224}\}$}
    \end{subfigure}}\\\vspace{3mm}
\resizebox{0.52\textwidth}{!}{
    \begin{subfigure}[b]{0.6\textwidth}
    \centering
        \begin{quantikz}
        \lstick{0}&\ctrl{3}&&&&&&\ctrl{3}&\\
        \lstick{1}&&\ctrl{2}&&&&\ctrl{2}&&\\
        \lstick{2}&&&\ctrl{1}&&\ctrl{1}&&&\\
        \lstick{3}&\targ{}&\targ{}&\targ{}&\gate{R_z(\theta_{255})}&\targ{}&\targ{}&\targ{}&
        \end{quantikz}     
    \caption{$\exp\{i\,\theta_{255}U_{255}\}$}
    \end{subfigure}}
\caption{Fifteen little blocks that make up the $Z$-factor circuit for $n=4$.}
\label{fig:zeta_components_n4}
\end{figure}

\begin{figure}[htbp]
    \resizebox{1\textwidth}{!}{
    \begin{subfigure}[b]{1.1\textwidth}
        \centering
        \begin{quantikz}
        \lstick{0}&&&&&&&&&&&&\ctrl{3}&&&\rstick{...}\\
        \lstick{1}&&&&\ctrl{2}&&&&\ctrl[style={red}]{2}&\ctrl[style={red}]{2}&&\ctrl[style={red}]{2}&&\ctrl[style={red}]{2}&&\rstick{...}\\
        \lstick{2}&\ctrl{1}&&\ctrl[style={red}]{1}&&\ctrl[style={red}]{1}&&\ctrl{1}&&&&&&&&\rstick{...}\\
        \lstick{3}&\targ{}&\gate{R_z(\theta_{})}&\targ[style={red}]{}&\targ{}&\targ[style={red}]{}&\gate{R_z(\theta_{})}&\targ{}&\targ[style={red}]{}&\targ[style={red}]{}&\gate{R_z(\theta_{})}&\targ[style={red}]{}&\targ{}&\targ[style={red}]{}&\gate{R_z(\theta_{})}&\rstick{...}
        \end{quantikz}
    \end{subfigure}}
    \resizebox{1\textwidth}{!}{
    \begin{subfigure}[b]{1.1\textwidth}
    \centering
    \begin{quantikz}
    \lstick{...}&&\ctrl[style={red}]{3}&\ctrl[style={red}]{3}&&&&&&\ctrl[style={red}]{3}&\ctrl[style={red}]{3}&&&\rstick{...}\\
    \lstick{...}&\ctrl[style={red}]{2}&&&\ctrl[style={red}]{2}&&&&\ctrl{2}&&&&&\rstick{...}\\
    \lstick{...}&&&&&\ctrl{1}&&\ctrl[style={red}]{1}&&&&\ctrl[style={red}]{1}&&\rstick{...}\\
    \lstick{...}&\targ[style={red}]{}&\targ[style={red}]{}&\targ[style={red}]{}&\targ[style={red}]{}&\targ{}&\gate{R_z(\theta_{})}&\targ[style={red}]{}&\targ{}&\targ[style={red}]{}&\targ[style={red}]{}&\targ[style={red}]{}&\gate{R_z(\theta_{})}&\rstick{...}\\
    \end{quantikz}
    \end{subfigure}}
    \resizebox{1\textwidth}{!}{
    \begin{subfigure}[b]{1.1\textwidth}
    \centering
    \begin{quantikz}
    \lstick{...}&&\ctrl[style={red}]{3}&\ctrl[style={red}]{3}&&\ctrl{3}&\gate[3]{Z(\Theta_Z)_{n=3}}&\rstick{...}\\
    \lstick{...}&&&&&&&\rstick{...}\\
    \lstick{...}&\ctrl{1}&&&&&&\rstick{...}\\
    \lstick{...}&\targ{}&\targ[style={red}]{}&\targ[style={red}]{}&\gate{R_z(\theta_{})}&\targ{}&\gate{R_z(\theta_{})}&\rstick{...}\\
    \end{quantikz}
    \end{subfigure}}
\caption{$Z$-factor quantum circuit for $n=4$; the highlighted gates simplify each other.}
\label{fig:zeta_circuit_n4}
\end{figure}

\subsection{Even contributions}\label{sec:even_4qubits}
The $\Psi$-factor~(\ref{eqn:def_Psi}) for a general 4-qubit system is defined by:
\begin{equation}
    \Psi(\Theta_\Psi)=\left[\prod_{(i,j)\in\mathcal{A}_\Psi}\exp\{i\,\theta_iU_i\}\exp\{i\,\theta_jU_j\}\right]\prod_{x=1}^{7}\left(\prod T_x^e\right)M_x^e\left(\prod T_x^e\right)
\end{equation}
where, according to the scalable structure, two different types of subsets of SRBB elements are involved:
\begin{itemize}
    \item a set of pairs called $\mathcal{A}_{\Psi}$ makes up the first sub-factor of type $A$ enclosed in squared brackets,
        \begin{equation}
        \begin{split}
        \mathcal{A}_{\Psi}&=\{[(2m-1)^2,(4m^2-2m)],\;1\leqslant m\leqslant8)\}=\\
        &=\{(1,2),(9,12),(25,30),(49,56),(81,90),(121,132),(169,182),(225,240)\}
        \end{split}
        \end{equation}
    \item a set of quadruples called $\mathcal{B}_{\Psi}$, divided into seven subsets $\mathcal{B}_{\Psi,x}$, makes up the second sub-factor of type $B$,
    \begin{equation}
    \begin{array}{ll}
        \mathcal{B}_{\Psi,1}=\left\{
            \begin{array}{l}
            \left[h_{4}(1),f_{4}(1),h_{3}(2),f_{3}(2)\right]\\
            \left[h_{16}(13),f_{16}(13),h_{15}(14),f_{15}(14)\right]\\
            \left[h_{16}(11),f_{16}(11),h_{15}(12),f_{15}(12)\right]\\
            \left[h_{14}(11),f_{14}(11),h_{13}(12),f_{13}(12)\right]
            \end{array}\right.
            &
            \mathcal{B}_{\Psi,2}=\left\{
            \begin{array}{l}
            \left[h_{6}(1),f_{6}(1),h_{5}(2),f_{5}(2)\right]\\
            \left[h_{16}(9),f_{16}(9),h_{15}(10),f_{15}(10)\right]\\
            \left[h_{14}(9),f_{14}(9),h_{13}(10),f_{13}(10)\right]\\
            \left[h_{12}(9),f_{12}(9),h_{11}(10),f_{11}(10)\right]
            \end{array}\right.\\
            &\\
        \mathcal{B}_{\Psi,3}=\left\{
            \begin{array}{l}
            \left[h_{8}(1),f_{8}(1),h_{7}(2),f_{7}(2)\right]\\
            \left[h_{16}(5),f_{16}(5),h_{15}(6),f_{15}(6)\right]\\
            \left[h_{14}(5),f_{14}(5),h_{13}(6),f_{13}(6)\right]\\
            \left[h_{12}(5),f_{12}(5),h_{11}(6),f_{11}(6)\right]
            \end{array}\right.
            &
            \mathcal{B}_{\Psi,4}=\left\{
            \begin{array}{l}
            \left[h_{10}(1),f_{10}(1),h_{9}(2),f_{9}(2)\right]\\
            \left[h_{16}(7),f_{16}(7),h_{15}(8),f_{15}(8)\right]\\
            \left[h_{14}(7),f_{14}(7),h_{13}(8),f_{13}(8)\right]\\
            \left[h_{12}(7),f_{12}(7),h_{11}(8),f_{11}(8)\right]
            \end{array}\right.\\
            &\\
        \mathcal{B}_{\Psi,5}=\left\{
            \begin{array}{l}
            \left[h_{12}(1),f_{12}(1),h_{11}(2),f_{11}(2)\right]\\
            \left[h_{10}(7),f_{10}(7),h_{9}(8),f_{9}(8)\right]\\
            \left[h_{10}(5),f_{10}(5),h_{9}(6),f_{9}(6)\right]\\
            \left[h_{8}(5),f_{8}(5),h_{7}(6),f_{7}(6)\right]
            \end{array}\right.
            &
            \mathcal{B}_{\Psi,6}=\left\{
            \begin{array}{l}
            \left[h_{14}(1),f_{14}(1),h_{13}(2),f_{13}(2)\right]\\
            \left[h_{16}(3),f_{16}(3),h_{15}(4),f_{15}(4)\right]\\
            \left[h_{12}(3),f_{12}(3),h_{11}(4),f_{11}(4)\right]\\
            \left[h_{10}(3),f_{10}(3),h_{9}(4),f_{9}(4)\right]
            \end{array}\right.\\
            &\\
        \mathcal{B}_{\Psi,7}=\left\{
            \begin{array}{l}
            \left[h_{16}(1),f_{16}(1),h_{15}(2),f_{15}(2)\right]\\
            \left[h_{14}(3),f_{14}(3),h_{13}(4),f_{13}(4)\right]\\
            \left[h_{8}(3),f_{8}(3),h_{7}(4),f_{7}(4)\right]\\
            \left[h_{6}(3),f_{6}(3),h_{5}(4),f_{5}(4)\right]
            \end{array}\right.
            &
    \end{array}
    \end{equation}
\end{itemize}
Furthermore, with reference to Table~\ref{tab:2cyclesP16even} and to the partition procedure described in Section~\ref{sec:even_3qubits}, it follows that
\begin{equation}
\begin{tabular}{l}
    $\prod T_1^e=P_{(2,4)}P_{(14,16)}P_{(12,16)}P_{(12,14)}$\\
    $\prod T_2^e=P_{(2,6)}P_{(10,16)}P_{(10,14)}P_{(10,12)}$\\
    $\prod T_3^e=P_{(2,8)}P_{(6,16)}P_{(6,14)}P_{(6,12)}$\\
    $\prod T_4^e=P_{(2,10)}P_{(8,16)}P_{(8,14)}P_{(8,12)}$\\
    $\prod T_5^e=P_{(2,12)}P_{(8,10)}P_{(6,10)}P_{(6,8)}$\\
    $\prod T_6^e=P_{(2,14)}P_{(4,16)}P_{(4,12)}P_{(4,10)}$\\
    $\prod T_7^e=P_{(2,16)}P_{(4,14)}P_{(4,8)}P_{(4,6)}$
\end{tabular}
\end{equation}
Moreover\footnote{For practical reasons, only $M_{1,2,3}^e$ are reported below given that the construction of the remaining ones is identical. To avoid repetitions and write compact formulas, we omit repeating the subscript for the $\theta_i$ angles, which will still be equal to that of the associated basis elements.}
{\footnotesize
\begin{equation}
\begin{aligned}
    M_1^e&=\prod T_1^e\left[\prod_{
    \begin{tabular}{cc}
    (2,4),\!\!\!\!\!\!&\!\!\!(14,16),\\
    (12,16),\!\!\!\!\!&(12,14) 
    \end{tabular}
    }\exp\{i\,\theta\,U_{h_\beta(\alpha-1)}\}\exp\{i\,\theta\,U_{f_\beta(\alpha-1)}\}\exp\{i\,\theta\,U_{h_{\beta-1}(\alpha)}\exp\{i\,\theta\,U_{f_{\beta-1}(\alpha)}\}\right]\prod T_1^e=\\
    &=\prod T_1^e\left[\exp\{i\,\theta\,U_{h_4(1)}\}\exp\{i\,\theta\,U_{f_4(1)}\}\exp\{i\,\theta\,U_{h_3(2)}\}\exp\{i\,\theta\,U_{f_3(2)}\}\right.\cdot\\
    &\cdot\exp\{i\,\theta\,U_{h_{16}(13)}\}\exp\{i\,\theta\,U_{f_{16}(13)}\}\exp\{i\,\theta\,U_{h_{15}(14)}\}\exp\{i\,\theta\,U_{f_{15}(14)}\}\cdot\\
    &\cdot\exp\{i\,\theta\,U_{h_{16}(11)}\}\exp\{i\,\theta\,U_{f_{16}(11)}\}\exp\{i\,\theta\,U_{h_{15}(12)}\}\exp\{i\,\theta\,U_{f_{15}(12)}\}\cdot\\
    &\cdot\left.\exp\{i\,\theta\,U_{h_{14}(11)}\}\exp\{i\,\theta\,U_{f_{14}(11)}\}\exp\{i\,\theta\,U_{h_{13}(12)}\}\exp\{i\,\theta\,U_{f_{13}(12)}\}\right]\prod T_1^e
\end{aligned}
\end{equation}
\begin{equation}
\begin{aligned}
    M_2^e&=\prod T_2^e\left[\prod_{
    \begin{tabular}{cc}
    (2,6),\!\!\!\!\!\!&\!\!\!(10,16),\\
    (10,14),\!\!\!\!\!&(10,12) 
    \end{tabular}
    }\exp\{i\,\theta\,U_{h_\beta(\alpha-1)}\}\exp\{i\,\theta\,U_{f_\beta(\alpha-1)}\}\exp\{i\,\theta\,U_{h_{\beta-1}(\alpha)}\exp\{i\,\theta\,U_{f_{\beta-1}(\alpha)}\}\right]\prod T_2^e=\\
    &=\prod T_2^e\left[\exp\{i\,\theta\,U_{h_6(1)}\}\exp\{i\,\theta\,U_{f_6(1)}\}\exp\{i\,\theta\,U_{h_5(2)}\exp\{i\,\theta\,U_{f_5(2)}\}\right.\cdot\\
    &\cdot\exp\{i\,\theta\,U_{h_{16}(9)}\}\exp\{i\,\theta\,U_{f_{16}(9)}\}\exp\{i\,\theta\,U_{h_{15}(10)}\}\exp\{i\,\theta\,U_{f_{15}(10)}\}\cdot\\
    &\cdot\exp\{i\,\theta\,U_{h_{14}(9)}\}\exp\{i\,\theta\,U_{f_{14}(9)}\}\exp\{i\,\theta\,U_{h_{13}(10)}\}\exp\{i\,\theta\,U_{f_{13}(10)}\}\cdot\\
    &\cdot\left.\exp\{i\,\theta\,U_{h_{12}(9)}\}\exp\{i\,\theta\,U_{f_{12}(9)}\}\exp\{i\,\theta\,U_{h_{11}(10)}\}\exp\{i\,\theta\,U_{f_{11}(10)}\}\right]\prod T_2^e
\end{aligned}
\end{equation}
\begin{equation}
\begin{aligned}
    M_3^e&=\prod T_3^e\left[\prod_{
    \begin{tabular}{cc}
    (2,8),\!\!\!\!\!\!&\!\!(6,16),\\
    (6,14),\!\!\!\!\!&(6,12) 
    \end{tabular}
    }\exp\{i\,\theta\,U_{h_\beta(\alpha-1)}\}\exp\{i\,\theta\,U_{f_\beta(\alpha-1)}\}\exp\{i\,\theta\,U_{h_{\beta-1}(\alpha)}\exp\{i\,\theta\,U_{f_{\beta-1}(\alpha)}\}\right]\prod T_3^e=\\
    &=\prod T_3^e\left[\exp\{i\,\theta\,U_{h_8(1)}\}\exp\{i\,\theta\,U_{f_8(1)}\}\exp\{i\,\theta\,U_{h_7(2)}\exp\{i\,\theta\,U_{f_7(2)}\}\right.\cdot\\
    &\cdot\exp\{i\,\theta\,U_{h_{16}(5)}\}\exp\{i\,\theta\,U_{f_{16}(5)}\}\exp\{i\,\theta\,U_{h_{15}(6)}\}\exp\{i\,\theta\,U_{f_{15}(6)}\}\cdot\\
    &\cdot\exp\{i\,\theta\,U_{h_{14}(5)}\}\exp\{i\,\theta\,U_{f_{14}(5)}\}\exp\{i\,\theta\,U_{h_{13}(6)}\}\exp\{i\,\theta\,U_{f_{13}(6)}\}\cdot\\
    &\cdot\left.\exp\{i\,\theta\,U_{h_{12}(5)}\}\exp\{i\,\theta\,U_{f_{12}(5)}\}\exp\{i\,\theta\,U_{h_{11}(6)}\}\exp\{i\,\theta\,U_{f_{11}(6)}\}\right]\prod T_3^e
\end{aligned}  
\end{equation}}
Accordingly, the first sub-factor of type $A$ and the factors called $M_x^e$ into the second sub-factor of type B become a $SU(2)$-block diagonal matrix, thus composing as a whole two unitary matrices belonging to $M_4ZYZ$. In both cases, the corresponding quantum circuit is given by Figure~\ref{fig:zyz_circuit_n4} (the rotation parameters are renamed as $\theta_j^*$). The remaining seven pairs of $\prod T_x^e$ are implemented as described in Appendix \ref{sec:cnot_sequences}; the result is illustrated in Figure~\ref{fig:prodTxe_diagram_n4}.
\begin{figure}[htbp]
\resizebox{0.9\textwidth}{!}{
    \begin{subfigure}[b]{1\textwidth}
        \centering
        $F_a(\theta_{1-8}^*)=$
        \begin{quantikz}
            \lstick{0}&&&&&&&&\ctrl{3}&\rstick{...}\\
            \lstick{1}&&&&\ctrl{2}&&&&&\rstick{...}\\
            \lstick{2}&&\ctrl{1}&&&&\ctrl{1}&&&\rstick{...}\\
            \lstick{3}&\gate{R_a(\theta_1^*)}&\targ{}&\gate{R_a(\theta_2^*)}&\targ{}&\gate{R_a(\theta_3^*)}&\targ{}&\gate{R_a(\theta_4^*)}&\targ{}&\rstick{...}
        \end{quantikz}
    \end{subfigure}}
    \resizebox{0.9\textwidth}{!}{
    \begin{subfigure}[b]{1\textwidth}
        \centering
        \begin{quantikz}
            \lstick{...}&&&&&&&&\\
            \lstick{...}&&&&\ctrl{2}&&&&\\
            \lstick{...}&&\ctrl{1}&&&&\ctrl{1}&&\\
            \lstick{...}&\gate{R_a(\theta_5^*)}&\targ{}&\gate{R_a(\theta_6^*)}&\targ{}&\gate{R_a(\theta_7^*)}&\targ{}&\gate{R_a(\theta_8^*)}&
        \end{quantikz}
    \end{subfigure}}\vspace{3mm}
    \resizebox{0.9\textwidth}{!}{
    \begin{subfigure}[b]{1\textwidth}
        \centering
        $M_4ZYZ=$
        \begin{quantikz}
            \lstick{0}&\gate[4]{F_z(\theta_{1-8}^*)}&\gate[4]{F_y(\theta_{9-16}^*)}&\gate[4]{F_z(\theta_{17-24}^*)}&\ctrl{3}&\\
            \lstick{1}&&&&&\\
            \lstick{2}&&&&&\\
            \lstick{3}&&&&\targ{}&
        \end{quantikz}
    \end{subfigure}}    
    \caption{$M_4ZYZ$ quantum circuit.}
    \label{fig:zyz_circuit_n4}
\end{figure}

\begin{figure}[htbp]
    \begin{subfigure}[c]{1\textwidth}
    \centering
        \begin{tabular}{||c|c|c|c||}
        \hline
        x&Bits&Binary&Control-Target\\
        \hline\hline
        1&$x_0x_1x_2$&001&(3,2)\\
        2&$x_0x_1x_2$&010&(3,1)\\
        3&$x_0x_1x_2$&011&(3,1),(3,2)\\
        4&$x_0x_1x_2$&100&(3,0)\\
        5&$x_0x_1x_2$&101&(3,0),(3,2)\\
        6&$x_0x_1x_2$&110&(3,0)(3,1)\\
        7&$x_0x_1x_2$&111&(3,0),(3,1),(3,2)\\
        \hline
        \end{tabular}   
    \end{subfigure}\vspace{5mm}
    \begin{subfigure}[c]{0.2\textwidth}
    \centering
        \begin{quantikz}
        \lstick{0}&&\\
        \lstick{1}&&\\
        \lstick{2}&\targ{}&\\
        \lstick{3}&\ctrl{-1}&
        \end{quantikz}
        \caption{$\prod T_1^e$}
    \end{subfigure}
    \begin{subfigure}[c]{0.2\textwidth}
    \centering
        \begin{quantikz}
        \lstick{0}&&\\
        \lstick{1}&\targ{}&\\
        \lstick{2}&&\\
        \lstick{3}&\ctrl{-2}&
        \end{quantikz}
        \caption{$\prod T_2^e$}
    \end{subfigure}
    \begin{subfigure}[c]{0.2\textwidth}
    \centering
        \begin{quantikz}
        \lstick{0}&&&\\
        \lstick{1}&\targ{}&&\\
        \lstick{2}&&\targ{}&\\
        \lstick{3}&\ctrl{-2}&\ctrl{-1}&
        \end{quantikz}
        \caption{$\prod T_3^e$}
    \end{subfigure}
    \begin{subfigure}[c]{0.2\textwidth}
    \centering
        \begin{quantikz}
        \lstick{0}&\targ{}&\\
        \lstick{1}&&\\
        \lstick{2}&&\\
        \lstick{3}&\ctrl{-3}&
        \end{quantikz}
        \caption{$\prod T_4^e$}
    \end{subfigure}\vspace{0.5cm}
    \begin{subfigure}[c]{0.2\textwidth}
    \centering
        \begin{quantikz}
        \lstick{0}&\targ{}&&\\
        \lstick{1}&&&\\
        \lstick{2}&&\targ{}&\\
        \lstick{3}&\ctrl{-3}&\ctrl{-1}&
        \end{quantikz}
        \caption{$\prod T_5^e$}
    \end{subfigure}\hfill
    \begin{subfigure}[c]{0.2\textwidth}
    \centering
        \begin{quantikz}
        \lstick{0}&\targ{}&&\\
        \lstick{1}&&\targ{}&\\
        \lstick{2}&&&\\
        \lstick{3}&\ctrl{-3}&\ctrl{-2}&
        \end{quantikz}
        \caption{$\prod T_6^e$}
    \end{subfigure}\hfill
    \begin{subfigure}[c]{0.3\textwidth}
    \centering
        \begin{quantikz}
        \lstick{0}&\targ{}&&&\\
        \lstick{1}&&\targ{}&&\\
        \lstick{2}&&&\targ{}&\\
        \lstick{3}&\ctrl{-3}&\ctrl{-2}&\ctrl{-1}&
        \end{quantikz}
        \caption{$\prod T_7^e$}
    \end{subfigure}
\caption{Diagram to find the control-target pairs for the CNOTs representing $\prod T_x^e$.}
\label{fig:prodTxe_diagram_n4}
\end{figure}

\subsection{Odd contributions}
The $\Phi$-factor~(\ref{eqn:def_Phi}) in the generale case $n=4$,
\begin{equation}
    \Phi(\Theta_\Phi)=\prod_{x=1}^7\left(\prod T_x^o\right)M_x^o\left(\prod T_x^o\right)
\end{equation}
is composed starting from the subset $\mathcal{C}_\Phi$, divided into seven subsets $C_{\Phi,x}$ of SRBB elements:
\begin{equation}
\begin{array}{ll}\mathcal{C}_{\Phi,1}=\left\{
        \begin{array}{l}
        \left[h_{3}(1),f_{3}(1),h_{4}(2),f_{4}(2)\right]\\
        \left[h_{15}(13),f_{15}(13),h_{16}(14),f_{16}(14)\right]\\
        \left[h_{15}(11),f_{15}(11),h_{16}(12),f_{16}(12)\right]\\
        \left[h_{13}(11),f_{13}(11),h_{14}(12),f_{14}(12)\right]
        \end{array}\right.
        &
        \mathcal{C}_{\Phi,2}=\left\{
        \begin{array}{l}
        \left[h_{5}(1),f_{5}(1),h_{6}(2),f_{6}(2)\right]\\
        \left[h_{15}(9),f_{15}(9),h_{16}(10),f_{16}(10)\right]\\
        \left[h_{13}(9),f_{13}(9),h_{14}(10),f_{14}(10)\right]\\
        \left[h_{11}(9),f_{11}(9),h_{12}(10),f_{12}(10)\right]
        \end{array}\right.\\
        &\\
    \mathcal{C}_{\Phi,3}=\left\{
        \begin{array}{l}
        \left[h_{7}(1),f_{7}(1),h_{8}(2),f_{8}(2)\right]\\
        \left[h_{15}(7),f_{15}(7),h_{16}(8),f_{16}(8)\right]\\
        \left[h_{13}(7),f_{13}(7),h_{14}(8),f_{14}(8)\right]\\
        \left[h_{11}(7),f_{11}(7),h_{12}(8),f_{12}(8)\right]
        \end{array}\right.
        &
        \mathcal{C}_{\Phi,4}=\left\{
        \begin{array}{l}
        \left[h_{9}(1),f_{9}(1),h_{10}(2),f_{10}(2)\right]\\
        \left[h_{15}(5),f_{15}(5),h_{16}(6),f_{16}(6)\right]\\
        \left[h_{13}(5),f_{13}(5),h_{14}(6),f_{14}(6)\right]\\
        \left[h_{11}(5),f_{11}(5),h_{12}(6),f_{12}(6)\right]
        \end{array}\right.\\
        &\\
    \mathcal{C}_{\Phi,5}=\left\{
        \begin{array}{l}
        \left[h_{11}(1),f_{11}(1),h_{12}(2),f_{12}(2)\right]\\
        \left[h_{9}(7),f_{9}(7),h_{10}(8),f_{10}(8)\right]\\
        \left[h_{9}(5),f_{9}(5),h_{10}(6),f_{10}(6)\right]\\
        \left[h_{7}(5),f_{7}(5),h_{8}(6),f_{8}(6)\right]
        \end{array}\right.
        &
        \mathcal{C}_{\Phi,6}=\left\{
        \begin{array}{l}
        \left[h_{13}(1),f_{13}(1),h_{14}(2),f_{14}(2)\right]\\
        \left[h_{15}(3),f_{15}(3),h_{16}(4),f_{16}(4)\right]\\
        \left[h_{11}(3),f_{11}(3),h_{12}(4),f_{12}(4)\right]\\
        \left[h_{9}(3),f_{9}(3),h_{10}(4),f_{10}(4)\right]
        \end{array}\right.\\
        &\\
    \mathcal{C}_{\Phi,7}=\left\{
        \begin{array}{l}
        \left[h_{15}(1),f_{15}(1),h_{16}(2),f_{16}(2)\right]\\
        \left[h_{13}(3),f_{13}(3),h_{14}(4),f_{14}(4)\right]\\
        \left[h_{7}(3),f_{7}(3),h_{8}(4),f_{8}(4)\right]\\
        \left[h_{5}(3),f_{5}(3),h_{6}(4),f_{6}(4)\right]
        \end{array}\right.
        &
\end{array}
\end{equation}
As for the even contributions of Section~\ref{sec:even_4qubits}, Table~\ref{tab:2cyclesP16odd} helps to partition the permutation group $P_{16}^{odd}$ as follows:
\begin{equation}
\begin{tabular}{l}
    $\prod T_1^o=P_{(2,3)}P_{(14,15)}P_{(12,15)}P_{(12,13)}$\\
    $
    \prod T_2^o=P_{(2,5)}P_{(10,15)}P_{(10,13)}P_{(10,11)}$\\
    $\prod T_3^o=P_{(2,7)}P_{(8,15)}P_{(8,13)}P_{(8,11)}$\\$
    \prod T_4^o=P_{(2,9)}P_{(6,15)}P_{(6,13)}P_{(6,11)}$\\
    $\prod T_5^o=P_{(2,11)}P_{(8,9)}P_{(6,9)}P_{(6,7)}$\\$
    \prod T_6^o=P_{(2,13)}P_{(4,15)}P_{(4,11)}P_{(4,9)}$\\
    $\prod T_7^o=P_{(2,15)}P_{(4,13)}P_{(4,7)}P_{(4,5)}$
\end{tabular}
\end{equation}
Moreover,
{\footnotesize
\begin{equation}\label{eqn:M1odd_n4}
\begin{aligned}
    M_1^o&=\prod T_1^o\left[\prod_{
    \begin{tabular}{cc}
    (2,3),\!\!\!\!\!\!&\!\!\!(14,15),\\
    (12,15),\!\!\!\!\!&(12,13) 
    \end{tabular}
    }\exp\{i\,\theta\,U_{h_\beta(\alpha-1)}\}\exp\{i\,\theta\,U_{f_\beta(\alpha-1)}\}\exp\{i\,\theta\,U_{h_{\beta+1}(\alpha)}\exp\{i\,\theta\,U_{f_{\beta+1}(\alpha)}\}\right]\prod T_1^o=\\
    &=\prod T_1^o\left[\exp\{i\,\theta\,U_{h_3(1)}\}\exp\{i\,\theta\,U_{f_3(1)}\}\exp\{i\,\theta\,U_{h_4(2)}\exp\{i\,\theta\,U_{f_4(2)}\}\right.\cdot\\
    &\cdot\exp\{i\,\theta\,U_{h_{15}(13)}\}\exp\{i\,\theta\,U_{f_{15}(13)}\}\exp\{i\,\theta\,U_{h_{16}(14)}\}\exp\{i\,\theta\,U_{f_{16}(14)}\}\cdot\\
    &\cdot\exp\{i\,\theta\,U_{h_{15}(11)}\}\exp\{i\,\theta\,U_{f_{15}(11)}\}\exp\{i\,\theta\,U_{h_{16}(12)}\}\exp\{i\,\theta\,U_{f_{16}(12)}\}\cdot\\
    &\cdot\left.\exp\{i\,\theta\,U_{h_{13}(11)}\}\exp\{i\,\theta\,U_{f_{13}(11)}\}\exp\{i\,\theta\,U_{h_{14}(12)}\}\exp\{i\,\theta\,U_{f_{14}(12)}\}\right]\prod T_1^o\\
\end{aligned}
\end{equation}}
{\footnotesize
\begin{equation}\label{eqn:M2odd_n4}
\begin{aligned}
    M_2^o&=\prod T_2^o\left[\prod_{
    \begin{tabular}{cc}
    (2,5),\!\!\!\!\!\!&\!\!\!(10,15),\\
    (10,13),\!\!\!\!\!&(10,11) 
    \end{tabular}
    }\exp\{i\,\theta\,U_{h_\beta(\alpha-1)}\}\exp\{i\,\theta\,U_{f_\beta(\alpha-1)}\}\exp\{i\,\theta\,U_{h_{\beta+1}(\alpha)}\exp\{i\,\theta\,U_{f_{\beta+1}(\alpha)}\}\right]\prod T_2^o=\\
    &=\prod T_2^o\left[\exp\{i\,\theta\,U_{h_5(1)}\}\exp\{i\,\theta\,U_{f_5(1)}\}\exp\{i\,\theta\,U_{h_6(2)}\exp\{i\,\theta\,U_{f_6(2)}\}\right.\cdot\\
    &\cdot\exp\{i\,\theta\,U_{h_{15}(9)}\}\exp\{i\,\theta\,U_{f_{15}(9)}\}\exp\{i\,\theta\,U_{h_{16}(10)}\}\exp\{i\,\theta\,U_{f_{16}(10)}\}\cdot\\
    &\cdot\exp\{i\,\theta\,U_{h_{13}(9)}\}\exp\{i\,\theta\,U_{f_{13}(9)}\}\exp\{i\,\theta\,U_{h_{14}(10)}\}\exp\{i\,\theta\,U_{f_{14}(10)}\}\cdot\\
    &\cdot\left.\exp\{i\,\theta\,U_{h_{11}(9)}\}\exp\{i\,\theta\,U_{f_{11}(9)}\}\exp\{i\,\theta\,U_{h_{12}(10)}\}\exp\{i\,\theta\,U_{f_{12}(10)}\}\right]\prod T_2^o
\end{aligned}
\end{equation}}
{\footnotesize
\begin{equation}\label{eqn:M3odd_n4}
\begin{aligned}
    M_3^o&=\prod T_3^o\left[\prod_{
    \begin{tabular}{cc}
    (2,7),\!\!\!\!\!\!&\!\!\!(8,15),\\
    (8,13),\!\!\!\!\!&(8,11) 
    \end{tabular}
    }\exp\{i\,\theta\,U_{h_\beta(\alpha-1)}\}\exp\{i\,\theta\,U_{f_\beta(\alpha-1)}\}\exp\{i\,\theta\,U_{h_{\beta+1}(\alpha)}\exp\{i\,\theta\,U_{f_{\beta+1}(\alpha)}\}\right]\prod T_3^o=\\
    &=\prod T_3^o\left[\exp\{i\,\theta\,U_{h_7(1)}\}\exp\{i\,\theta\,U_{f_7(1)}\}\exp\{i\,\theta\,U_{h_8(2)}\exp\{i\,\theta\,U_{f_8(2)}\}\right.\cdot\\
    &\cdot\exp\{i\,\theta\,U_{h_{15}(7)}\}\exp\{i\,\theta\,U_{f_{15}(7)}\}\exp\{i\,\theta\,U_{h_{16}(8)}\}\exp\{i\,\theta\,U_{f_{16}(8)}\}\cdot\\
    &\cdot\exp\{i\,\theta\,U_{h_{13}(7)}\}\exp\{i\,\theta\,U_{f_{13}(7)}\}\exp\{i\,\theta\,U_{h_{14}(8)}\}\exp\{i\,\theta\,U_{f_{14}(8)}\}\cdot\\
    &\cdot\left.\exp\{i\,\theta\,U_{h_{11}(7)}\}\exp\{i\,\theta\,U_{f_{11}(7)}\}\exp\{i\,\theta\,U_{h_{12}(8)}\}\exp\{i\,\theta\,U_{f_{12}(8)}\}\right]\prod T_3^o
\end{aligned}  
\end{equation}}
Therefore, the $M_x^o$ factors become $2\times2$-block diagonal matrix where each block no longer belongs to $SU(2)$ (but only to $U(2)$), thus composing as a whole a unitary matrix that can be decomposed only in part through a $M_4ZYZ$~\cite{sarkar2023scalable}. The circuit for $M_x^o$ is given by Figure~\ref{fig:Mxo_circuit_n4}. Finally, Figure~\ref{fig:prodTxo_diagram_n4} shows the circuits corresponding to the remaining seven pairs of $\prod T_x^o$ factors (refer to Appendix~\ref{sec:cnot_sequences}).
\begin{figure}[htbp]
\centering
\resizebox{0.8\textwidth}{!}{
    \begin{subfigure}[b]{1.2\textwidth}
        \centering
        $F_{pre-scaling}(\theta_{1-7,x}^*):$
        \begin{quantikz}
            \lstick{0}&\gate{R_z(\theta_{1,x}^*)}&\ctrl{1}&&\ctrl{1}&&&\ctrl{2}&&&&\ctrl{2}&\\
            \lstick{1}&\gate{R_z(\theta_{2,x}^*)}&\targ{}&\gate{R_z(\theta_{3,x}^*)}&\targ{}&\ctrl{1}&&&&\ctrl{1}&&&\\
            \lstick{2}&&&&\gate{R_z(\theta_{4,x}^*)}&\targ{}&\gate{R_z(\theta_{5,x}^*)}&\targ{}&\gate{R_z(\theta_{6,x}^*)}&\targ{}&\gate{R_z(\theta_{7,x}^*)}&\targ{}&\\
            \lstick{3}&&&&&&&&&&&&
        \end{quantikz}
    \end{subfigure}}\vspace{3mm}
    \resizebox{0.8\textwidth}{!}{
    \begin{subfigure}[b]{1.2\textwidth}
        \centering
        $F_{post-scaling}(\theta_{8-14,x}^*):$
        \begin{quantikz}
            \lstick{0}&\ctrl{2}&&&&\ctrl{2}&&&\ctrl{1}&&\ctrl{1}&\gate{R_z(\theta_{14,x}^*)}&\\
            \lstick{1}&&&\ctrl{1}&&&&\ctrl{1}&\targ{}&\gate{R_z(\theta_{12,x}^*)}&\targ{}&\gate{R_z(\theta_{13,x}^*)}&\\
            \lstick{2}&\targ{}&\gate{R_z(\theta_{8,x}^*)}&\targ{}&\gate{R_z(\theta_{9,x}^*)}&\targ{}&\gate{R_z(\theta_{10,x}^*)}&\targ{}&\gate{R_z(\theta_{11,x}^*)}&&&&\\
            \lstick{3}&&&&&&&&&&&&\\
        \end{quantikz}
    \end{subfigure}}\vspace{3mm}
    \resizebox{0.8\textwidth}{!}{
    \begin{subfigure}[b]{1.2\textwidth}
        \centering
        $M_x^o(\theta_{1-14}^*)=$
        \begin{quantikz}
            \lstick{0}&\gate[3]{F_{pre-scaling}(\theta_{1-7,x}^*)}&\gate[4]{M_4ZYZ}&\gate[3]{F_{post-scaling}(\theta_{8-14,x}^*)}&\\
            \lstick{1}&&&&\\
            \lstick{2}&&&&\\
            \lstick{3}&&&&
        \end{quantikz}
    \end{subfigure}}   
\caption{$M_x^o$ quantum circuit.}
\label{fig:Mxo_circuit_n4}
\end{figure}
\begin{figure}[htbp]
\resizebox{1\textwidth}{!}{
    \begin{subfigure}[c]{1\textwidth}
    \centering\vspace{5mm}
        \begin{tabular}{||c|c|c|c|c||}
        \hline
        x&Bits&Binary&Control-Target&$k$-index\\
        \hline\hline
        1&$x_0x_1x_2$&001&(3,2)&2\\
        2&$x_0x_1x_2$&010&(3,1)&1\\
        3&$x_0x_1x_2$&011&(3,1),(3,2)&1\\
        4&$x_0x_1x_2$&100&(3,0)&0\\
        5&$x_0x_1x_2$&101&(3,0),(3,2)&0\\
        6&$x_0x_1x_2$&110&(3,0),(3,1)&0\\
        7&$x_0x_1x_2$&111&(3,0),(3,1),(3,2)&0\\
        \hline
        \end{tabular}   
    \end{subfigure}}\vspace{5mm}
    \resizebox{0.25\textwidth}{!}{
    \begin{subfigure}[c]{0.3\textwidth}
    \centering
        \begin{quantikz}
        \lstick{0}&&&&\\
        \lstick{1}&&&&\\
        \lstick{2}&\ctrl{1}&\targ{}&\ctrl{1}&\\
        \lstick{3}&\targ{}&\ctrl{-1}&\targ{}&
        \end{quantikz}
        \caption{$\prod T_1^o$}
    \end{subfigure}}
    \resizebox{0.25\textwidth}{!}{
    \begin{subfigure}[c]{0.3\textwidth}
    \centering
        \begin{quantikz}
        \lstick{0}&&&&\\
        \lstick{1}&\ctrl{2}&\targ{}&\ctrl{2}&\\
        \lstick{2}&&&&\\
        \lstick{3}&\targ{}&\ctrl{-2}&\targ{}&
        \end{quantikz}
        \caption{$\prod T_2^o$}
    \end{subfigure}}
    \resizebox{0.25\textwidth}{!}{
    \begin{subfigure}[c]{0.3\textwidth}
    \centering
        \begin{quantikz}
        \lstick{0}&&&&&\\
        \lstick{1}&\ctrl{2}&\targ{}&&\ctrl{2}&\\
        \lstick{2}&&&\targ{}&&\\
        \lstick{3}&\targ{}&\ctrl{-2}&\ctrl{-1}&\targ{}&
        \end{quantikz}
        \caption{$\prod T_3^o$}
    \end{subfigure}}\vspace{0.3cm}
    \resizebox{0.25\textwidth}{!}{
    \begin{subfigure}[c]{0.3\textwidth}
    \centering
        \begin{quantikz}
        \lstick{0}&\ctrl{3}&\targ{}&\ctrl{3}&\\
        \lstick{1}&&&&\\
        \lstick{2}&&&&\\
        \lstick{3}&\targ{}&\ctrl{-3}&\targ{}&
        \end{quantikz}
        \caption{$\prod T_4^o$}
    \end{subfigure}}
    \resizebox{0.25\textwidth}{!}{
    \begin{subfigure}[c]{0.3\textwidth}
    \centering
        \begin{quantikz}
        \lstick{0}&\ctrl{3}&\targ{}&&\ctrl{3}&\\
        \lstick{1}&&&&&\\
        \lstick{2}&&&\targ{}&&\\
        \lstick{3}&\targ{}&\ctrl{-3}&\ctrl{-1}&\targ{}&
        \end{quantikz}
        \caption{$\prod T_5^o$}
    \end{subfigure}}
    \resizebox{0.25\textwidth}{!}{
    \begin{subfigure}[c]{0.3\textwidth}
    \centering
        \begin{quantikz}
        \lstick{0}&\ctrl{3}&\targ{}&&\ctrl{3}&\\
        \lstick{1}&&&\targ{}&&\\
        \lstick{2}&&&&&\\
        \lstick{3}&\targ{}&\ctrl{-3}&\ctrl{-2}&\targ{}&
        \end{quantikz}
        \caption{$\prod T_6^o$}
    \end{subfigure}}\vspace{0.3cm}
    \resizebox{0.25\textwidth}{!}{
    \begin{subfigure}[c]{0.3\textwidth}
    \centering
        \begin{quantikz}
        \lstick{0}&\ctrl{3}&\targ{}&&&\ctrl{3}&\\
        \lstick{1}&&&\targ{}&&&\\
        \lstick{2}&&&&\targ{}&&\\
        \lstick{3}&\targ{}&\ctrl{-3}&\ctrl{-2}&\ctrl{-1}&\targ{}&
        \end{quantikz}
        \caption{$\prod T_7^o$}
    \end{subfigure}} 
\caption{Diagram to find the control-target pair and the $k$-index for the CNOTs representing $\prod T_x^o$.}
\label{fig:prodTxo_diagram_n4}
\end{figure}

\subsection{4-qubit circuit to approximate SU(16)}
Composing the circuits of the previous subsections in the correct order, we can draw the overall circuit to approximate an arbitrary $SU(16)$ operator according to equation~(\ref{eqn:U_approx}) with one single layer ($l=1$). The approximating operator can be rewritten in the following way:
\begin{equation}
\begin{split}
    &\mathcal{U}_{approx}(\Theta)=Z_1(\Theta_{Z})\,\Psi_1(\Theta_\Psi)\,\Phi_1(\Theta_\Phi)=\\
    &=\prod_{i\in\mathcal{J}_Z}\exp\{i\;\theta_iU_i\}\prod_{j\in\mathcal{A}_\Psi}\exp\{i\;\theta_jU_j\}\cdot\prod_{x=1}^7\left[\prod_{k_e\in\mathcal{B}_{\Psi,x}}\exp\{i\;\theta_{k_e}U_{k_e}\}\right]\prod_{x=1}^7\left[\prod_{k_o\in\mathcal{C}_{\Phi,x}}\exp\{i\;\theta_{k_o}U_{k_o}\}\right]
\end{split}
\end{equation}
Finally, keeping in mind the reversed order of the factors compared to the mathematical writing, in Figure~\ref{fig:n4circuit} is illustrated the overall circuit from left to right, line by line, taking care to label its sub-components. Figure~\ref{fig:n4circuit} highlights which parts of the quantum circuit repeat identically for this particular case $n=4$ and it is possible to immediately notice which CNOTs are simplified (red gates).
\begin{figure}[htbp]
\resizebox{1\textwidth}{!}{
    \begin{subfigure}[b]{1.6\textwidth}
    \centering
        \begin{quantikz}
        \lstick{0}&\ctrl{3}\gategroup[4,steps=5,style={dashed,rounded corners},label style={label position=below,anchor=north,yshift=-0.3cm}]{$\prod T_7^o$}&\targ{}&&&\ctrl{3}&\gate[4]{M_7^o}&\ctrl{3}\gategroup[4,steps=5,style={dashed,rounded corners},label style={label position=below,anchor=north,yshift=-0.3cm}]{$\prod T_7^o$}&\targ[style={red}]{}&&&\ctrl[style={red}]{3}&\ctrl[style={red}]{3}\gategroup[4,steps=4,style={dashed,rounded corners},label style={label position=below,anchor=north,yshift=-0.3cm}]{$\prod T_6^o$}&\targ[style={red}]{}&&\ctrl{3}&\gate[4]{M_6^o}&\ctrl{3}\gategroup[4,steps=4,style={dashed,rounded corners},label style={label position=below,anchor=north,yshift=-0.3cm}]{$\prod T_6^o$}&\targ[style={red}]{}&&\ctrl[style={red}]{3}&\ctrl[style={red}]{3}\gategroup[4,steps=4,style={dashed,rounded corners},label style={label position=below,anchor=north,yshift=-0.3cm}]{$\prod T_5^o$}&\targ[style={red}]{}&&\ctrl{3}&\gate[4]{M_5^o}&\rstick{...}\\
        \lstick{1}&&&\targ{}&&&&&&\targ[style={red}]{}&&&&&\targ[style={red}]{}&&&&&\targ{}&&&&&&&\rstick{...}\\
        \lstick{2}&&&&\targ{}&&&&&&\targ{}&&&&&&&&&&&&&\targ{}&&&\rstick{...}\\
        \lstick{3}&\targ{}&\ctrl{-3}&\ctrl{-2}&\ctrl{-1}&\targ{}&&\targ{}&\ctrl[style={red}]{-3}&\ctrl[style={red}]{-2}&\ctrl{-1}&\targ[style={red}]{}&\targ[style={red}]{}&\ctrl[style={red}]{-3}&\ctrl[style={red}]{-2}&\targ{}&&\targ{}&\ctrl[style={red}]{-3}&\ctrl{-2}&\targ[style={red}]{}&\targ[style={red}]{}&\ctrl[style={red}]{-3}&\ctrl{-1}&\targ{}&&\rstick{...}
        \end{quantikz}        
    \end{subfigure}}
\resizebox{1\textwidth}{!}{
    \begin{subfigure}[b]{1.6\textwidth}
    \centering
        \begin{quantikz}
        \lstick{...}&\ctrl{3}\gategroup[4,steps=4,style={dashed,rounded corners},label style={label position=below,anchor=north,yshift=-0.3cm}]{$\prod T_5^o$}&\targ[style={red}]{}&&\ctrl[style={red}]{3}&\ctrl[style={red}]{3}\gategroup[4,steps=3,style={dashed,rounded corners},label style={label position=below,anchor=north,yshift=-0.3cm}]{$\prod T_4^o$}&\targ[style={red}]{}&\ctrl{3}&\gate[4]{M_4^o}&\ctrl{3}\gategroup[4,steps=3,style={dashed,rounded corners},label style={label position=below,anchor=north,yshift=-0.3cm}]{$\prod T_4^o$}&\targ{}&\ctrl{3}&\gategroup[4,steps=4,style={dashed,rounded corners},label style={label position=below,anchor=north,yshift=-0.3cm}]{$\prod T_3^o$}&&&&\gate[4]{M_3^o}&\gategroup[4,steps=4,style={dashed,rounded corners},label style={label position=below,anchor=north,yshift=-0.3cm}]{$\prod T_3^o$}&&&&\gategroup[4,steps=3,style={dashed,rounded corners},label style={label position=below,anchor=north,yshift=-0.3cm}]{$\prod T_2^o$}&&&\gate[4]{M_2^o}&\rstick{...}\\
        \lstick{...}&&&&&&&&&&&&\ctrl{2}&\targ{}&&\ctrl{2}&&\ctrl{2}&\targ[style={red}]{}&&\ctrl[style={red}]{2}&\ctrl[style={red}]{2}&\targ[style={red}]{}&\ctrl{2}&&\rstick{...}\\
        \lstick{...}&&&\targ{}&&&&&&&&&&&\targ{}&&&&&\targ{}&&&&&&\rstick{...}\\
        \lstick{...}&\targ{}&\ctrl[style={red}]{-3}&\ctrl{-1}&\targ[style={red}]{}&\targ[style={red}]{}&\ctrl[style={red}]{-3}&\targ{}&&\targ{}&\ctrl{-3}&\targ{}&\targ{}&\ctrl{-2}&\ctrl{-1}&\targ{}&&\targ{}&\ctrl[style={red}]{-2}&\ctrl{-1}&\targ[style={red}]{}&\targ[style={red}]{}&\ctrl[style={red}]{-2}&\targ{}&&\rstick{...}
        \end{quantikz}        
    \end{subfigure}}
\resizebox{1\textwidth}{!}{
    \begin{subfigure}[b]{1.6\textwidth}
    \centering
        \begin{quantikz}
        \lstick{...}&\gategroup[4,steps=3,style={dashed,rounded corners},label style={label position=below,anchor=north,yshift=-0.3cm}]{$\prod T_2^o$}&&&\gategroup[4,steps=3,style={dashed,rounded corners},label style={label position=below,anchor=north,yshift=-0.3cm}]{$\prod T_1^o$}&&&\gate[4]{M_1^o}&\gategroup[4,steps=3,style={dashed,rounded corners},label style={label position=below,anchor=north,yshift=-0.3cm}]{$\prod T_1^o$}&&&\targ{}\gategroup[4,steps=3,style={dashed,rounded corners},label style={label position=below,anchor=north,yshift=-0.3cm}]{$\prod T_7^e$}&&&\gate[4]{M_7^e}&\targ[style={red}]{}\gategroup[4,steps=3,style={dashed,rounded corners},label style={label position=below,anchor=north,yshift=-0.3cm}]{$\prod T_7^e$}&&&\targ[style={red}]{}\gategroup[4,steps=2,style={dashed,rounded corners},label style={label position=below,anchor=north,yshift=-0.3cm}]{$\prod T_6^e$}&&\gate[4]{M_6^e}&\rstick{...}\\
        \lstick{...}&\ctrl{2}&\targ{}&\ctrl{2}&&&&&&&&&\targ{}&&&&\targ[style={red}]{}&&&\targ[style={red}]{}&&\rstick{...}\\
        \lstick{...}&&&&\ctrl{1}&\targ{}&\ctrl{1}&&\ctrl{1}&\targ{}&\ctrl{1}&&&\targ{}&&&&\targ{}&&&&\rstick{...}\\
        \lstick{...}&\targ{}&\ctrl{-2}&\targ{}&\targ{}&\ctrl{-1}&\targ{}&&\targ{}&\ctrl{-1}&\targ{}&\ctrl{-3}&\ctrl{-2}&\ctrl{-1}&&\ctrl[style={red}]{-3}&\ctrl[style={red}]{-2}&\ctrl{-1}&\ctrl[style={red}]{-3}&\ctrl[style={red}]{-2}&&\rstick{...}
        \end{quantikz}        
    \end{subfigure}}
    \resizebox{1\textwidth}{!}{
    \begin{subfigure}[b]{1.6\textwidth}
    \centering
        \begin{quantikz}
        \lstick{...}&\targ[style={red}]{}\gategroup[4,steps=2,style={dashed,rounded corners},label style={label position=below,anchor=north,yshift=-0.3cm}]{$\prod T_6^e$}&&\targ[style={red}]{}\gategroup[4,steps=2,style={dashed,rounded corners},label style={label position=below,anchor=north,yshift=-0.3cm}]{$\prod T_5^e$}&&\gate[4]{M_5^e}&\targ[style={red}]{}\gategroup[4,steps=2,style={dashed,rounded corners},label style={label position=below,anchor=north,yshift=-0.3cm}]{$\prod T_5^e$}&&\targ[style={red}]{}\gategroup[4,steps=1,style={dashed,rounded corners},label style={label position=below,anchor=north,yshift=-0.3cm}]{$\prod T_4^e$}&\gate[4]{M_4^e}&\targ{}\gategroup[4,steps=1,style={dashed,rounded corners},label style={label position=below,anchor=north,yshift=-0.3cm}]{$\prod T_4^e$}&\gategroup[4,steps=2,style={dashed,rounded corners},label style={label position=below,anchor=north,yshift=-0.3cm}]{$\prod T_3^e$}&&\gate[4]{M_3^e}&\gategroup[4,steps=2,style={dashed,rounded corners},label style={label position=below,anchor=north,yshift=-0.3cm}]{$\prod T_3^e$}&&\gategroup[4,steps=1,style={dashed,rounded corners},label style={label position=below,anchor=north,yshift=-0.3cm}]{$\prod T_2^e$}&\gate[4]{M_2^e}&\rstick{...}\\
        \lstick{...}&&\targ{}&&&&&&&&&\targ{}&&&\targ[style={red}]{}&&\targ[style={red}]{}&&\rstick{...}\\
        \lstick{...}&&&&\targ{}&&&\targ{}&&&&&\targ{}&&&\targ{}&&&\rstick{...}\\
        \lstick{...}&\ctrl[style={red}]{-3}&\ctrl{-2}&\ctrl[style={red}]{-3}&\ctrl{-1}&&\ctrl[style={red}]{-3}&\ctrl{-1}&\ctrl[style={red}]{-3}&&\ctrl{-3}&\ctrl{-2}&\ctrl{-1}&&\ctrl[style={red}]{-2}&\ctrl{-1}&\ctrl[style={red}]{-2}&&\rstick{...}
        \end{quantikz}        
    \end{subfigure}}
    \resizebox{1\textwidth}{!}{
    \begin{subfigure}[b]{1.6\textwidth}
    \centering
        \begin{quantikz}
        \lstick{...}&\gategroup[4,steps=1,style={dashed,rounded corners},label style={label position=below,anchor=north,yshift=-0.3cm}]{$\prod T_2^e$}&\gategroup[4,steps=1,style={dashed,rounded corners},label style={label position=below,anchor=north,yshift=-0.3cm}]{$\prod T_1^e$}&\gate[4]{M_1^e}&\gategroup[4,steps=1,style={dashed,rounded corners},label style={label position=below,anchor=north,yshift=-0.3cm}]{$\prod T_1^e$}&\gate[4]{M_4ZYZ}\gategroup[4,steps=1,style={dashed,rounded corners},label style={label position=below,anchor=north,yshift=-0.3cm}]{sub-factor A}&\gate[4]{Z(\Theta_Z)}&\\
        \lstick{...}&\targ{}&&&&&&\\
        \lstick{...}&&\targ{}&&\targ{}&&&\\
        \lstick{...}&\ctrl{-2}&\ctrl{-1}&&\ctrl{-1}&&&
        \end{quantikz}        
    \end{subfigure}}
\caption{Quantum circuit for general 4-qubit special unitary operators; red gates are simplified gates.}
\label{fig:n4circuit}
\end{figure}

\clearpage
\renewcommand{\thesubfigure}{\arabic{subfigure}}
\section{Tested circuits}\label{app:circuits}
\label{circuitsTested}
\subsection{List of tested circuits for n = 2}~
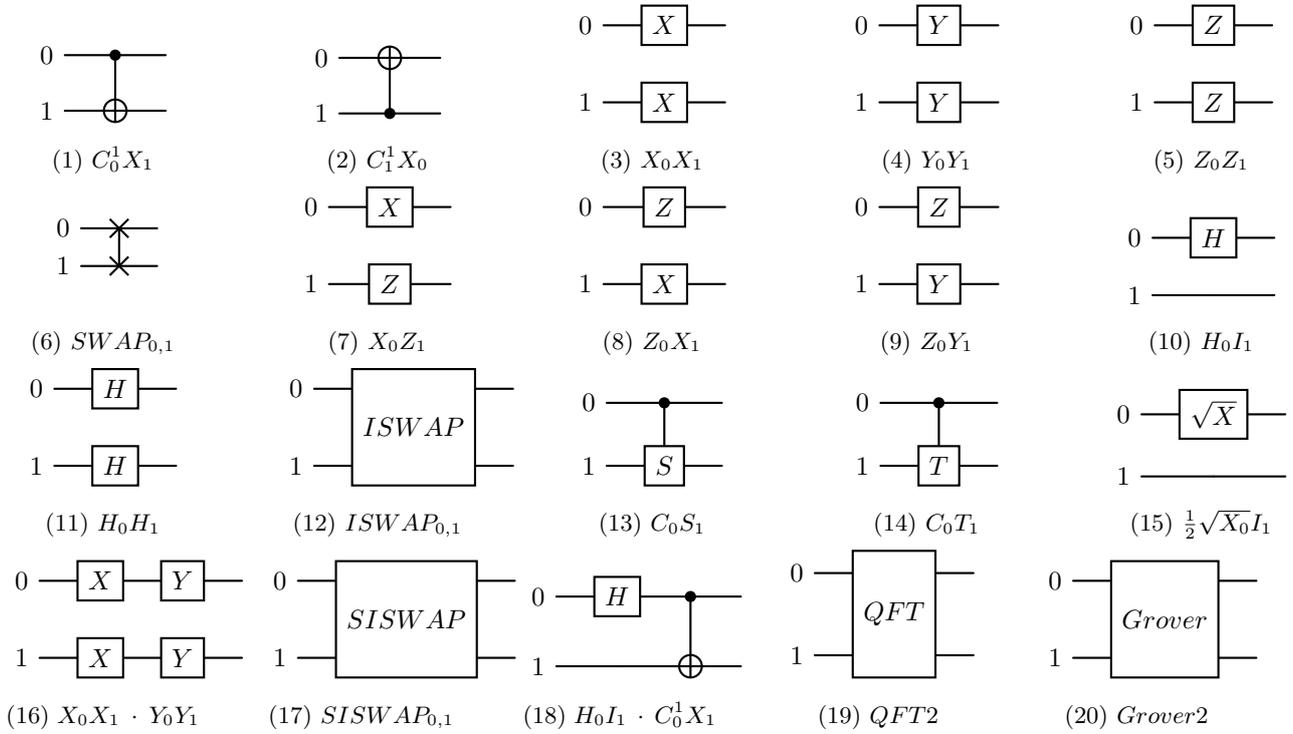
\begin{figure}[H]
    \begin{subfigure}[b]{0.2\textwidth}
        \centering
        \begin{quantikz}
        \lstick{0}&\ctrl{1}&\\
        \lstick{1}&\targ{}&
        \end{quantikz}
    \end{subfigure}
    \hfill
    \begin{subfigure}[b]{0.2\textwidth}
        \centering
        \begin{quantikz}
        \lstick{0}&\targ{}&\\
        \lstick{1}&\ctrl{-1}&
        \end{quantikz}
    \end{subfigure}
    \hfill
    \begin{subfigure}[b]{0.2\textwidth}
        \centering
        \begin{quantikz}
        \lstick{0}&\gate{X}&\\
        \lstick{1}&\gate{X}& 
        \end{quantikz}
    \end{subfigure}
    \hfill
    \begin{subfigure}[b]{0.2\textwidth}
        \centering
        \begin{quantikz}
        \lstick{0}&\gate{Y}&\\
        \lstick{1}&\gate{Y}&  
        \end{quantikz}
    \end{subfigure}\vspace{0.5cm}

    \begin{subfigure}[b]{0.2\textwidth}
        \centering
        \begin{quantikz}
        \lstick{0}&\gate{Z}&\\
        \lstick{1}&\gate{Z}&  
        \end{quantikz}
    \end{subfigure}
    \hfill
    \begin{subfigure}[b]{0.3\textwidth}
        \centering
        \begin{quantikz}
        \lstick{0}&\gate[2]{\sqrt{\text{iSWAP}}}&\\
        \lstick{1}&&  
        \end{quantikz}
    \end{subfigure}
    \hfill
    \begin{subfigure}[b]{0.2\textwidth}
        \centering
        \begin{quantikz}
        \lstick{0}&\gate{X}&\\
        \lstick{1}&\gate{Z}&  
        \end{quantikz}
    \end{subfigure}
    \hfill
    \begin{subfigure}[b]{0.2\textwidth}
        \centering
        \begin{quantikz}
        \lstick{0}&\gate{Z}&\\
        \lstick{1}&\gate{X}&  
        \end{quantikz}
    \end{subfigure}\vspace{0.5cm}

    \begin{subfigure}[b]{0.2\textwidth}
        \centering
        \begin{quantikz}
        \lstick{0}&\gate{Z}&\\
        \lstick{1}&\gate{Y}&  
        \end{quantikz}
    \end{subfigure}
    \hfill
    \begin{subfigure}[b]{0.2\textwidth}
        \centering
        \begin{quantikz}
        \lstick{0}&\gate{H}&\\
        \lstick{1}&&  
        \end{quantikz}
    \end{subfigure}
    \hfill
    \begin{subfigure}[b]{0.2\textwidth}
        \centering
        \begin{quantikz}
        \lstick{0}&\gate{H}&\\
        \lstick{1}&\gate{H}&  
        \end{quantikz}
    \end{subfigure}
    \hfill
    \begin{subfigure}[b]{0.3\textwidth}
        \centering
        \begin{quantikz}
        \lstick{0}&\gate[2]{\text{iSWAP}}&\\
        \lstick{1}&&  
        \end{quantikz}
    \end{subfigure}\vspace{0.5cm}

    \begin{subfigure}[b]{0.2\textwidth}
        \centering
        \begin{quantikz}
        \lstick{0}&\ctrl{1}&\\
        \lstick{1}&\gate{S}&  
        \end{quantikz}
    \end{subfigure}
    \hfill
    \begin{subfigure}[b]{0.2\textwidth}
        \centering
        \begin{quantikz}
        \lstick{0}&\ctrl{1}&\\
        \lstick{1}&\gate{T}&  
        \end{quantikz}
    \end{subfigure}
    \hfill
    \begin{subfigure}[b]{0.2\textwidth}
        \centering
        \begin{quantikz}
        \lstick{0}&\gate{\sqrt{X}}&\\
        \lstick{1}&&  
        \end{quantikz}
    \end{subfigure}
    \hfill
    \begin{subfigure}[b]{0.3\textwidth}
        \centering
        \begin{quantikz}
        \lstick{0}&\gate{X}&\gate{Y}&\\
        \lstick{1}&\gate{X}&\gate{Y}&  
        \end{quantikz}
    \end{subfigure}\vspace{0.5cm}

    \begin{subfigure}[b]{0.15\textwidth}
        \centering
        \begin{quantikz}
        \lstick{0}&\swap{1}&\\
        \lstick{1}&\targX{}&\\
        \end{quantikz}
    \end{subfigure}
    \hfill
    \begin{subfigure}[b]{0.3\textwidth}
        \centering
        \begin{quantikz}
        \lstick{0}&\gate{H}&\ctrl{1}&\\
        \lstick{1}&&\targ{}&
        \end{quantikz}
    \end{subfigure}
    \hfill
    \begin{subfigure}[b]{0.25\textwidth}
        \centering
        \begin{quantikz}
        \lstick{0}&\gate[2]{\text{QFT}}&\\
        \lstick{1}&&  
        \end{quantikz}
    \end{subfigure}
    \hfill
    \begin{subfigure}[b]{0.25\textwidth}
        \centering
        \begin{quantikz}
        \lstick{0}&\gate[2]{\text{Grover}}&\\
        \lstick{1}&&  
        \end{quantikz}
    \end{subfigure}
\label{fig:N_2Circuits}
\end{figure}
\newpage
\subsection{List of tested circuits for n = 3}~
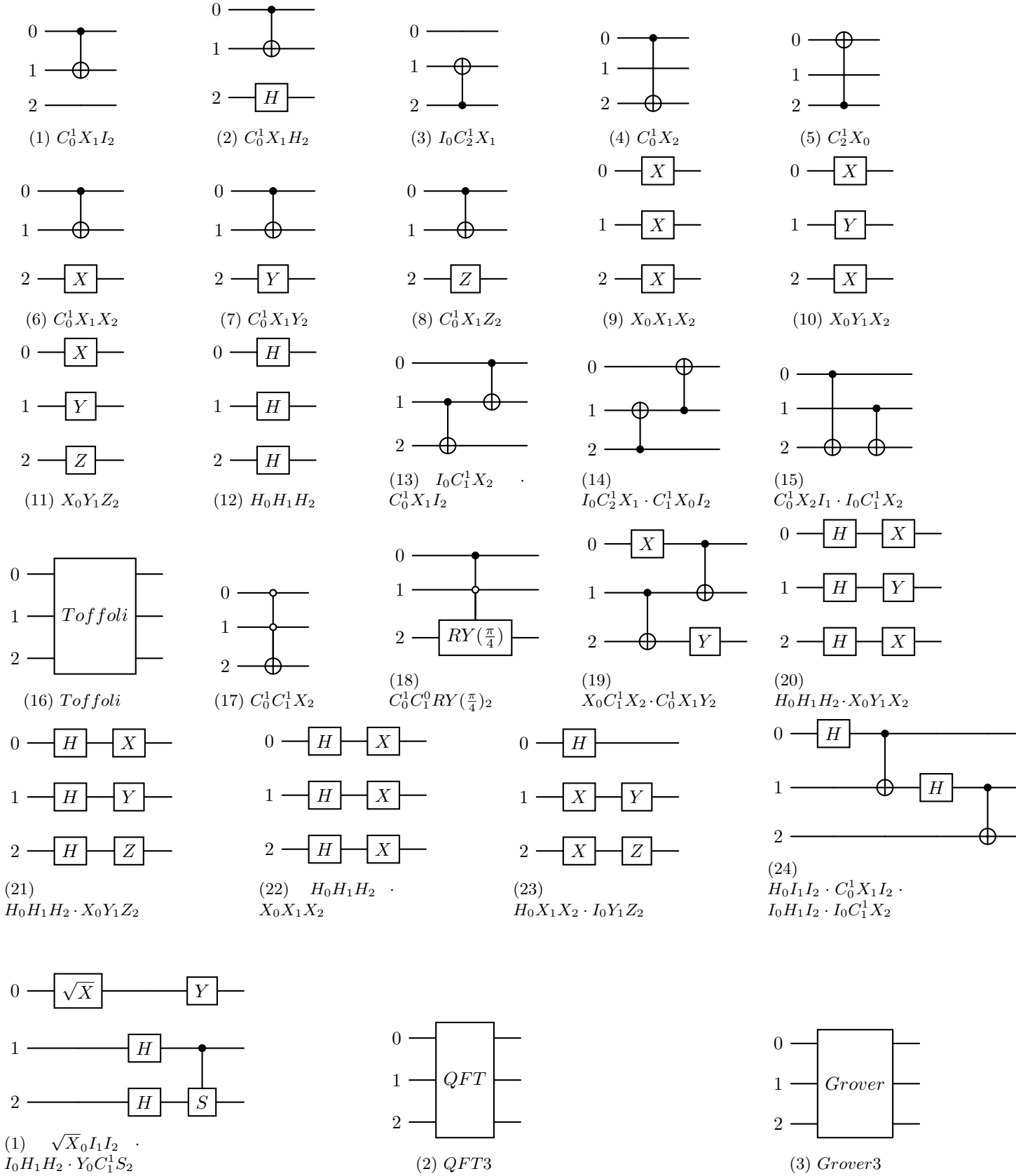
\begin{figure}[H]
    \begin{subfigure}[b]{0.15\textwidth}
        \centering
        \begin{quantikz}
        \lstick{0}&\ctrl{1}&\\
        \lstick{1}&\targ{}&\\
        \lstick{2}&&
        \end{quantikz}
    \end{subfigure}
    \hfill
    \begin{subfigure}[b]{0.2\textwidth}
        \centering
        \begin{quantikz}
        \lstick{0}&\ctrl{1}&\\
        \lstick{1}&\targ{}&\\
        \lstick{2}&\gate{H}&
        \end{quantikz}
    \end{subfigure}
    \hfill
    \begin{subfigure}[b]{0.15\textwidth}
        \centering
        \begin{quantikz}
        \lstick{0}&&\\
        \lstick{1}&\targ{}&\\
        \lstick{2}&\ctrl{-1}&
        \end{quantikz}
    \end{subfigure}
    \hfill
    \begin{subfigure}[b]{0.15\textwidth}
        \centering
        \begin{quantikz}
        \lstick{0}&\ctrl{2}&\\
        \lstick{1}&&\\
        \lstick{2}&\targ{}&
        \end{quantikz}
    \end{subfigure}\vspace{0.3cm}

    \begin{subfigure}[b]{0.2\textwidth}
        \centering
        \begin{quantikz}
        \lstick{0}&\ctrl{1}&\\
        \lstick{1}&\targ{}&\\
        \lstick{2}&\gate{X}&
        \end{quantikz}
    \end{subfigure}
    \hfill
    \begin{subfigure}[b]{0.2\textwidth}
        \centering
        \begin{quantikz}
        \lstick{0}&\ctrl{1}&\\
        \lstick{1}&\targ{}&\\
        \lstick{2}&\gate{Y}&
        \end{quantikz}
    \end{subfigure}
    \hfill
    \begin{subfigure}[b]{0.2\textwidth}
        \centering
        \begin{quantikz}
        \lstick{0}&\ctrl{1}&\\
        \lstick{1}&\targ{}&\\
        \lstick{2}&\gate{Z}&
        \end{quantikz}
    \end{subfigure}
    \hfill
    \begin{subfigure}[b]{0.2\textwidth}
        \centering
        \begin{quantikz}
        \lstick{0}&\gate{X}&\\
        \lstick{1}&\gate{X}&\\
        \lstick{2}&\gate{X}&
        \end{quantikz}
    \end{subfigure}\vspace{0.3cm}

    \begin{subfigure}[b]{0.2\textwidth}
        \centering
        \begin{quantikz}
        \lstick{0}&\gate{X}&\\
        \lstick{1}&\gate{Y}&\\
        \lstick{2}&\gate{X}&
        \end{quantikz}
    \end{subfigure}
    \hfill
    \begin{subfigure}[b]{0.2\textwidth}
        \centering
        \begin{quantikz}
        \lstick{0}&\gate{X}&\\
        \lstick{1}&\gate{Y}&\\
        \lstick{2}&\gate{Z}&
        \end{quantikz}
    \end{subfigure}
    \hfill
    \begin{subfigure}[b]{0.2\textwidth}
        \centering
        \begin{quantikz}
        \lstick{0}&\gate{H}&\\
        \lstick{1}&\gate{H}&\\
        \lstick{2}&\gate{H}&
        \end{quantikz}
    \end{subfigure}
    \hfill
    \begin{subfigure}[b]{0.3\textwidth}
        \centering
        \begin{quantikz}
        \lstick{0}&&\ctrl{1}&\\
        \lstick{1}&\ctrl{1}&\targ{}&\\
        \lstick{2}&\targ{}&&
        \end{quantikz}
    \end{subfigure}\vspace{0.4cm}

    \begin{subfigure}[b]{0.28\textwidth}
        \centering
        \begin{quantikz}
        \lstick{0}&&\targ{}&\\
        \lstick{1}&\targ{}&\ctrl{-1}&\\
        \lstick{2}&\ctrl{-1}&&
        \end{quantikz}
    \end{subfigure}
    \hfill
    \begin{subfigure}[b]{0.28\textwidth}
        \centering
        \begin{quantikz}
        \lstick{0}&\ctrl{2}&&\\
        \lstick{1}&&\ctrl{1}&\\
        \lstick{2}&\targ{}&\targ{}&
        \end{quantikz}
    \end{subfigure}
    \hfill
    \resizebox{0.2\textwidth}{!}{
    \begin{subfigure}[b]{0.25\textwidth}
        \centering
        \begin{quantikz}
        \lstick{0}&\gate[3]{\text{Toffoli}}&\\
        \lstick{1}&&\\
        \lstick{2}&&
        \end{quantikz}
    \end{subfigure}}
    \hfill
    \resizebox{0.2\textwidth}{!}{
    \begin{subfigure}[b]{0.25\textwidth}
        \centering
        \begin{quantikz}
        \lstick{0}&\gate[3]{\text{Grover}}&\\ 
        \lstick{1}&&\\
        \lstick{2}&&
        \end{quantikz}
    \end{subfigure}}\vspace{0.4cm}

    \begin{subfigure}[b]{0.15\textwidth}
        \centering
        \begin{quantikz}
        \lstick{0}&\targ{}&\\
        \lstick{1}&&\\
        \lstick{2}&\ctrl{-2}&
        \end{quantikz}
    \end{subfigure}
    \hfill
    \begin{subfigure}[b]{0.25\textwidth}
        \centering
        \begin{quantikz}
        \lstick{0}&\ctrl{2}&\\
        \lstick{1}&\ctrl[open]{1}&\\
        \lstick{2}&\gate{RY(\frac{\pi}{4})}&
        \end{quantikz}
    \end{subfigure}
    \hfill
    \begin{subfigure}[b]{0.25\textwidth}
        \centering
        \begin{quantikz}
        \lstick{0}&\gate{X}&\ctrl{1}&\\
        \lstick{1}&\ctrl{1}&\targ{}&\\
        \lstick{2}&\targ{}&\gate{Y}&
        \end{quantikz}
    \end{subfigure}
    \hfill
    \begin{subfigure}[b]{0.25\textwidth}
        \centering
        \begin{quantikz}
        \lstick{0}&\gate{H}&\gate{X}&\\
        \lstick{1}&\gate{H}&\gate{Y}&\\
        \lstick{2}&\gate{H}&\gate{X}&
        \end{quantikz}
    \end{subfigure}\vspace{0.4cm}

    \begin{subfigure}[b]{0.25\textwidth}
        \centering
        \begin{quantikz}
        \lstick{0}&\gate{H}&\gate{X}&\\
        \lstick{1}&\gate{H}&\gate{Y}&\\
        \lstick{2}&\gate{H}&\gate{Z}&
        \end{quantikz}
    \end{subfigure}
    \hfill
    \begin{subfigure}[b]{0.25\textwidth}
        \centering
        \begin{quantikz}
        \lstick{0}&\gate{H}&\gate{X}&\\
        \lstick{1}&\gate{H}&\gate{X}&\\
        \lstick{2}&\gate{H}&\gate{X}&
        \end{quantikz}
    \end{subfigure}
    \hfill
    \begin{subfigure}[b]{0.3\textwidth}
        \centering
        \begin{quantikz}
        \lstick{0}&\gate{H}&&\\
        \lstick{1}&\gate{X}&\gate{Y}&\\
        \lstick{2}&\gate{X}&\gate{Z}&
        \end{quantikz}
    \end{subfigure}
    \hfill
    \begin{subfigure}[b]{0.15\textwidth}
        \centering
        \begin{quantikz}
        \lstick{0}&\ctrl[open]{2}&\\
        \lstick{1}&\ctrl[open]{1}&\\
        \lstick{2}&\targ{}&
        \end{quantikz}
    \end{subfigure}\vspace{0.4cm}
    
    \resizebox{0.38\textwidth}{!}{
    \begin{subfigure}[b]{0.5\textwidth}
        \centering
        \begin{quantikz}
        \lstick{0}&\gate{\sqrt{X}}&&\gate{Y}&\\
        \lstick{1}&&\gate{H}&\ctrl{1}&\\
        \lstick{2}&&\gate{H}&\gate{S}&
        \end{quantikz}
    \end{subfigure}}
    \hfill
    \resizebox{0.38\textwidth}{!}{
    \begin{subfigure}[b]{0.5\textwidth}
        \centering
        \begin{quantikz}
        \lstick{0}&\gate{H}&\ctrl{1}&&&\\
        \lstick{1}&&\targ{}&\gate{H}&\ctrl{1}&\\
        \lstick{2}&&&&\targ{}&
        \end{quantikz}
    \end{subfigure}}
    \hfill
    \resizebox{0.2\textwidth}{!}{
    \begin{subfigure}[b]{0.25\textwidth}
        \centering
        \begin{quantikz}
        \lstick{0}&\gate[3]{\text{QFT}}&\\
        \lstick{1}&&\\
        \lstick{2}&&
        \end{quantikz}
    \end{subfigure}}
\label{fig:N_3Circuits}
\end{figure}
\newpage
\subsection{List of tested circuits for n = 4}~
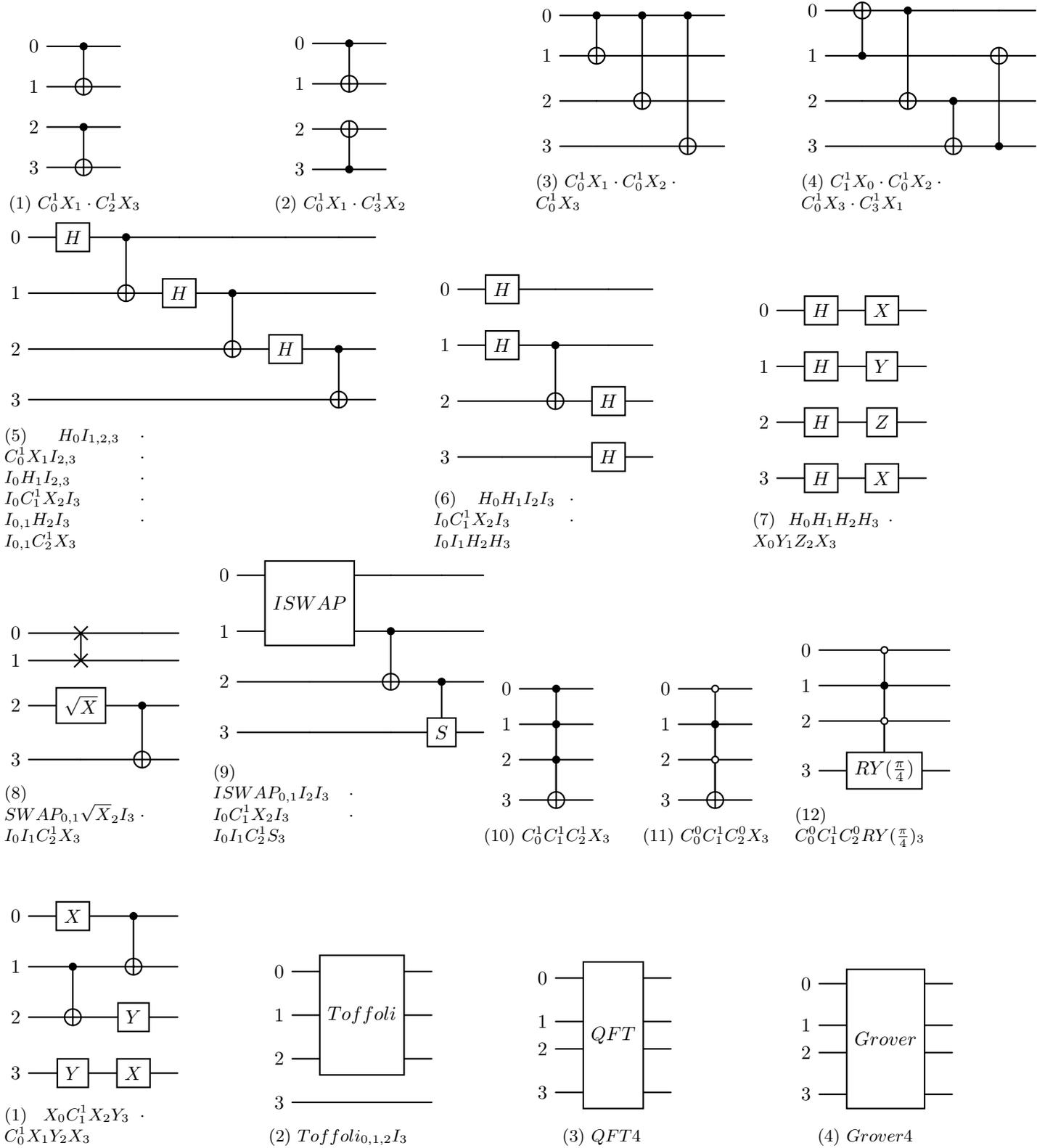
\begin{figure}[H]
    \begin{subfigure}[b]{0.15\textwidth}
        \centering
        \begin{quantikz}
        \lstick{0}&\ctrl{1}&\\
        \lstick{1}&\targ{}&\\
        \lstick{2}&\ctrl{1}&\\
        \lstick{3}&\targ{}&
        \end{quantikz}
    \end{subfigure}
    \hfill
    \begin{subfigure}[b]{0.15\textwidth}
        \centering
        \begin{quantikz}
        \lstick{0}&\ctrl{1}&\\
        \lstick{1}&\targ{}&\\
        \lstick{2}&\targ{}&\\
        \lstick{3}&\ctrl{-1}&
        \end{quantikz}
    \end{subfigure}
    \hfill
    \begin{subfigure}[b]{0.29\textwidth}
        \centering
        \begin{quantikz}
        \lstick{0}&\ctrl{1}&\ctrl{2}&\ctrl{3}&\\
        \lstick{1}&\targ{}&&&\\
        \lstick{2}&&\targ{}&&\\
        \lstick{3}&&&\targ{}&
        \end{quantikz}
    \end{subfigure}
    \hfill
    \begin{subfigure}[b]{0.38\textwidth}
        \centering
        \begin{quantikz}
        \lstick{0}&\targ{}&\ctrl{2}&&&\\
        \lstick{1}&\ctrl{-1}&&&\targ{}&\\
        \lstick{2}&&\targ{}&\ctrl{1}&&\\
        \lstick{3}&&&\targ{}&\ctrl{-2}&
        \end{quantikz}
    \end{subfigure}\vspace{0.5cm}

    \resizebox{0.3\textwidth}{!}{
    \begin{subfigure}[b]{0.35\textwidth}
        \centering
        \begin{quantikz}
        \lstick{0}&\gate{H}&&&\\
        \lstick{1}&\gate{H}&\ctrl{1}&&\\
        \lstick{2}&&\targ{}&\gate{H}&\\
        \lstick{3}&&&\gate{H}&
        \end{quantikz}
    \end{subfigure}}
    \hfill
    \resizebox{0.2\textwidth}{!}{
    \begin{subfigure}[b]{0.25\textwidth}
        \centering
        \begin{quantikz}
        \lstick{0}&\gate{H}&\gate{X}&\\
        \lstick{1}&\gate{H}&\gate{Y}&\\
        \lstick{2}&\gate{H}&\gate{Z}&\\
        \lstick{3}&\gate{H}&\gate{X}&
        \end{quantikz}
    \end{subfigure}}
    \hfill
    \resizebox{0.25\textwidth}{!}{
    \begin{subfigure}[b]{0.3\textwidth}
        \centering
        \begin{quantikz}
        \lstick{0}&\swap{1}&&\\
        \lstick{1}&\targX{}&&\\
        \lstick{2}&\gate{\sqrt{X}}&\ctrl{1}&\\
        \lstick{3}&&\targ{}&
        \end{quantikz}
    \end{subfigure}}
    \hfill
    \resizebox{0.18\textwidth}{!}{
    \begin{subfigure}[b]{0.25\textwidth}
        \centering
        \begin{quantikz}
        \lstick{0}&\gate[4]{Grover}&\\
        \lstick{1}&&\\
        \lstick{2}&&\\
        \lstick{3}&&
        \end{quantikz}
    \end{subfigure}}\vspace{0.5cm}

    \resizebox{0.3\textwidth}{!}{
    \begin{subfigure}[b]{0.55\textwidth}
        \centering
        \begin{quantikz}
        \lstick{0}&\gate{H}&\ctrl{1}&&&&&\\
        \lstick{1}&&\targ{}&\gate{H}&\ctrl{1}&&&\\
        \lstick{2}&&&&\targ{}&\gate{H}&\ctrl{1}&\\
        \lstick{3}&&&&&&\targ{}&
        \end{quantikz}
    \end{subfigure}}
    \hfill
    \begin{subfigure}[b]{0.15\textwidth}
        \centering
        \begin{quantikz}
        \lstick{0}&\ctrl{3}&\\
        \lstick{1}&\ctrl{2}&\\
        \lstick{2}&\ctrl{1}&\\
        \lstick{3}&\targ{}&
        \end{quantikz}
    \end{subfigure}
    \hfill
    \begin{subfigure}[b]{0.15\textwidth}
        \centering
        \begin{quantikz}
        \lstick{0}&\ctrl[open]{3}&\\
        \lstick{1}&\ctrl{2}&\\
        \lstick{2}&\ctrl[open]{1}&\\
        \lstick{3}&\targ{}&
        \end{quantikz}
    \end{subfigure}
    \hfill
    \begin{subfigure}[b]{0.25\textwidth}
        \centering
        \begin{quantikz}
        \lstick{0}&\ctrl[open]{3}&\\
        \lstick{1}&\ctrl{2}&\\
        \lstick{2}&\ctrl[open]{1}&\\
        \lstick{3}&\gate{RY(\frac{\pi}{4})}&
        \end{quantikz}
    \end{subfigure}\vspace{0.5cm}

     \resizebox{0.3\textwidth}{!}{
    \begin{subfigure}[b]{0.4\textwidth}
        \centering
        \begin{quantikz}
        \lstick{0}&\gate[2]{\text{iSWAP}}&&&\\
        \lstick{1}&&\ctrl{1}&&\\
        \lstick{2}&&\targ{}&\ctrl{1}&\\
        \lstick{3}&&&\gate{S}&
        \end{quantikz}
    \end{subfigure}}
    \hfill
    \resizebox{0.2\textwidth}{!}{
    \begin{subfigure}[b]{0.3\textwidth}
        \centering
        \begin{quantikz}
        \lstick{0}&\gate{X}&\ctrl{1}&\\
        \lstick{1}&\ctrl{1}&\targ{}&\\
        \lstick{2}&\targ{}&\gate{Y}&\\
        \lstick{3}&\gate{Y}&\gate{X}&
        \end{quantikz}
    \end{subfigure}}
    \hfill
    \resizebox{0.18\textwidth}{!}{
    \begin{subfigure}[b]{0.25\textwidth}
        \centering
        \begin{quantikz}
        \lstick{0}&\gate[3]{\text{Toffoli}}&\\
        \lstick{1}&&\\
        \lstick{2}&&\\
        \lstick{3}&&
        \end{quantikz}
    \end{subfigure}}
    \hfill
    \resizebox{0.18\textwidth}{!}{
    \begin{subfigure}[b]{0.25\textwidth}
        \centering
        \begin{quantikz}
        \lstick{0}&\gate[4]{\text{QFT}}&\\
        \lstick{1}&&\\
        \lstick{2}&&\\
        \lstick{3}&&
        \end{quantikz}
    \end{subfigure}}
\label{fig:N_4Circuits}
\end{figure}
\subsection{List of tested circuits for n = 5, 6}~
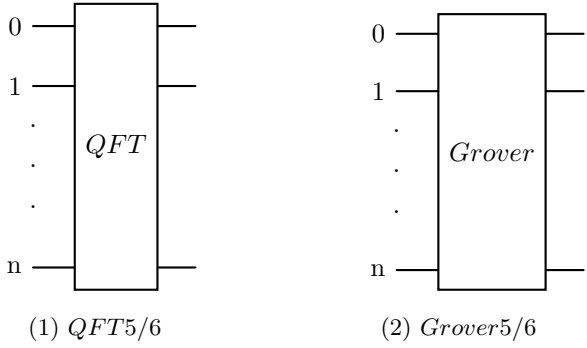
\begin{figure}[H]
    \begin{subfigure}[b]{0.35\textwidth}
        \centering
        \begin{quantikz}
        \lstick{0}&\gate[5]{\text{QFT}}&\\
        \lstick{1}&&\\
        \lstick{2}&&\\
        \lstick{\vdots}&&\\
        \lstick{n}&&
        \end{quantikz}
    \end{subfigure}
    \begin{subfigure}[b]{0.35\textwidth}
        \centering
        \begin{quantikz}
        \lstick{0}&\gate[5]{\text{Grover}}&\\
        \lstick{1}&&\\
        \lstick{2}&&\\
        \lstick{\vdots}&&\\
        \lstick{n}&&
        \end{quantikz}
    \end{subfigure}
\label{fig:N_5/6Circuits}
\end{figure}

\section{Example of resulting matrices}
\label{a_exampleMatrices}
\subsection{QFT2}

Ideal QFT$_2$:\\
\[
\begin{bmatrix}
    0.5 + 0j & 0.5 + 0j & 0.5 + 0j & 0.5 + 0j \\
    0.5 + 0j & 0 + 0.5j & -0.5 + 0j & 0 - 0.5j \\
    0.5 + 0j & -0.5 + 0j & 0.5 + 0j & -0.5 + 0j \\
    0.5 + 0j & 0 - 0.5j & -0.5 + 0j & 0 + 0.5j
\end{bmatrix}\]\\
Approximated QFT$_2$ with Frobenius loss:\\\\
\resizebox{1\textwidth}{!}{\(
\begin{bmatrix}
    0.5001 + 0.00124j & 0.50034 + 0.00034j & 0.49992 - 0.00065j & 0.49964 + 0.00004j \\
    0.49992 + 0.00058j & 0.00006 + 0.49964j & -0.5001 + 0.00131j & -0.00036 - 0.50034j \\
    0.49999 + 0.00211j & -0.49989 - 0.00065j & 0.49999 - 0.00071j & -0.50012 - 0.00053j \\
    0.49999 + 0.00074j & -0.00046 - 0.50012j & -0.49999 + 0.00209j & 0.00058 + 0.49989j
\end{bmatrix}\)}\vspace{0.5cm}
Approximated QFT$_2$ with Fidelity loss:\\
\[
\begin{bmatrix}
    0.5 + 10^{-5}j & 0.5 + 0j & 0.5 + 0j & 0.5 + 10^{-5}j \\
    0.50001 + 0j & -10^{-5} + 0.5j & -0.49999 + 0j & -0 - 0.5j \\
    0.49999 + 0j & -0.5 - 10^{-5}j & 0.5 + 0j & -0.5 + 10^{-5}j \\
    0.5 + 0j & -0 - 0.5j & -0.50001 + 0j & -0 + 0.5j
\end{bmatrix}\]\\
\vspace{0.5cm}
Approximated QFT$_2$ using Trace Distance loss:\\
\resizebox{1\textwidth}{!}{\(
\begin{bmatrix}
    0.49956 + 0.00033j & 0.49994 + 0.00015j & 0.50081 - 0.00151j & 0.49968 - 0.00084j \\
    0.50113 + 0.00061j & -0.00076 + 0.50015j & -0.4997 - 0.00017j & -0.00041 - 0.49902j \\
    0.49954 - 0.00048j & -0.49995 - 0.00091j & 0.50009 - 0.00025j & -0.50042 + 0.00115j \\
    0.49976 + 0.00112j & 0.00125 - 0.49995j & -0.4994 + 0.00005j & 0.00061 + 0.50088j
\end{bmatrix}\)}

\subsection{Frobenius loss random matrix}

Ideal random matrix:\\\\
\resizebox{1\textwidth}{!}{\(
\begin{bmatrix}
    -0.29584 + 0.229j & -0.04615 + 0.15104j & -0.17965 - 0.05721j & 0.88599 + 0.1207j \\
    0.60385 - 0.21653j & -0.1687 - 0.347j & 0.49773 + 0.1392j & 0.40585 + 0.08832j \\
    0.08131 - 0.03817j & -0.77354 - 0.04134j & -0.46739 + 0.41132j & -0.06164 - 0.02057j \\
    -0.56843 + 0.34255j & -0.29486 - 0.37306j & 0.53784 + 0.14295j & -0.09534 - 0.12097j
\end{bmatrix}\)}\vspace{0.5cm}
Approximated random matrix:\\\\
\resizebox{1\textwidth}{!}{\(
\begin{bmatrix}
    -0.29592 + 0.22861j & -0.04634 + 0.15084j & -0.17981 - 0.05736j & 0.8859 + 0.12175j \\
    0.60394 - 0.21609j & -0.16906 - 0.34673j & 0.49789 + 0.13962j & 0.40561 + 0.08859j \\
    0.0818 - 0.03893j & -0.7735 - 0.03856j & -0.46708 + 0.41182j & -0.06188 - 0.02029j \\
    -0.56897 + 0.34176j & -0.29565 - 0.37295j & 0.53712 + 0.144j & -0.09522 - 0.12115j
\end{bmatrix}\)}

\subsection{Fidelity loss random matrix}

Ideal random matrix:\\\\
\resizebox{1\textwidth}{!}{$
\begin{bmatrix}
    -0.16928 - 0.26816j & 0.76479 - 0.14228j & -0.07799 + 0.53619j & -0.02627 + 0.00328j \\
    -0.01908 + 0.21297j & -0.04915 - 0.38823j & -0.65553 - 0.00019j & 0.58473 + 0.17177j \\
    -0.07927 + 0.06076j & 0.33226 - 0.12021j & -0.03272 - 0.4948j & 0.11523 - 0.77846j \\
    -0.34491 - 0.85134j & -0.32301 - 0.11193j & -0.12928 - 0.11863j & 0.0702 - 0.06068j
\end{bmatrix}$}\vspace{0.5cm}
Approximated random matrix:\\\\
\resizebox{1\textwidth}{!}{$
\begin{bmatrix}
    -0.2166 - 0.20012j & 0.77381 - 0.12037j & -0.09693 + 0.53626j & -0.05025 - 0.0164j \\
    -0.11102 + 0.17193j & -0.08947 - 0.37078j & -0.66499 + 0.00496j & 0.58424 + 0.17046j \\
    -0.08118 + 0.07857j & 0.31621 - 0.1384j & -0.02304 - 0.50683j & 0.12519 - 0.77137j \\
    -0.3595 - 0.8539j & -0.32064 - 0.1389j & -0.04184 - 0.04047j & 0.03412 - 0.12229j
\end{bmatrix}$}

\subsection{Trace distance random matrix}

Ideal random matrix:\\\\
\resizebox{1\textwidth}{!}{$
\begin{bmatrix}
    -0.54567 - 0.69544j & 0.00628 + 0.03428j & 0.33735 - 0.26454j & 0.05243 - 0.17569j \\
    -0.43366 - 0.02622j & -0.03354 - 0.02404j & 0.04403 + 0.8673j & -0.15501 + 0.17713j \\
    -0.15789 + 0.01538j & -0.42342 - 0.18806j & -0.21585 + 0.02265j & 0.82598 + 0.17559j \\
    0.06695 - 0.01476j & -0.8375 - 0.28464j & 0.0952 - 0.07683j & -0.43968 - 0.06774j
\end{bmatrix}$}\vspace{0.5cm}
Approximated random matrix:\\\\
\resizebox{1\textwidth}{!}{$
\begin{bmatrix}
    -0.54806 - 0.69461j & 0.02581 + 0.02394j & 0.33803 - 0.26225j & 0.05547 - 0.17259j \\
    -0.43273 - 0.02469j & -0.01735 - 0.02715j & 0.04159 + 0.86756j & -0.1555 + 0.18033j \\
    -0.15529 + 0.00533j & -0.42411 - 0.1861j & -0.21809 + 0.03073j & 0.82615 + 0.17414j \\
    0.06245 - 0.03515j & -0.83762 - 0.28567j & 0.09951 - 0.06444j & -0.43896 - 0.07j
\end{bmatrix}$} 

\section*{Acknowledgments}
We thank Rohit S. Sarkar and Bibhas Adhikari for helpful discussions regarding the mathematical aspects of their algorithm on the SRBB. Giacomo Belli and Michele Amoretti acknowledge financial support from the European Union - NextGenerationEU, PNRR MUR project PE0000023-NQSTI. This research benefits from the High Performance Computing facility of the University of Parma, Italy (HPC.unipr.it) and also from IBM Quantum Credits awarded to Michele Amoretti - project Crosstalk-aware Quantum Multiprogramming.

\bibliographystyle{plainnat}
\bibliography{biblio.bib}

\begin{thebibliography}{48}
\providecommand{\natexlab}[1]{#1}
\providecommand{\url}[1]{\texttt{#1}}
\expandafter\ifx\csname urlstyle\endcsname\relax
  \providecommand{\doi}[1]{doi: #1}\else
  \providecommand{\doi}{doi: \begingroup \urlstyle{rm}\Url}\fi

\bibitem[Ashhab et~al.(2022)Ashhab, Yamamoto, Yoshihara, and Semba]{ashhab2022numerical}
Sahel Ashhab, Naoki Yamamoto, Fumiki Yoshihara, and Kouichi Semba.
\newblock Numerical analysis of quantum circuits for state preparation and unitary operator synthesis.
\newblock \emph{Physical Review A}, 106\penalty0 (2):\penalty0 022426, 2022.
\newblock \doi{10.1103/physreva.106.022426}.

\bibitem[Barenco et~al.(1995)Barenco, Bennett, Cleve, DiVincenzo, Margolus, Shor, Sleator, Smolin, and Weinfurter]{barenco1995elementary}
Adriano Barenco, Charles~H Bennett, Richard Cleve, David~P DiVincenzo, Norman Margolus, Peter Shor, Tycho Sleator, John~A Smolin, and Harald Weinfurter.
\newblock Elementary gates for quantum computation.
\newblock \emph{Physical review A}, 52\penalty0 (5):\penalty0 3457, 1995.
\newblock \doi{10.1103/physreva.52.3457}.

\bibitem[Belli et~al.(2024)Belli, Mordacci, and Amoretti]{10821064}
Giacomo Belli, Marco Mordacci, and Michele Amoretti.
\newblock A scalable quantum neural network for approximate unitary synthesis.
\newblock In \emph{2024 IEEE International Conference on Quantum Computing and Engineering (QCE)}, volume~02, pages 49--54, 2024.
\newblock \doi{10.1109/QCE60285.2024.10251}.

\bibitem[Belli et~al.(2025{\natexlab{a}})Belli, Mordacci, and Amoretti]{belli2025srbb}
Giacomo Belli, Marco Mordacci, and Michele Amoretti.
\newblock Srbb-based quantum state preparation.
\newblock \emph{arXiv preprint arXiv:2503.13647}, 2025{\natexlab{a}}.
\newblock \doi{10.48550/arXiv.2503.13647}.

\bibitem[Belli et~al.(2025{\natexlab{b}})Belli, Mordacci, and Amoretti]{srbb-syn}
Giacomo Belli, Marco Mordacci, and Michele Amoretti.
\newblock {Approximate Gate Synthesis with the Standard Recursive Block Basis decomposition (srbb-synthesis)}.
\newblock https://github.com/qslab-unipr/srbb-synthesis, 2025{\natexlab{b}}.

\bibitem[Benedetti et~al.(2019)Benedetti, Lloyd, Sack, and Fiorentini]{benedetti2019parameterized}
Marcello Benedetti, Erika Lloyd, Stefan Sack, and Mattia Fiorentini.
\newblock Parameterized quantum circuits as machine learning models.
\newblock \emph{Quantum Science and Technology}, 4\penalty0 (4):\penalty0 043001, 2019.
\newblock \doi{10.1088/2058-9565/ab4eb5}.

\bibitem[Bergholm et~al.(2005)Bergholm, Vartiainen, M{\"o}tt{\"o}nen, and Salomaa]{bergholm2005quantum}
Ville Bergholm, Juha~J Vartiainen, Mikko M{\"o}tt{\"o}nen, and Martti~M Salomaa.
\newblock Quantum circuits with uniformly controlled one-qubit gates.
\newblock \emph{Physical Review A—Atomic, Molecular, and Optical Physics}, 71\penalty0 (5):\penalty0 052330, 2005.
\newblock \doi{10.1103/physreva.71.052330}.

\bibitem[Bergholm et~al.(2018)Bergholm, Izaac, Schuld, Gogolin, Ahmed, Ajith, Alam, Alonso-Linaje, AkashNarayanan, Asadi, et~al.]{bergholm2018pennylane}
Ville Bergholm, Josh Izaac, Maria Schuld, Christian Gogolin, Shahnawaz Ahmed, Vishnu Ajith, M~Sohaib Alam, Guillermo Alonso-Linaje, B~AkashNarayanan, Ali Asadi, et~al.
\newblock Pennylane: Automatic differentiation of hybrid quantum-classical computations.
\newblock \emph{arXiv preprint arXiv:1811.04968}, 2018.
\newblock \doi{10.48550/arXiv.1811.04968}.

\bibitem[Bilek and Wold(2022)]{bilek2022recursive}
Stian Bilek and Kristian Wold.
\newblock Recursive variational quantum compiling.
\newblock \emph{arXiv preprint arXiv:2203.08514}, 2022.
\newblock \doi{10.48550/arXiv.2203.08514}.

\bibitem[Bullock and Markov(2004)]{bullock2004asymptotically}
Stephen~S Bullock and Igor~L Markov.
\newblock Asymptotically optimal circuits for arbitrary n-qubit diagonal comutations.
\newblock \emph{Quantum Inf. Comput.}, 4\penalty0 (1):\penalty0 27--47, 2004.
\newblock \doi{10.26421/qic4.1-3}.

\bibitem[Cerezo et~al.(2021)Cerezo, Arrasmith, Babbush, Benjamin, Endo, Fujii, McClean, Mitarai, Yuan, Cincio, et~al.]{cerezo2021variational}
Marco Cerezo, Andrew Arrasmith, Ryan Babbush, Simon~C Benjamin, Suguru Endo, Keisuke Fujii, Jarrod~R McClean, Kosuke Mitarai, Xiao Yuan, Lukasz Cincio, et~al.
\newblock Variational quantum algorithms.
\newblock \emph{Nature Reviews Physics}, 3\penalty0 (9):\penalty0 625--644, 2021.
\newblock \doi{10.1038/s42254-021-00348-9}.

\bibitem[Dawson and Nielsen(2006)]{dawson2006solovay}
Christopher~M Dawson and Michael~A Nielsen.
\newblock The solovay-kitaev algorithm.
\newblock \emph{Quantum Information \& Computation}, 6\penalty0 (1):\penalty0 81--95, 2006.
\newblock \doi{10.26421/qic6.1-6}.

\bibitem[Dunjko and Wittek(2020)]{dunjko2020non}
Vedran Dunjko and Peter Wittek.
\newblock A non-review of quantum machine learning: trends and explorations.
\newblock \emph{Quantum Views}, 4:\penalty0 32, 2020.
\newblock \doi{10.22331/qv-2020-03-17-32}.

\bibitem[Jiang et~al.(2020)Jiang, Sun, Teng, Wu, Wu, and Zhang]{jiang2020optimal}
Jiaqing Jiang, Xiaoming Sun, Shang-Hua Teng, Bujiao Wu, Kewen Wu, and Jialin Zhang.
\newblock Optimal space-depth trade-off of cnot circuits in quantum logic synthesis.
\newblock In \emph{Proceedings of the Fourteenth Annual ACM-SIAM Symposium on Discrete Algorithms}, pages 213--229. SIAM, 2020.
\newblock \doi{10.1137/1.9781611975994.13}.

\bibitem[Kang and Ma(2023)]{kang2023cnot}
Yifan Kang and Henry Ma.
\newblock \emph{CNOT-Optimal Circuit Synthesis}.
\newblock Phd thesis, Massachusetts Institute of Technology, 2023.

\bibitem[Kirillov(2008)]{kirillov2008introduction}
Alexander~A Kirillov.
\newblock \emph{An introduction to Lie groups and Lie algebras}, volume 113.
\newblock Cambridge University Press, 2008.
\newblock \doi{10.1017/cbo9780511755156}.

\bibitem[Kliuchnikov et~al.(2013{\natexlab{a}})Kliuchnikov, Maslov, and Mosca]{kliuchnikov2013asymptotically}
Vadym Kliuchnikov, Dmitri Maslov, and Michele Mosca.
\newblock Asymptotically optimal approximation of single qubit unitaries by clifford and t circuits using a constant number of ancillary qubits.
\newblock \emph{Physical review letters}, 110\penalty0 (19):\penalty0 190502, 2013{\natexlab{a}}.
\newblock \doi{10.1103/physrevlett.110.190502}.

\bibitem[Kliuchnikov et~al.(2013{\natexlab{b}})Kliuchnikov, Maslov, and Mosca]{kliuchnikov2013fast}
Vadym Kliuchnikov, Dmitri Maslov, and Michele Mosca.
\newblock Fast and efficient exact synthesis of single-qubit unitaries generated by clifford and t gates.
\newblock \emph{Quantum Information \& Computation}, 13\penalty0 (7-8):\penalty0 607--630, 2013{\natexlab{b}}.
\newblock \doi{10.26421/qic13.7-8-4}.

\bibitem[Kliuchnikov et~al.(2015{\natexlab{a}})Kliuchnikov, Bocharov, Roetteler, and Yard]{kliuchnikov2015framework}
Vadym Kliuchnikov, Alex Bocharov, Martin Roetteler, and Jon Yard.
\newblock A framework for approximating qubit unitaries.
\newblock \emph{arXiv preprint arXiv:1510.03888}, 2015{\natexlab{a}}.
\newblock \doi{10.48550/arXiv.1510.03888}.

\bibitem[Kliuchnikov et~al.(2015{\natexlab{b}})Kliuchnikov, Maslov, and Mosca]{kliuchnikov2015practical}
Vadym Kliuchnikov, Dmitri Maslov, and Michele Mosca.
\newblock Practical approximation of single-qubit unitaries by single-qubit quantum clifford and t circuits.
\newblock \emph{IEEE Transactions on Computers}, 65\penalty0 (1):\penalty0 161--172, 2015{\natexlab{b}}.
\newblock \doi{10.1109/tc.2015.2409842}.

\bibitem[Krol et~al.(2022)Krol, Sarkar, Ashraf, Al-Ars, and Bertels]{krol2022efficient}
Anna~M Krol, Aritra Sarkar, Imran Ashraf, Zaid Al-Ars, and Koen Bertels.
\newblock Efficient decomposition of unitary matrices in quantum circuit compilers.
\newblock \emph{Applied Sciences}, 12\penalty0 (2):\penalty0 759, 2022.
\newblock \doi{10.3390/app12020759}.

\bibitem[Madden and Simonetto(2022)]{madden2022best}
Liam Madden and Andrea Simonetto.
\newblock Best approximate quantum compiling problems.
\newblock \emph{ACM Transactions on Quantum Computing}, 3\penalty0 (2):\penalty0 1--29, 2022.
\newblock \doi{10.1145/3505181}.

\bibitem[Madden et~al.(2022)Madden, Akhriev, and Simonetto]{madden2022sketching}
Liam Madden, Albert Akhriev, and Andrea Simonetto.
\newblock Sketching the best approximate quantum compiling problem.
\newblock In \emph{2022 IEEE International Conference on Quantum Computing and Engineering (QCE)}, pages 492--502. IEEE, 2022.
\newblock \doi{10.1109/qce53715.2022.00071}.

\bibitem[Mottonen and Vartiainen(2005)]{mottonen2005decompositions}
M~Mottonen and JJ~Vartiainen.
\newblock Decompositions of general quantum gates.
\newblock \emph{arXiv preprint quant-ph/0504100}, 2005.
\newblock \doi{10.48550/arXiv.quant-ph/0504100}.

\bibitem[M{\"o}tt{\"o}nen et~al.(2004)M{\"o}tt{\"o}nen, Vartiainen, Bergholm, and Salomaa]{mottonen2004quantum}
Mikko M{\"o}tt{\"o}nen, Juha~J Vartiainen, Ville Bergholm, and Martti~M Salomaa.
\newblock Quantum circuits for general multiqubit gates.
\newblock \emph{Physical review letters}, 93\penalty0 (13):\penalty0 130502, 2004.
\newblock \doi{10.1103/physrevlett.93.130502}.

\bibitem[Nakahara(2018)]{nakahara2018geometry}
Mikio Nakahara.
\newblock \emph{Geometry, topology and physics}.
\newblock CRC press, 2018.
\newblock \doi{10.1201/9781315275826}.

\bibitem[Pham et~al.(2013)Pham, Van~Meter, and Horsman]{pham2013optimization}
Tien~Trung Pham, Rodney Van~Meter, and Dominic Horsman.
\newblock Optimization of the solovay-kitaev algorithm.
\newblock \emph{Physical Review A—Atomic, Molecular, and Optical Physics}, 87\penalty0 (5):\penalty0 052332, 2013.
\newblock \doi{10.1103/physreva.87.052332}.

\bibitem[Quantum(2024{\natexlab{a}})]{ibm_quantum_brisbane}
IBM Quantum.
\newblock Ibm brisbane quantum system.
\newblock \url{https://quantum.ibm.com/services/resources?system=ibm_brisbane}, 2024{\natexlab{a}}.
\newblock Accessed: 2024-11-14.

\bibitem[Quantum(2024{\natexlab{b}})]{ibm_quantum_fez}
IBM Quantum.
\newblock Ibm fez quantum system.
\newblock \url{https://quantum.ibm.com/services/resources?system=ibm_fez}, 2024{\natexlab{b}}.
\newblock Accessed: 2024-11-14.

\bibitem[Regula et~al.(2021)Regula, Takagi, and Gu]{regula2021operational}
Bartosz Regula, Ryuji Takagi, and Mile Gu.
\newblock Operational applications of the diamond norm and related measures in quantifying the non-physicality of quantum maps.
\newblock \emph{Quantum}, 5:\penalty0 522, 2021.
\newblock \doi{10.22331/q-2021-08-09-522}.

\bibitem[Ross(2015{\natexlab{a}})]{Ross2015OptimalAC}
Neil~J. Ross.
\newblock Optimal ancilla-free clifford+v approximation of z-rotations.
\newblock \emph{Quantum Inf. Comput.}, 15:\penalty0 932--950, 2015{\natexlab{a}}.
\newblock \doi{10.26421/qic15.11-12-4}.

\bibitem[Ross(2015{\natexlab{b}})]{ross2015algebraic}
Neil~J Ross.
\newblock Algebraic and logical methods in quantum computation.
\newblock \emph{arXiv preprint arXiv:1510.02198}, 2015{\natexlab{b}}.
\newblock \doi{10.48550/arXiv.1510.02198}.

\bibitem[Ross and Selinger(2016)]{ross2016optimal}
Neil~J Ross and Peter Selinger.
\newblock Optimal ancilla-free clifford+ t approximation of z-rotations.
\newblock \emph{Quantum Inf. Comput.}, 16\penalty0 (11\&12):\penalty0 901--953, 2016.
\newblock \doi{10.26421/qic16.11-12-1}.

\bibitem[Sarkar(2024)]{Sarkar:2024bax}
Rohit~Sarma Sarkar.
\newblock \emph{Scalable Quantum Circuit Representation of Unitary Matrices}.
\newblock Phd thesis, Indian Institute of Technology Kharagpur, Department of Mathematics, India, 2024.

\bibitem[Sarkar and Adhikari(2023)]{sarkar2023scalable}
Rohit~Sarma Sarkar and Bibhas Adhikari.
\newblock Scalable quantum circuits for n-qubit unitary matrices.
\newblock In \emph{2023 IEEE International Conference on Quantum Computing and Engineering (QCE)}, volume~1, pages 1078--1088. IEEE, 2023.
\newblock \doi{10.1109/qce57702.2023.00122}.

\bibitem[Sarkar and Adhikari(2024)]{sarkar2024quantum}
Rohit~Sarma Sarkar and Bibhas Adhikari.
\newblock A quantum neural network framework for scalable quantum circuit approximation of unitary matrices.
\newblock \emph{arXiv preprint arXiv:2405.00012}, 2024.
\newblock \doi{10.48550/arXiv.2405.00012}.

\bibitem[Schuld and Petruccione(2018)]{schuld2018supervised}
Maria Schuld and Francesco Petruccione.
\newblock \emph{Supervised learning with quantum computers}, volume~17.
\newblock Springer, 2018.
\newblock \doi{10.1007/978-3-319-96424-9}.

\bibitem[Selinger(2012)]{Selinger2012EfficientCA}
Peter Selinger.
\newblock Efficient clifford+t approximation of single-qubit operators.
\newblock \emph{Quantum Inf. Comput.}, 15:\penalty0 159--180, 2012.
\newblock \doi{10.26421/qic15.1-2-10}.

\bibitem[Sharma et~al.(2020)Sharma, Khatri, Cerezo, and Coles]{sharma2020noise}
Kunal Sharma, Sumeet Khatri, Marco Cerezo, and Patrick~J Coles.
\newblock Noise resilience of variational quantum compiling.
\newblock \emph{New Journal of Physics}, 22\penalty0 (4):\penalty0 043006, 2020.
\newblock \doi{10.1088/1367-2630/ab784c}.

\bibitem[Shende et~al.(2004)Shende, Markov, and Bullock]{shende2004minimal}
Vivek~V Shende, Igor~L Markov, and Stephen~S Bullock.
\newblock Minimal universal two-qubit controlled-not-based circuits.
\newblock \emph{Physical Review A—Atomic, Molecular, and Optical Physics}, 69\penalty0 (6):\penalty0 062321, 2004.
\newblock \doi{10.1103/physreva.69.062321}.

\bibitem[Shende et~al.(2005)Shende, Bullock, and Markov]{shende2005synthesis}
Vivek~V Shende, Stephen~S Bullock, and Igor~L Markov.
\newblock Synthesis of quantum logic circuits.
\newblock In \emph{Proceedings of the 2005 Asia and South Pacific Design Automation Conference}, pages 272--275, 2005.
\newblock \doi{10.1109/aspdac.2005.1466172}.

\bibitem[Vartiainen et~al.(2004)Vartiainen, M{\"o}tt{\"o}nen, and Salomaa]{vartiainen2004efficient}
Juha~J Vartiainen, Mikko M{\"o}tt{\"o}nen, and Martti~M Salomaa.
\newblock Efficient decomposition of quantum gates.
\newblock \emph{Physical review letters}, 92\penalty0 (17):\penalty0 177902, 2004.
\newblock \doi{10.1103/physrevlett.92.177902}.

\bibitem[Vatan and Williams(2004)]{vatan2004optimal}
Farrokh Vatan and Colin Williams.
\newblock Optimal quantum circuits for general two-qubit gates.
\newblock \emph{Physical Review A—Atomic, Molecular, and Optical Physics}, 69\penalty0 (3):\penalty0 032315, 2004.
\newblock \doi{10.1103/physreva.69.032315}.

\bibitem[Vidal and Dawson(2004)]{vidal2004universal}
Guifre Vidal and Christopher~M Dawson.
\newblock Universal quantum circuit for two-qubit transformations with three controlled-not gates.
\newblock \emph{Physical Review A}, 69\penalty0 (1):\penalty0 010301, 2004.
\newblock \doi{10.1103/physreva.69.010301}.

\bibitem[Watrous(2012)]{watrous2012simpler}
John Watrous.
\newblock Simpler semidefinite programs for completely bounded norms.
\newblock \emph{arXiv preprint arXiv:1207.5726}, 2012.
\newblock \doi{10.48550/arXiv.1207.5726}.

\bibitem[Younis et~al.(2020)Younis, Sen, Yelick, and Iancu]{younis2020qfast}
Ed~Younis, Koushik Sen, Katherine Yelick, and Costin Iancu.
\newblock Qfast: Quantum synthesis using a hierarchical continuous circuit space.
\newblock \emph{arXiv preprint arXiv:2003.04462}, 2020.
\newblock \doi{10.48550/arXiv.2003.04462}.

\bibitem[Younis et~al.(2021)Younis, Sen, Yelick, and Iancu]{younis2021qfast}
Ed~Younis, Koushik Sen, Katherine Yelick, and Costin Iancu.
\newblock Qfast: Conflating search and numerical optimization for scalable quantum circuit synthesis.
\newblock In \emph{2021 IEEE International Conference on Quantum Computing and Engineering (QCE)}, pages 232--243. IEEE, 2021.
\newblock \doi{10.1109/qce52317.2021.00041}.

\bibitem[Zhiyenbayev et~al.(2018)Zhiyenbayev, Akulin, and Mandilara]{zhiyenbayev2018quantum}
Y~Zhiyenbayev, VM~Akulin, and Aikaterini Mandilara.
\newblock Quantum compiling with diffusive sets of gates.
\newblock \emph{Physical Review A}, 98\penalty0 (1):\penalty0 012325, 2018.
\newblock \doi{10.1103/physreva.98.012325}.

\end{thebibliography}

\end{document}